\documentclass[reqno,12pt]{amsart}

\date{July 8, 2020}

\usepackage[margin=1in]{geometry}
\usepackage{xspace}
\usepackage[utf8]{inputenc}
\usepackage{graphicx}
\usepackage{amssymb}
\usepackage{mathrsfs}
\usepackage[active]{srcltx}
\usepackage{url}
\usepackage{hyperref}
\usepackage{amsmath,amstext,amsxtra,amsgen,amsbsy,amsopn,amscd,amsthm,amsfonts}
\usepackage{latexsym}
\usepackage{dsfont}
\usepackage{enumitem}

\newtheorem{proposition}{Proposition}
\newtheorem{theorem}[proposition]{Theorem}
\newtheorem{lemma}[proposition]{Lemma}
\newtheorem{corollary}[proposition]{Corollary}

\theoremstyle{definition}

\theoremstyle{remark}

\newtheorem{remark}[proposition]{Remark}

\numberwithin{equation}{section}
\numberwithin{proposition}{section}

\newcommand\R{{\ensuremath {\mathbb R} }}
\newcommand\C{{\ensuremath {\mathbb C} }}

\newcommand\Z{{\ensuremath {\mathbb Z} }}

\renewcommand\phi{\varphi}

\renewcommand\le{\leqslant}
\renewcommand\ge{\geqslant}
\renewcommand\epsilon{\varepsilon}

\renewcommand\tilde{\widetilde}

\newcommand\ii{{\ensuremath {\infty}}}

\newcommand{\cK}{\mathcal{K}}
\newcommand{\cM}{\mathcal{M}}

\newcommand{\cO}{\mathcal{O}}

\newcommand{\cU}{\mathcal{U}}

\newcommand\1{{\ensuremath {\mathds 1} }}
\newcommand{\Sph}{\mathbb{S}}

\DeclareMathOperator{\dist}{dist}

\DeclareMathOperator{\im}{Im}

\DeclareMathOperator{\re}{Re}
\DeclareMathOperator{\spec}{spec}
\DeclareMathOperator{\supp}{supp}

\DeclareMathOperator{\Tr}{Tr}

\DeclareMathOperator{\vol}{Vol}

\newcommand{\dps}{\displaystyle}
\newcommand{\loc}{{\rm loc}}

\newcommand\nnfootnote[1]{%
  \begin{NoHyper}
  \renewcommand\thefootnote{}\footnote{#1}%
  \addtocounter{footnote}{-1}%
  \end{NoHyper}
}

\title{Sharp Weyl laws with singular potentials}

\author[R. L. Frank]{Rupert L. Frank}
\address[R. L. Frank]{Mathematisches Institut, Ludwig-Maximilians-Universit\"at M\"unchen, Theresienstr. 39, 80333 M\"unchen, Germany, and Mathematics 253-37, Caltech, Pasa\-de\-na, CA 91125, USA}
\email{r.frank@lmu.de, rlfrank@caltech.edu}

\author[J. Sabin]{Julien Sabin}
\address[J. Sabin]{Laboratoire de Math\'ematiques d'Orsay, Univ. Paris-Sud, CNRS, Universit\'e Paris-Saclay, 91405 Orsay, France} 
\email{julien.sabin@universite-paris-saclay.fr}

\begin{document}

\begin{abstract}
	We consider the Laplace--Beltrami operator on a three-dimensional Riemannian manifold perturbed by a potential from the Kato class and study whether various forms of Weyl's law remain valid under this perturbation. We show that a pointwise Weyl law holds, modified by an additional term, for any Kato class potential with the standard sharp remainder term. The additional term is always of lower order than the leading term, but it may or may not be of lower order than the sharp remainder term. In particular, we provide examples of singular potentials for which this additional term violates the sharp pointwise Weyl law of the standard Laplace-Beltrami operator. For the proof we extend the method of Avakumovi\'c to the case of Schr\"odinger operators with singular potentials.
\end{abstract}

\maketitle

\nnfootnote{\copyright\, 2020 by
the authors. This paper may be reproduced, in its entirety, for
non-commercial purposes.}

\section{Introduction and main results}

In this paper we investigate the question to what extent the Weyl law in its various forms for the Laplace--Beltrami operator on a compact manifold remains valid when a singular potential is added.

To be more precise, let $(M,g)$ be a three-dimensional compact Riemannian manifold without boundary. We comment on the three-dimensionality assumption in Remark \ref{rk:3d} below. Our results are valid under not very restrictive regularity assumptions on $M$, namely $C^4$-smoothness, as discussed in Remark \ref{rk:regularity-g}, but they are also new in the $C^\infty$ setting.

We denote by $-\Delta_g$ the (nonnegative) Laplace--Beltrami operator on $(M,g)$. Its spectrum is discrete and we denote by $\1(-\Delta_g\le t)(x,y)$ the integral kernel of its spectral projection $\1(-\Delta_g\le t)$ corresponding to eigenvalues $\le t$. Moreover,
$$
N(t,-\Delta_g) = \Tr\1(-\Delta_g\le t) = \int_M \1(-\Delta_g\le t)(x,x)\,dv_g(x)
$$
denotes the number of its eigenvalues $\le t$, counting multiplicities.

We are interested in four different forms of Weyl's law concerning the limit $t\to\infty$, namely

\begin{tabular}{ll}
 $\bullet$ integrated Weyl law: & $\dps N(t,-\Delta_g) = \frac{t^{3/2}}{6\pi^2}\,\vol_g(M) + o(t^{3/2})$ \\
 $\bullet$ pointwise Weyl law:& $\dps \1(-\Delta_g\le t)(x,x) = \frac{t^{3/2}}{6\pi^2} + o(t^{3/2}) \ \text{uniformly in}\ x\in M$ \\
 $\bullet$ sharp integrated Weyl law: & $\dps N(t,-\Delta_g) = \frac{t^{3/2}}{6\pi^2}\,\vol_g(M) + \mathcal O(t)$ \\
 $\bullet$ sharp pointwise Weyl law: &  $\dps \1(-\Delta_g\le t)(x,x) = \frac{t^{3/2}}{6\pi^2} + \mathcal O(t)\ \text{uniformly in}\ x\in M$.
\end{tabular}

Clearly, the pointwise versions imply the integrated versions and the sharp versions imply the regular versions. The integrated Weyl law was originally proved by Weyl \cite{Weyl-1911} in Euclidean space and its pointwise version in Euclidean space is due to Carleman \cite{Carleman-34,Carleman-36}. The extension of the pointwise Weyl law to manifolds is due to Minakshisundaram and Pleijel \cite{MinPle-49}. The sharp pointwise Weyl law for closed manifolds is due to Avakumovi\'c \cite{Avakumovic-56}, Levitan \cite{Levitan-52}, and H\"ormander \cite{Hormander-68}. The adjective `sharp' refers to the fact that there are manifolds, for instance the round sphere or, more generally, Zoll manifolds, for which the remainder $\mathcal O(t)$ cannot be improved. On the other hand, there is a large literature on improved Weyl laws with $o(t)$ remainders for certain manifolds, both in integrated \cite{DuiGui-75,Landau-27,Hlawka-50,Walfisz-59,Heath-Brown-99,Huxley-03,Berard-77} and in pointwise \cite{Safarov-88,SogZel-02} form. These improved Weyl laws will not play a major role in this paper and we only comment on them in Remark \ref{rk:improved-remainder} below.

Let $V:M\to\R$ belong to the Kato class, that is, it is measurable and
$$
\lim_{\epsilon\to 0} \sup_{x\in M} \int_{d_g(x,y)<\epsilon}\frac{|V(y)|}{d_g(x,y)}\,dv_g(y) = 0 \,.
$$
Under this assumption $-\Delta_g+V$ can be defined via a quadratic form as a selfadjoint, bounded below operator with discrete spectrum. When $V$ is smooth, it is well-known \cite{Hormander-68} that all four versions of Weyl's law remain valid when $-\Delta_g$ is replaced by $-\Delta_g +V$. Our goal in this paper is to investigate to what extent this is true for general Kato class $V$. We will find both preservation and violation, depending on the specific form of Weyl's law. Let us summarize our main results.

\begin{theorem}[Pointwise Weyl law for Kato class potentials]\label{thm:weylptw}
	Let $V:M\to\R$ be in the Kato class. Then, as $t\to+\infty$ and uniformly in $x\in M$
	$$
	\1(-\Delta_g +V\le t)(x,x) = \frac{t^{3/2}}{6\pi^2} + o(t^{3/2}) \,.
	$$
\end{theorem}

On the other hand, we show that there is no $\epsilon>0$ such that the term $o(t^{3/2})$ in Theorem \ref{thm:weylptw} can be replaced by $\mathcal O( t^{(3-\epsilon)/2})$ for all Kato class $V$ and, in particular, we show that the sharp form of the pointwise Weyl law is \emph{not} valid for all Kato class potentials. In fact, in Proposition \ref{lem:example-potential} we will show that the sharp pointwise Weyl law fails for the potential
\begin{equation}\label{eq:counterexample}
V(x) = \gamma \,\frac{\chi(d_g(x,x_0))}{d_g(x,x_0)^{2-\eta}}
\end{equation}
with $\eta\in(0,1)$, $x_0\in M$, $\gamma\in\R\setminus\{0\}$ and $\chi$ a cut-off function which is $\equiv 1$ near zero. More precisely, we will see that there is an explicit singular term in $\1(-\Delta_g +V\le t)(x,x)$ that lives at scale $t^{-1/2}$ around $x_0$ and has size $t^{(3-\eta)/2}$. We refer to this proposition and Lemma~\ref{xifunction} for a detailed analysis of this example.

Next, we show that the sharp pointwise Weyl law remains valid for $V$ satisfying a mild additional regularity condition.

\begin{theorem}[Sharp pointwise Weyl law]\label{thm:weylptwsharp}
	Assume that $V:M\to\R$ satisfies for some $\epsilon'>0$,
	\begin{equation}
	\label{eq:weylptwsharpass}
	\sup_{x\in M} \int_{d_g(y,x)<\epsilon'} \frac{|V(y)|}{d_g(y,x)^2}\,dv_g(y) <\infty \,.
	\end{equation}
	Then, as $t\to+\infty$ and uniformly in $x\in M$,
	$$
	\1(-\Delta_g +V\le t)(x,x) = \frac{t^{3/2}}{6\pi^2} + \mathcal O(t) \,.
	$$
\end{theorem}

Note that the condition in this theorem is, in particular, satisfied if $V\in L^q(M)$ for some $q>3$. For comparison, for any $q<3$ there is an $\eta\in(0,1)$ such that \eqref{eq:counterexample} belongs to $L^q$, so $q=3$ is the threshold for the validity of the sharp pointwise Weyl law on the $L^q$ scale.

Finally, we discuss the integrated forms of Weyl's law. Clearly, Theorem \ref{thm:weylptw} yields this in the form of an $o(t^{3/2})$ remainder. Remarkably, however, despite the possible failure of the sharp pointwise Weyl law, the integrated form remains valid in the strong form for arbitrary Kato class potentials. This was recently shown by Huang and Sogge \cite{HuaSog-20}. Here we give an independent proof of their result and extend it to potentials $V$ that are sums of Kato class and $L^{3/2}(M)$ functions. Note that the Kato class and $L^{3/2}(M)$ share the same critical scaling behavior and that neither one is contained in the other one.

\begin{theorem}[Sharp integrated Weyl law]\label{thm:weylintsharp}
	Assume that $V:M\to\R$ is the sum of a function in the Kato class and a function in $L^{3/2}(M)$. Then, as $t\to+\infty$,
	\begin{equation}\label{eq:weyl-integrated-manifold}
	N(t,-\Delta_g+V)=\frac{t^{3/2}}{6\pi^2}+\cO(t) \,.
	\end{equation}
\end{theorem}

This concludes the summary of our main results. Let us now comment on the motivation for these results and on their method of proofs.

Potentials with singularities appear naturally in physics, most notably the Coulomb case $|x|^{-1}$ in three dimensions. Spectral asymptotics in this case, both in integrated \cite{Hughes-85,SieWei-87,SieWei-89,SolSpi-03} and in pointwise form \cite{FefSec-94,FefSec-95}, have been extensively studied. As we will explain below, however, our asymptotic regime is different and we find phenomena of another nature.

Our motivation to consider Kato class potentials is two-fold. First, it is natural to investigate the validity of Weyl's law under rather minimal assumptions on the potential $V$ for which the Schr\"odinger operator $-\Delta_g+V$ can be defined as selfadjoint, bounded below operator. Both membership to the Kato class and to $L^{3/2}(M)$ are almost optimal conditions in this respect. For instance, singularities of the type $|x|^{-\beta}$ for $\beta<2$ are allowed for both classes. Note that the exponent $\beta=2$ is critical with respect to the scaling of the Laplacian and if, in this case, the coefficient of the singularity is too large negative, then the corresponding quadratic form is not bounded from below. We mention in passing (see Remark \ref{katovs32}) that neither the Kato class is contained in $L^{3/2}$, nor the other way around.

Second, the property that makes the Kato class more natural for our analysis than $L^{3/2}$ is its important property that all eigenfunctions of $-\Delta_g+V$ with Kato class $V$ are bounded. In fact, this class is almost sharp with respect to this property, in the sense that the fact that $e^{-t(-\Delta_g+V)}$ maps $L^2$ to $L^\ii$ essentially implies that $V$ belongs to the Kato class \cite{AizSim-82}. In contrast, eigenfunctions for potentials in $L^{3/2}$ need not be bounded. Since boundedness of eigenfunctions is an obvious requirement for a pointwise Weyl law to hold, the Kato class appears naturally in this context.

When comparing our results with other results about Weyl's law with singular potentials, it is important to distinguish between the `large eigenvalue regime' $-\Delta_g+V\le t$, $t\to+\infty$, which we consider here, and the `semiclassical regime' $-h^2\Delta_g+V\le 0$, $h\to0$. The large eigenvalue regime corresponds to a semiclassical regime $-h^2\Delta_g -1 + h^2 V\le 0$, $h=t^{-1/2}$, with a small parameter in front of the nonconstant part of the potential. Effects of singularities have been studied in details in the semiclassical regime but, as far as we know, their study in the large eigenvalue regime has only begun very recently in \cite{BlaSirSog-19,HuaSog-20}. 

The analysis in the semiclassical regime is, in part, motivated by that of the Schr\"odinger operator $-\Delta - Z^{4/3} \phi(Z^{1/3}x)$ in $L^2(\R^3)$ with a potential satisfying $\phi(y) \sim |y|^{-1}$ as $y\to 0$, which appears in atomic phyiscs in connection with Thomas--Fermi theory. Via rescaling we see that it corresponds to the semiclassical regime with $h=Z^{-1/3}$. It is well known that the singularity of $\phi$ has an effect, called the Scott correction, on the semiclassical expansion of the Riesz means of order one (meaning that instead of counting eigenvalues, we sum the negative eigenvalues). On the other hand, this effect does not appear when looking at the number of eigenvalues less than a fixed energy level $E<0$ \cite{HelKnaSieWei-92}. This phenomenon was thoroughly studied in \cite{Sobolev-96} where it was shown, in particular, that singularities $|x|^{-\beta}$ with $\beta<2$ do not affect the local eigenvalue counting function to order $\mathcal O(h^{-d+1})$. We emphasize that the results \cite{HelKnaSieWei-92,Sobolev-96} involve (at least locally) integrated Weyl laws and that it is not clear to us whether these methods can make assertions about pointwise Weyl laws. Moreover, the specific powerlike form of the singularity seems to be important, whereas Kato class potentials can have much more complicated singularities. Violations of the main term in the Weyl law in the semi-classical regime have also been studied, see for instance \cite{Laptev-93,BirLap-96}.

Let us return to the large eigenvalue regime studied in this paper. One of our key findings, which is behind Theorems \ref{thm:weylptw}, \ref{thm:weylptwsharp} and \ref{thm:weylintsharp}, is that there is an \emph{additional term} in the pointwise Weyl law. Theorem \ref{thm:weylptw} says that this term is $o(t^{3/2})$ for all Kato class potentials, Theorem \ref{thm:weylptwsharp} says that it is $\mathcal O(t)$ under the additional assumption \eqref{eq:weylptwsharpass} and Theorem \ref{thm:weylintsharp} says that it is $\mathcal O(t)$ when integrated over $M$ for all Kato class potentials. On the other hand, for the potential \eqref{eq:counterexample} we show that this term contributes a constant times $t^{(3-\eta)/2}$ at the singularity.

We have not seen this additional term discussed in the literature before. This term is neither of semiclassical origin nor of spectral origin, in contrast to the term giving rise to the Scott correction \cite{Hughes-85,SieWei-87,SieWei-89}, for instance.

Besides drawing attention to this additional term in the pointwise Weyl law, a goal of this paper is to advertise the method that we use to prove such results. It is the method of Avakumovi\'c \cite{Avakumovic-56}, which was developed in the 1950's but seems to have been forgotten in the last decades. As we will demonstrate, this method, which was originally developed for the case $V=0$, is a powerful tool (at least in three dimensions) to get very precise results on Weyl laws and is well suited to extensions to singular potentials $V$. For comparison, it is not clear to us how to adapt the approaches of the related works \cite{Sobolev-96,HuaSog-20} to obtain pointwise Weyl laws. In particular, we hope that our result can be used as a benchmark for the proof of pointwise Weyl laws with singular potentials by other methods, in the sense that such other methods must reveal the additional terms violating the standard Weyl law that we uncover in case of a strong enough singularity.

In the remaining part of this introduction we now discuss Avakumovi\'c's method and our new additional ingredients.

Avakumovi\'c's method relies on two ingredients, namely Tauberian theorems and para\-metrix estimates. The Tauberian theorems used by Avakumovi\'c concern a higher order Stieltjes transform instead of the Fourier transform, as in the work of Levitan \cite{Levitan-52} and then in many works following \cite{Hormander-68}. Consequently, one will seek a parametrix for the resolvent (or its powers) for spectral parameters far away from the spectrum and not for the wave propagator.

In connection with the parametrix estimates, two of Avakumovi\'c's fundamental insights are the following. First, an exponentially small remainder in these estimates is acceptable in the Tauberian theorems. This observation immediately leads to the remainder bound $\mathcal O(t^{(N-1)/2}\dist(x,\partial\Omega)^{-1})$ in the case of an open set $\Omega\subset\R^N$ with Dirichlet boundary conditions; see \cite{Avakumovic-52}. The second key insight is that some error terms which are not exponentially small are themselves (higher order) Stieltjes transforms and therefore acceptable in a sufficiently precise version of the Tauberian theorems.

Our main task is to prove corresponding parametrix estimates in the case where $V$ is added. In fact, we will use the same parametrix for the Green's function as Avakumovi\'c does, so the task will be to bound the effect of $V$ on the difference between the true Green's function and the parametrix. Several terms in this difference can be estimated by modifying Avakumovi\'c's arguments. In fact, in the case of bounded potentials $V$, this has already been done by Bojani\'c in \cite{Bojanic-53}. When $V$ is unbounded, however, there are other terms which are not as small as the previous ones and which need to be treated separately. These are the terms that lead to the `additional term' in the pointwise Weyl law, mentioned above, and that may lead to a violation of the sharp pointwise Weyl law.

This analysis eventually leads to the decomposition
\begin{multline}
\label{eq:decompintro}
(-\Delta_g + V + \lambda)^{-2}(x,x) = \frac{1}{8\pi\sqrt\lambda} - \int_0^\infty \frac{t^{3/2}\, r_{0}^V(t,x) + R^V(t,x)}{(t+\lambda)^3}\,dt\\
+ \text{exponentially small as}\ \lambda\to+\ii.
\end{multline}
The term $R^V(t,x)$ here is of size $\mathcal O(t)$ and can be dealt with as in Avakumovi\'c's paper. The important new contribution is the term $t^{3/2} r_{0}^V(t,x)$, which is defined through an explicit, absolutely convergent power series and encompasses the singular contributions coming from $V$.  This is the term that ultimately give rise to violations of the sharp pointwise Weyl law. A careful analysis of the structural properties of this term eventually leads to the various versions of Weyl's law in Theorems \ref{thm:weylptw}, \ref{thm:weylptwsharp} and \ref{thm:weylintsharp}.

One can get from the decomposition \eqref{eq:decompintro} to a corresponding decomposition of the spectral density by a Tauberian theorem for a higher order Stieltjes transform. Note that
$$
(-\Delta_g + V + \lambda)^{-2}(x,x) = 2 \int_0^\infty \frac{\1(-\Delta_g+V\le t)(x,x)}{(t+\lambda)^3}\,dt \,,
\qquad
\frac{1}{8\pi\sqrt\lambda} = 2 \int_0^\infty \frac{t^{3/2}/(6\pi^2)}{(t+\lambda)^3}\,dt \,.
$$
The Tauberian theorem gives
\begin{equation}
\label{eq:decomp1intro}
\1(-\Delta_g + V\le t)(x,x) = \frac{t^{3/2}}{6\pi^2} - \frac12\, r_{0}^V(t,x)\, t^{3/2} + \mathcal O(t) \,.
\end{equation}
Due to the presence of the term $t^{3/2} r_{0}^V(t,x)$, we need a more precise Tauberian theorem than Avakumovi\'c. Fortunately for us, however, such a theorem can be established using techniques that Avakumovi\'c introduced in a different context in \cite{Avakumovic-50}, relying on the Tauberian theorem of Ingham \cite{Ingham-35} and Karamata \cite{Karamata-34}. Since we have not been able to find the relevant Tauberian theorem in the literature, we have decided to add a complete proof in Sections \ref{sec:tauberian} and \ref{sec:tauberian2}.

The precise statement of the decomposition in \eqref{eq:decompintro} appears in Proposition \ref{prop:greender} and its proof takes up all of Section \ref{sec:avakv}. We think it is worthwhile to precede this proof by our interpretation of Avakumovi\'c's result for $V=0$ in Section \ref{sec:avak}. One difference is that we find it more convenient to work with $(-\Delta_g +\lambda)^{-2}(x,x)$ than with Avakumovi\'c's
$$
\lim_{y\to x} \left( (-\Delta+\lambda)^{-1}(x,y) - \lim_{\lambda'\to 0} \left( (-\Delta+\lambda')^{-1}(x,y) - \frac{1}{\lambda'\vol_g M} \right) \right).
$$

The derivation of Theorems \ref{thm:weylptw}, \ref{thm:weylptwsharp} and \ref{thm:weylintsharp} from decomposition \eqref{eq:decomp1intro}, as well as the counterexamples to the strong pointwise Weyl law appear in Section \ref{sec:weylproof}.

In a final Section \ref{sec:eucl} we sketch an application of these ideas to the case of subsets of Euclidean space and in an appendix we recall relevant results about the Kato class.

\medskip

\noindent\textbf{Acknowledgements:} Partial support through US National Science Foundation grant DMS-1363432 (R.L.F.) and ANR DYRAQ ANR-17-CE40-001 (J. S.) are acknowledged.


\section{A Tauberian theorem for the Laplace transform}\label{sec:tauberian}

In this and in the next section we give a proof of the Tauberian theorem that is required for the proof of our main result. The proof of this theorem, which concerns the Stieltjes transform, relies on another Tauberian theorem, namely one for the Laplace transform. The latter, which is due to Ingham \cite[Thm.~III]{Ingham-35} and Karamata \cite[Satz~II]{Karamata-36} (see also \cite{Karamata-34}) is the topic of this section. For our application we need a nonasymptotic version of this theorem.

We denote by $BV_\loc[0,\infty)$ the set of functions on $[0,\infty)$ whose total variation on any compact subinterval of $[0,\infty)$ is finite. Let $A\in BV_\loc[0,\infty)$ be  such that $A(u)=\cO(e^{\alpha u})$ as $u\to+\ii$ for all $\alpha>0$. Under these assumptions, the Laplace transform
\begin{equation}
\label{eq:laplacetrafo}
f(s):=\int_0^\ii e^{-us}dA(u)
\end{equation}
is well-defined and analytic on $\re s>0$. Here, $dA$ denotes the Lebesgue--Stieljes measure on $[0,\infty)$ associated to $A$ in such a way that $dA([0,t))=A(t_-)-A(0)$. This is equivalent to the definition of the Lebesgue--Stieljes measure \cite[Chap.~4]{CarBru-book} when $A$ is extended by $A(0)$ on $(-\ii,0)$. We will often assume that $A(0)=0$ and we emphasize that this does not imply an assumption on $A(0_+)$.

The crucial assumption on $f$ is that $\frac{f(s)-a}{s}$ has a continuous extension to $s=it$ for $t\in[-T,T]$ for some $a\in\C$ (which of course implies that $a=f(0)$), and that $A$ has some generalized monotonicity property.

The following theorem is the main result of this section.

\begin{theorem}\label{coro:tauber-approx-increasing}
	Let $A\in BV_\loc[0,\infty)$ with $A(0)=0$ and $A(u)=\cO(e^{\alpha u})$ as $u\to+\ii$ for all $\alpha>0$. Assume that there are $c,\delta>0$ such that for all $u+\delta\ge v\ge u\ge0$ we have $A(v)-A(u)\ge -c$ and assume that there are $a\in\C$ and $T>0$ such that $\phi(s):=\frac{f(s)-a}{s}$, with $f$ from \eqref{eq:laplacetrafo}, can be continuously extended to $s=it$ for $t\in[-T,T]$. Then, for all $\omega>0$ and with an absolute constant $C$,
	$$
	|A(\omega)|\le C\left(|a|+\int_{-T}^T|\phi(it)|\,dt+\left(1+\frac{1}{\delta T}\right)c\right).
	$$
\end{theorem}

As we already mentioned, this theorem is due to Ingham \cite[Thm.~III]{Ingham-35} and Karamata \cite[Satz~II]{Karamata-36} (see also \cite{Karamata-34}). A slightly weaker form that, however, is \emph{not} sufficient for our purposes, appears in \cite[Sec.~22.37]{Titchmarsh-book-vol2}. In these works the result is only stated in an asymptotic form, whereas we need a nonasymptotic form. Such a version is probably well-known to experts and it is stated in \cite[Hilfss.~6]{Gromes-70} without proof and with reference to the unpublished work \cite{Pauly-65}, to which we do not have access. We also note that in \cite[Sec.~1.6]{Avakumovic-56} Avakumovi\'c sketches how this result under the additional assumption $A(\omega)=\mathcal O(\sqrt\omega)$ can be deduced from known results. While the required nonasymptotic forms of these known results are probably also known to experts, we have not been able to find an appropriate reference. In view of this situation we believe it is useful to include a self-contained and complete proof of Theorem~\ref{coro:tauber-approx-increasing}.

\begin{remark}
	Avakumovi\'c \cite{Avakumovic-56} uses the following consequence of Theorem \ref{coro:tauber-approx-increasing}: Let $A\in BV_\loc[0,\infty)$ with $A(0)=0$ and $A(u)=\cO(e^{\alpha u})$ for all $\alpha>0$. Assume that there is a $c>0$ such that for all $0\le u\le v\le u+1$ we have $A(v)-A(u)\ge -c$ and assume that
	$$s\mapsto \int_0^\ii e^{-su}A(u)\,du=:\phi(s)$$
	can be extended analytically to $|s|<\epsilon$ for some $\epsilon>0$. Then for all $\omega>0$,
	$$
	|A(\omega)|\le C\left(\int_{-\epsilon/2}^{\epsilon/2}|\phi(it)|\,dt+\left(1+\frac2\epsilon\right)c\right)
	$$
	where $C$ is an absolute constant. Indeed, since $\phi(s)=f_A(s)/s$, where $f_A$ is the Laplace transform of $A$, this follows from Corollary \ref{coro:tauber-approx-increasing} with $a=0$, $T=\epsilon/2$ and $\delta =1$.
\end{remark}


The proof of Theorem \ref{coro:tauber-approx-increasing} uses the Jackson kernel, whose properties are summarized in the following lemma.

\begin{lemma}
	For any $\eta\in[-1,1]$, define
	$$\ell(\eta):=\begin{cases}
	-\frac43(1-2|\eta|)^3+\frac83(1-|\eta|)^3 & \text{if}\ |\eta|\le1/2,\\
	\frac83(1-|\eta|)^3 & \text{if}\ 1/2\le|\eta|\le1.
	\end{cases}$$
	Then for all $y\in\R$,
	$$K(y):=\int_{-1}^1\ell(\eta)e^{-i\eta y}\,d\eta=\left(\frac{\sin(y/4)}{y/4}\right)^4
	$$
	and, in particular, $K>0$ almost everywhere and $K,\, |y| K,\, y^2 K \in L^1(\R)$.
\end{lemma}

\begin{proof}[Proof of Theorem \ref{coro:tauber-approx-increasing}]
	We follow closely the proof of \cite[Satz II]{Karamata-36}. First, we extend $A$ by zero on $(-\ii,0)$. The property $A(v)-A(u)\ge -c$ then holds for all $u\le v\le u+\delta$.
	
	\emph{Step 1.} As a preliminary step we note that by iterating the assumption on $A$, we obtain for all $u\le v$
	\begin{equation}
	\label{eq:avarcond}
	A(v)-A(u) \ge -c \left( 1+ \frac{v-u}{\delta}\right).
	\end{equation}
	
	\emph{Step 2.} We now show that
	\begin{equation}
	\label{eq:ikproof1}
	\left| \int_0^\infty T K(T(v-u)) A(u)\,du \right| \le \frac{8\pi}3  |a| + \frac43  \int_{-T}^T |\phi(it)|\,dt =:\Theta\,.
	\end{equation}
	For all $s\in\C$ with $\re s>0$, we have by integration by parts
	$$\frac{f(s)}{s}=\int_0^\ii e^{-su}A(u)\,du$$
	and therefore
	$$\phi(s)=\frac{f(s)-a}{s}=\int_0^\ii(A(u)-a)e^{-su}\,du.$$
	Now let $v\in\R$, $\sigma>0$ and $t\in[-T,T]$. We multiply the previous formula at $s=\sigma+it$ by $\ell(t/T)e^{iv t}$ and integrate over $t\in[-T,T]$. By Fubini, it follows that
	$$\int_{-T}^T\phi(\sigma+it)\ell(t/T)e^{iv t}\,dt=\int_0^\ii T K(T(v-u)) \left(A(u)-a\right) e^{-\sigma u} \,du.$$
	Since $\phi(s)$ is continuous up to $\re s=0$, we can take the limit $\sigma\to0$ of the left side. On the right side, we write
	$$
	A(u)-a = \left( A(u) + c \left( 1 + \frac{u}{\delta} \right)\right) - \left( a + c\left( 1 + \frac{u}{\delta} \right)  \right).
	$$
	For the second parenthesis on the right side we can apply dominated convergence, using the fact that $(1+|y|) K(y)\in L^1$. For the first parenthesis we apply monotone convergence, using \eqref{eq:avarcond}, $A(0)=0$ and $K\ge 0$. In this way, we obtain in the limit $\sigma\to 0$
	$$\int_{-T}^T\phi(it)\ell(t/T)e^{iv t}\,dt=\int_0^\ii T K(T(v-u)) \left(A(u)-a\right) \,du.$$
	Using
	$$
	\int_0^\ii T K(T(v-u)) \,du = \int_{-\infty}^{Tv} K(y)\,dy \le \int_\R K(y)\,dy = \frac{8\pi}3 \,, 
	$$
	and $0\le\ell(\eta)\le 4/3$, we deduce \eqref{eq:ikproof1}.
	
	\emph{Step 3.} Let us introduce
	$$
	B(u) := - \inf\left\{ A(u') :\ 0\le u'\le u \right\}.
	$$
	Since $A(0)=0$ we have $B\ge 0$ and, clearly,
	\begin{equation}
	\label{eq:ikproof2}
	A(u) \ge - B(u)
	\qquad\text{for all}\ u\ge 0 \,.
	\end{equation}
	For later purposes, we note that the assumption on $A$ implies
	$$
	B(v)-B(u) \le c
	\qquad\text{for all}\ 0\le u\le v \le u+\delta \,,
	$$
	and therefore, by iteration,
	\begin{eqnarray}
	\label{eq:ikproof3}
	B(v)-B(u) \le c \left( 1+ \frac{v-u}{\delta}\right)
	\qquad\text{for all}\ 0\le u\le v \,.
	\end{eqnarray}
	
	\emph{Step 4.} We now prove an upper bound on $A$ in terms of $B$, namely,
	\begin{equation}
	\label{eq:ikproof4}
		A(v) \le B(v) + c\left( 1 + \frac{M_1}{\delta T}\right) + \frac{3}{4\pi} \Theta 
		\qquad\text{for all}\ v\ge 0 \,,
	\end{equation}
	with $M_1 = \int_0^\infty y K(y)\,dy / \int_0^\infty K(y)\,dy$. Indeed, by \eqref{eq:ikproof1} and \eqref{eq:avarcond}, we have for any $v\ge 0$,
	\begin{align*}
	\Theta & \ge \int_0^v TK(T(v-u)) A(u)\,du + \int_v^\infty T K(T(v-u)) A(u)\,du \\
	& \ge - B(v) \int_0^v TK(T(v-u)) \,du + A(v) \int_v^\infty T K(T(v-u)) \,du \\
	& \qquad \qquad - c \int_v^\infty T K(T(v-u)) \left( 1+ \frac{u-v}{\delta}\right) du \\
	& \ge - \int_0^\infty K(y)\,dy \left( B(v) - A(v) - c\left( 1 + \frac{M_1}{\delta T}\right) \right),
	\end{align*}
	which is the same as \eqref{eq:ikproof4}.
	
	\emph{Step 5.} Next, we prove that $B$ is bounded, namely,
	\begin{equation}
	\label{eq:ikproof5}
	B(v) \le M \left( \Theta + c \left( 1+ \frac{1}{\delta T} \right)\right)
	\end{equation}
	for an absolute constant $M$. Indeed, by \eqref{eq:ikproof1}, \eqref{eq:avarcond}, \eqref{eq:ikproof4} and \eqref{eq:ikproof3} we have for any $v\ge 0$,
	\begin{align*}
	-\Theta & \le \int_0^{v} T K(T(v-T^{-1}-u)) A(u) \,du + \int_{v}^\infty T K(T(v-T^{-1}-u)) A(u) \,du \\
	& \le A(v) \int_{-\infty}^{v} T K(T(v-T^{-1}-u)) \,du + c \int_{-\infty}^{v} T K(T(v-T^{-1}-u)) \left( 1 + \frac{v-u}{\delta}\right)du \\
	& \quad + \int_{v}^\infty T K(T(v-T^{-1}-u)) B(u) \,du \\
	& \quad + \left(  c\left( 1 + \frac{M_1}{\delta T}\right) + \pi^{-1} \Theta  \right) \int_{v}^\infty T K(T(v-T^{-1}-u)) \,du \\
	& \le A(v) \int_{-\infty}^{v} T K(T(v-T^{-1}-u)) \,du + c \int_{-\infty}^{v} T K(T(v-T^{-1}-u)) \left( 1 + \frac{v-u}{\delta}\right)du \\
	& \quad + B(v) \int_{v}^\infty T K(T(v-T^{-1}-u)) \,du + c\int_{v}^\infty T K(T(v-T^{-1}-u)) \left( 1+ \frac{u-v}{\delta}\right) du \\
	& \quad + \left(  c\left( 1 + \frac{M_1}{\delta T}\right) + \pi^{-1} \Theta  \right) \int_{v}^\infty T K(T(v-T^{-1}-u)) \,du \,.
	\end{align*}
	This implies, with some absolute constant $M_2$,
	\begin{multline*}
	A(v) \int_{-\infty}^{v} T K(T(v-T^{-1}-u)) \,du \\
	\ge - B(v) \int_{v}^\infty T K(T(v-T^{-1}-u)) \,du - M_2 \left( \Theta + c \left( 1+ \frac{1}{\delta T} \right)\right),
    \end{multline*}
	which is the same as
	$$
	A(v) \int_{-1}^\infty K(y) \,dy \ge - B(v) \int_{-\infty}^{-1} K(y) \,dy - M_2 \left( \Theta + c \left( 1+ \frac{1}{\delta T} \right)\right).
	$$
	Since $B$ is nondecreasing, this implies
	$$
	-B(v) \int_{-1}^\infty K(y) \,dy \ge - B(v) \int_{-\infty}^{-1} K(y) \,dy - M_2 \left( \Theta + c \left( 1+ \frac{1}{\delta T} \right)\right),
	$$
	which is the same as \eqref{eq:ikproof5} with $M=M_2/ \int_{-1}^1 K(y)\,dy$.
	
	\emph{Step 6.} The bound in the theorem now follows from \eqref{eq:ikproof1}, \eqref{eq:ikproof4} and \eqref{eq:ikproof5}.
\end{proof}


\section{A Tauberian theorem for the Stieltjes transform}\label{sec:tauberian2}

In this section we will deduce a Tauberian theorem for the Stieltjes transform from the one for the Laplace transform proved in the previous section. The Stieltjes transform or, more precisely, a higher order version of it arises naturally in our context through powers of the resolvent of elliptic operators.

\begin{theorem}\label{thm:stieljes-3D}
 Let $\delta,t_0>0$. Let $B_0,B_1,B_2:[0,\infty)\to\R$ be such that $B_0$ is bounded, $B_1$ is nondecreasing and right-continuous, and assume that
 $$C_{B_2}:=\sup_{t\ge0}\frac{|B_2(t)|}{t_0+t},\quad -C_{B_0}:=\inf_{0\le u\le v\le u+\delta}\min\{u(B_0(v^2)-B_0(u^2)),0\} $$
 are finite. Assume also that $u\mapsto u^3B_0(u^2)$, $u\mapsto B_2(u^2)$ are right-continuous and belong to $BV_{{\rm loc}}[0,\infty)$. Assume that there are $C_0>0$, $\Lambda>0$ and $\epsilon_0\in(0,\eta)$ such that for all $\lambda\ge\Lambda$,
 $$\left|\int_0^\ii\frac{B_0(t)t^{3/2}+B_1(t)+B_2(t)}{(t+\lambda)^3}\,dt\right|\le C_0 \,\frac{e^{-\epsilon_0\sqrt{\lambda}}}{\lambda}.$$ 
 Then $B_0(t)\,t^{3/2}+B_1(t)+B_2(t)=\mathcal O(t)$ as $t\to+\infty$. More precisely, one has for all $t\ge 0$,
 $$
 |B_0(t)\,t^{3/2}+B_1(t)+B_2(t)|\le \mathcal B \,\epsilon_0^{-1} \left( t + \epsilon_0^{-2} \right) 
 $$
 where, for some universal constant $C$,
 \begin{align*}
 \mathcal B & \le C \epsilon_0^3 \left( 1 + \epsilon_0\sqrt\Lambda\, e^{\epsilon_0\sqrt{\Lambda}/2} \right) \left( \|B_0\|_{L^\infty}\Lambda^{3/2} + |B_1(0)| + C_{B_2}t_0(1+\delta^{-1}\epsilon_0^{-1})+C_{B_2}\Lambda \right) \\
& \quad + C \epsilon_0^3 \Lambda \left(  \epsilon_0\sqrt\Lambda\, e^{\epsilon_0\sqrt{\Lambda}/2} + \epsilon_0^{-2} \Lambda^{-1} \right) C_0 \\
 & \quad + C \left( \epsilon_0 + \delta^{-1} \right) \left( C_{B_0} + \|B_0\|_{L^\infty} \left(\delta + \epsilon_0^{-1} \right) + C_{B_2} \left( 1 + \ln(1+\delta \epsilon_0) \right)\right).
 \end{align*}
\end{theorem}

This theorem generalizes \cite[Thm.~22.38]{Titchmarsh-book-vol2}, whose proof seems to be influenced by Avakumovi\'c's work \cite{Avakumovic-50}. In fact, \cite[Thm.~22.38]{Titchmarsh-book-vol2} corresponds to the special case where $B_0$ is a negative constant, where $B_2\equiv 0$ and where $A$ vanishes near the origin. For our application it will be crucial that $B_0$ is \emph{not} required to be constant. 

We will typically apply this theorem in the situation where $\Lambda=C_1/\epsilon_0^2$, $\delta=\epsilon_0^{-1}$ and $t_0=0$. In this case the bound on $\mathcal B$ simplifies to
\begin{align}\label{eq:stieltjessimple}
\mathcal B & \le C \left( \left(1+C_1^2e^{\sqrt{C_1}/2}\right)\left(\|B_0\|_{L^\infty} + \epsilon_0^3  |B_1(0)| + \epsilon_0  C_{B_2} + \epsilon_0 C_0\right) + \epsilon_0  C_{B_0} \right).
\end{align}

\begin{proof}
	To enhance the readability of the proof, we present here only the proof of the qualitative assertion $\mathcal O(t)$ in the theorem, without keeping track of the precise dependence of the constants. The details of these more precise bounds are deferred to Appendix \ref{sec:tauberdetails}. We abbreviate
	$$
	A(t) := B_0(t)\,t^{3/2}+B_1(t)+B_2(t) \,.
	$$
	
	\emph{Step 1.} We show that, for all $t>0$,
	\begin{equation}
	\label{eq:qualabound}
	|A(t)|\le C(1+t^{3/2}) \,, 
	\end{equation}
	see \eqref{eq:quantabound} for a quantitative version. The bounds on $B_0$ and $B_2$ imply that 
	$$\int_0^\ii\frac{B_0(t) t^{3/2}}{(t+\lambda)^3}\,dt \le C \lambda^{-1/2},\ \int_0^\ii\frac{B_2(t)}{(t+\lambda)^3}\,dt \le C(\lambda^{-2}+\lambda^{-1}),$$
	so, in view of the exponential decay of the Stieltjes transform of $A$, we deduce that
	$$\int_0^\ii\frac{B_1(t)-B_1(0)}{(t+\lambda)^3}\,dt \le C(\lambda^{-1/2}+\lambda^{-2}).$$
	By the monotonicity of $B_1$, we deduce that for all $\lambda>0$,
	$$(B_1(\lambda)-B_1(0))/(8\lambda^2) \le \int_\lambda^\ii\frac{B_1(t)-B_1(0)}{(t+\lambda)^3}\,dt\le C(\lambda^{-1/2}+\lambda^{-2}) \,,$$
	which gives $|B_1(t)|\leq C (1+t^{3/2})$ and therefore also \eqref{eq:qualabound}.
 
 	\emph{Step 2.} We claim that
 	\begin{equation}
 	\label{eq:qualanal1}
 	s\mapsto  \int_0^\ii e^{-su} A(u^2)\,du
 	\qquad\text{is analytic in}\ \{ \re s>0\}\cup\{ |s|<\epsilon_0\} \,.
 	\end{equation}
 	This qualitative assertion is supplemented in Appendix \ref{sec:tauberdetails} by the quantitative bounds \eqref{eq:quantalapl} for $s=\sigma>0$ and \eqref{eq:quantalapl2} for $|s|\le\epsilon_0/2$.
 
 	In order to prove \eqref{eq:qualanal1} we use the representation formula (recall the Yukawa potential)
 	$$e^{-s\sqrt{t}}=\frac{1}{\pi}\int_0^\ii\frac{\sin(s\sqrt{\lambda})}{\lambda+t}\,d\lambda \,.$$
 	This implies that for all $s,t>0$,
	\begin{equation}\label{eq:repr}
 	\frac{e^{-s\sqrt{t}}}{\sqrt{t}}=\int_0^\ii\frac{\lambda^{3/2}\,\kappa(s\sqrt\lambda)}{(t+\lambda)^3}\,d\lambda
 	\end{equation}
 	with
 	\begin{eqnarray}
 	\label{eq:kappa}
 	\kappa(z):=\frac{8}{\pi}\frac{\sin z - z\cos z}{z^3}
 	\qquad\text{for all} \ z\in\C \,.
 	\end{eqnarray}
 	Using \eqref{eq:repr} we can write
 	$$
 	\int_0^\ii e^{-su} A(u^2)\,du = \frac12 \int_0^\ii\frac{e^{-s\sqrt{t}}}{\sqrt{t}}A(t)\,dt = \frac12 \int_0^\ii \lambda^{3/2}\,\kappa(s\sqrt\lambda) \int_0^\ii\frac{A(t)}{(t+\lambda)^3}\,dt\,d\lambda \,.
 	$$
 	The fact that $\kappa$ is entire and satisfies the bound
 	\begin{equation}
 	\label{eq:kappabound}
 	|\kappa(z)|\le Ce^{|z|} \qquad\text{for all}\ z\in\C \,,
 	\end{equation}
 	together with the assumed exponential decay for the Stieltjes transform implies \eqref{eq:qualanal1}.
 	
 	\emph{Step 3.} In order to apply Theorem \ref{coro:tauber-approx-increasing}, we would like to deal with functions vanishing near the origin and we define, with a parameter $u_0>0$ to be specified later,
 	$$
 	g_0(u):=\begin{cases}
 	A((u-u_0)^2) - A(0) & \text{if}\ u\ge u_0 \,,\\
 	0 & \text{if}\ u<u_0 \,.
 	\end{cases}
 	$$
 	We claim that
 	\begin{equation}
 	\label{eq:qualanal2}
 	s\mapsto  \int_0^\ii e^{-su} \,dg_0(u)
 	\qquad\text{is analytic in}\ \{ \re s>0\}\cup\{ |s|<\epsilon_0\}
 	\end{equation}
 	and
 	\begin{equation}\label{eq:bound-dg0}
 	\left|\int_0^\ii e^{-tu}dg_0(u)\right|\le Ce^{-tu_0}(1+t^{-3})
 	\end{equation}
 	for all $t>0$. These qualitative assertions are supplemented by the quantitative bounds \eqref{eq:quantg0lapl} for $s=\sigma>0$ and \eqref{eq:quantg0lapl2} for $|s|\le\epsilon_0/2$.
 	
 	To prove \eqref{eq:qualanal2} and \eqref{eq:bound-dg0} we define $g(u):=A(u^2)-A(0)$. Then, using $g(0)=0$, we see that
	$$
	\int_0^\ii e^{-su}dg(u) = s\int_0^\ii e^{-su} g(u)\,du = s\int_0^\ii e^{-su} A(u^2)\,du - A(0) \,.
	$$
	According to \eqref{eq:qualanal1} this is analytic in $\{\re s>0\}\cup\{|s|<\epsilon_0\}$. Moreover, from \eqref{eq:qualabound} we obtain the rough bound, for all $t>0$,
 	$$\left|\int_0^\ii e^{-tu}dg(u)\right|\le C(1+t^{-3})$$
 	In terms of the function $g_0$, we have
 	$$
 	\int_0^\ii e^{-su}dg_0(u) = \int_{u_0}^\ii e^{-su}dg(u-u_0) = \int_0^\ii e^{-s(u+u_0)}dg(u) = e^{-su_0}\int_0^\ii e^{-su}dg(u) \,.
 	$$
 	Therefore the analyticity properties and the bound of the corresponding integral with $g$ imply \eqref{eq:qualanal2} and \eqref{eq:bound-dg0}.
 	
	\emph{Step 4.} We apply Theorem \ref{coro:tauber-approx-increasing} to the function
 	$$h(u):=\int_{[0,u]}\frac{dg_0(w)}{w^2} \,.$$
	Note that, since $g_0$ vanishes for $u\le u_0$,  the function $h$ is well-defined, belongs to $h\in BV_\loc[0,\infty)$ and satisfies $h(0)=0$. Moreover, its Laplace transform
	$$
	f(s) := \int_0^\ii e^{-su}dh(u)
	$$
	is  well-defined for $\re s>0$ by \eqref{eq:bound-dg0} and satisfies
 	\begin{equation}\label{eq:quantf}
 		f(s) = \int_0^\ii \frac{e^{-su}}{u^2}dg_0(u)
 		= \int_s^\infty \int_{s'}^\infty \int_0^\infty e^{-s'' u} \,dg_0(u)\,ds''\,ds' \,.
 	\end{equation}
 	Therefore, in view of \eqref{eq:qualanal2}, $f$ is analytic in $\{\re s>0\}\cup\{ |s|<\epsilon_0\}$. This qualitative assertion is supplemented with the quantitative bounds \eqref{eq:quanthlapl} on $f(0)$ and \eqref{eq:quanthlapl2} on $(f(s)-f(0))/s$ for $|s|\le\epsilon_0/2$.
 	
   \emph{Step 5.} We claim that the function $g$ satisfies the almost monotonicity property $h(v)-h(u)\ge -c$ for $u+\delta\ge v\ge u\ge0$. The proof is not complicated, but somewhat lengthy and we defer it to Appendix \ref{sec:tauberdetails}; see \eqref{eq:bound-c}.
   
   \emph{Step 6.} According to Steps 4 and 5, the function $h$ satisfies all the assumptions of Theorem \ref{coro:tauber-approx-increasing} and we deduce that $h(u)=\cO(1)$ as $u\to+\ii$. Note now that since $g_0$ is right-continuous,
    $$\int_0^x uh(u)\,du=\int_{[0,x]} \frac{dg_0(v)}{v^2} \int_v^xu\,du = \frac12 x^2 h(x) - \frac12 g_0(x).$$
    Hence, $g_0(x)=x^2 h(x) -2\int_0^x uh(u)\,du=\cO(x^2)$ as $x\to+\ii$, meaning that $A(u)=\cO(u)$ as $u\to+\ii$. A quantitative version of this bound, valid for all $u>0$ is proved in \eqref{eq:quantaboundfinal}. This completes the proof of Theorem \ref{thm:stieljes-3D}.
\end{proof}

\section{The method of Avakumovi\'c}\label{sec:avak}

We now apply Theorem \ref{thm:stieljes-3D} to prove the sharp Weyl laws in 3D, inspired by the method of Avakumovi\'c. Let $(M,g)$ a three-dimensional smooth compact Riemannian manifold without boundary, and let $-\Delta_g$ the (non-negative) Laplace-Beltrami operator on $(M,g)$. For $\lambda>0$ and $x,y\in M$ with $x\neq y$, denote the Green's function by
$$G_{\lambda}(x,y)=(-\Delta_g+\lambda)^{-1}(x,y).$$
The key result to obtain a sharp Weyl law is the following Avakumovi\'c-type result:

\begin{proposition}\label{prop:avak}
	There are constants $C_1,C_2,C_3>0$, $\epsilon_0>0$, and a function $R: [0,\ii)\times M\to\R$ such that for all $x\in M$ and for all $\lambda \ge C_1/\epsilon_0^2$,
	\begin{align*}
		& \left| \partial_\lambda G_\lambda(x,x)+\frac{1}{8\pi\sqrt{\lambda}}-\int_0^\ii\frac{R(t,x)}{(t+\lambda)^3}\,dt \right| \le C_2\, \frac{e^{-\epsilon_0\sqrt{\lambda}/4}}{\lambda}
	\end{align*}
	and for all $x\in M$, $R(\cdot,x)$ is $C^1$ on $[0,\infty)$ and for all $t\ge0$,
	$$
	|R(t,x)| \le C_3\, t.
	$$
	More precisely, the constants $C_1,C_2,C_3$ and $\epsilon_0$ can be chosen as
	$$
	C_1 = C\cM_{\rho/2} \,,
	\qquad
	C_2 = \frac{C\, \mathcal M_{\rho/2}^{3/2}\,\mathcal M_{\rho/2}'}{\rho} \,,
	\qquad
	C_3 = \frac{C \,\mathcal M_{\rho/2}^{1/2} \, \mathcal M_{\rho/2}' }{\rho} \,,
	\qquad
	\epsilon_0 = \frac{\rho}{C\,\mathcal M_{\rho/2}^{1/2}} \,,
	$$
	where $C$ is an absolute constant, $\rho$ is the injectivity radius of $M$ and $\mathcal M_{\rho/2}$ and $\mathcal M_{\rho/2}'$ are explicit constants, defined below, which depend in a scale invariant way on $M$.
\end{proposition}

\begin{corollary}\label{coro:weyl-3D}
	We have, as $\lambda\to\infty$, uniformly in $x\in M$,
	$$\1(-\Delta_g\le\lambda)(x,x)=\frac{\lambda^{3/2}}{6\pi^2}+\cO(\lambda) \,.$$
	More precisely, we have for all $x\in M$ and all $t\ge 0$,
	$$
	\left| \1(-\Delta_g\le\lambda)(x,x) - \frac{\lambda^{3/2}}{6\pi^2} \right| \le  C_0 \, \rho^{-1} \left( t + \rho^{-2} \right),
	$$
	where $C_0$ depends only on upper bounds on $\mathcal M_{\rho/2}$, $\mathcal M_{\rho/2}'$ and $\rho^3/\vol_g M$.
\end{corollary}

\begin{proof}[Proof of Corollary \ref{coro:weyl-3D} assuming Proposition \ref{prop:avak}]
	Define for all $t\ge0$
	$$A(t,x) :=\1(-\Delta_g\le t)(x,x)-t^{3/2}/(6\pi^2)+R(t,x)/2.$$
	Let $\lambda_n$ denote the eigenvalues of $-\Delta_g$ in nondecreasing order and repeated according to multiplicities and $\phi_n$ the corresponding normalized eigenfunctions. Then
	$$\1(-\Delta_g\le t)(x,x)=\sum_{\lambda_n\le t}|\phi_n(x)|^2$$
	and thus
	$$\int_0^\ii\frac{\1(-\Delta_g\le t)(x,x)}{(\lambda+t)^3}\,dt=\frac12\sum_{n\ge0}\frac{|\phi_n(x)|^2}{(\lambda_n+\lambda)^2}=-\frac12\partial_\lambda G_\lambda(x,x).$$
	Next, we have 
	$$\int_0^\ii \frac{t^{3/2}}{(\lambda+t)^3}\,dt=\frac{3\pi}{8\sqrt{\lambda}},$$
	so that, by Proposition \ref{prop:avak},
	$$\left|\int_0^\ii\frac{A(t,x)}{(t+\lambda)^3}\,dt\right|=\left|-\frac12\partial_\lambda G_\lambda(x,x)-\frac{1}{16\pi\sqrt{\lambda}}+\frac12\int_0^\ii\frac{R(t,x)}{(t+\lambda)^3}\,dt\right|\le C_2\, \frac{e^{-\epsilon_0\sqrt{\lambda}/4}}{\lambda}$$
	for all $\lambda\ge C_1/\epsilon_0^2$. The result then follows from Theorem \ref{thm:stieljes-3D} in the simplified form \eqref{eq:stieltjessimple} with $B_0\equiv  -(6\pi)^{-1}$, $B_1(t)=\1(-\Delta_g\le t)(x,x)$, $B_2(t)= R(t,x)/2$. Note, in particular, that the almost monotonicity assumption on $B_0$ in that theorem is satisfied with $C_{B_0}=0$ and any $\delta\ge 0$. In order to obtain the claimed bound we apply Theorem \ref{thm:stieljes-3D} with the choice $\delta = \epsilon_0^{-1}$. (The parameters $\epsilon_0$ in Theorem \ref{thm:stieljes-3D} and in Proposition \ref{prop:avak} have essentially the same meaning, up to a multiplicative constant.) We also use the fact that $B_1(0)= 1/\vol_g M$.
\end{proof}


\subsection{Proof of Proposition \ref{prop:avak}}\label{sec:parametrix}

We first recall the construction of a local parametrix for $G_\lambda$. Let us denote by $\rho>0$ the injectivity radius of $M$. Then, for any $y\in M$, the exponential $\exp_y$ at $y$ is well-defined on the set $\{v\in T_y M, |v|<\rho\}$, and maps it to the set $\{x\in M,\ d_g(x,y)<\rho\}$. As explained in Avakumovi\'c \cite{Avakumovic-56}, in Minakshisundaram--Pleijel \cite{MinPle-49}, or in \cite[Sec. III.E.3]{BerGauMaz-book}, for any $(x,y)\in M\times M$ with $d_g(x,y)<\rho$, the function
$$\theta(x,y):=|\det T_{\exp_y^{-1}(x)} \exp_y|,$$
where $T$ is the tangent map, is such that $\theta(y,y)=1$ and such that the Riemannian volume in normal coordinates around $y$ is $\theta(\exp_y(v),y)\,dv$, where $dv$ is the Lebesgue measure on $T_y M$. The function $U_0:=\theta^{-1/2}$ is then smooth on $V_\rho:=\{(x,y)\in M\times M,\ d_g(x,y)<\rho\}$. In order to quantify the dependence of our error estimates on the geometry of the manifold, we introduce for $\epsilon\in(0,\rho]$
$$
\mathcal M_\epsilon := \sup_{(x,y)\in V_\epsilon} \frac{ U_0(x,y) + \epsilon^2 |(-\Delta_g)_x U_0(x,y)| }{\min\{U_0(x,y)^2,U_0(y,x)^2\}}
\qquad\text{and}\qquad
\mathcal M_\epsilon' := \sup_{(x,y)\in V_\epsilon} U_0(x,y) \,.  
$$
Let $\chi\in C^\ii_0(\R)$ with $\chi\equiv1$ on $[-1/2,1/2]$ and $\chi\equiv0$ outside of $(-1,1)$. We consider $\chi$ as fixed and consider $L^\infty$ bounds on $\chi$, $\chi'$ and $\chi''$ as absolute constants. For $\epsilon\in(0,\rho]$ and $x\neq y$, we define
$$
T_{\lambda,\epsilon}(x,y):=\frac{e^{-\sqrt{\lambda}d_g(x,y)}}{4\pi d_g(x,y)}U_0(x,y)\chi(\epsilon^{-1} d_g(x,y))
$$
and note that
\begin{eqnarray}
	\label{eq:tbounds}
	\|T_{\lambda,\epsilon} \|_{L^\ii_x L^2_y}\le \frac{C}{\lambda^{1/4}}.
\end{eqnarray}
Next, let
$$
\gamma_{\lambda,\epsilon}(x,y):=T_{\lambda,\epsilon}(x,y)-G_\lambda(x,y)
$$
and compute using the expression of the Laplace-Beltrami operator in normal coordinates \cite[Sec. II.G.V]{BerGauMaz-book}
\begin{align*}
	R_{\lambda,\epsilon}(x,y) & := (-\Delta_g+\lambda)_x\gamma_{\lambda,\epsilon}(x,y) \\
	& = \frac{e^{-\sqrt{\lambda}d_g(x,y)}}{4\pi d_g(x,y)} \chi(\epsilon^{-1} d_g(x,y)) (-\Delta_g)_xU_0(x,y)\\
	&  \quad - \frac{e^{-\sqrt{\lambda}d_g(x,y)}}{4\pi d_g(x,y)} \, U_0(x,y) \left( \epsilon^{-2} \chi''(\epsilon^{-1}d_g(x,y)) - 2\epsilon^{-1}\sqrt\lambda \,\chi'(\epsilon^{-1} d_g(x,y)) \right).
\end{align*}
This implies
$$|R_{\lambda,\epsilon}(x,y)|\le C \mathcal M_\epsilon \, \frac{e^{-\frac{\sqrt{\lambda}}{2} d_g(x,y)}}{\epsilon^2 d_g(x,y)} \min\{ U_0(x,y)^2,U_0(y,x)^2\} \1(d_g(x,y)<\epsilon)$$
and, consequently,
\begin{eqnarray}
	\label{eq:rbounds}
	\|R_{\lambda,\epsilon}\|_{L^\ii_x L^1_y}+ \|R_{\lambda,\epsilon}\|_{L^\ii_y L^1_x}\le \frac{C\mathcal M_\epsilon }{\epsilon^2 \lambda}\quad\text{and}\qquad
	\|R_{\lambda,\epsilon} \|_{L^\ii_x L^2_y}+ \|R_{\lambda,\epsilon} \|_{L^\ii_y L^2_x}\le \frac{C\mathcal M_\epsilon\mathcal M_\epsilon'}{\epsilon^2 \lambda^{1/4}}.
\end{eqnarray}
Applying $(-\Delta_g+\lambda)^{-1}$ to the equation defining $R_{\lambda,\epsilon}$, we obtain the integral equation
$$\gamma_{\lambda,\epsilon}(x,y)=\int_M T_{\lambda,\epsilon}(x,z) R_{\lambda,\epsilon}(z,y)\,dv_g(z)-\int_M\gamma_{\lambda,\epsilon}(x,z) R_{\lambda,\epsilon}(z,y)\,dv_g(z).$$
Formally, this equation can be solved by iteration and one obtains the series representation
$$\gamma_{\lambda,\epsilon}(x,y)=\sum_{n\ge1}(-1)^{n+1}\gamma_{\lambda,\epsilon}^{(n)}(x,y)$$
with
$$
\gamma_{\lambda,\epsilon}^{(n)}(x,y):=\int_M\,dv_g(z_1)\cdots\int_M\,dv_g(z_n)\, T_{\lambda,\epsilon}(x,z_1)R_{\lambda,\epsilon}(z_1,z_2)\cdots R_{\lambda,\epsilon}(z_n,y).
$$

We now prove bounds on $\gamma_{\lambda,\epsilon}^{(n)}$ which show, in particular, that the above series converges provided $\lambda\epsilon^2$ is large enough. By the Cauchy-Schwarz inequality with respect to the $z_1$-integration we obtain
\begin{equation}\label{eq:bound-gamman}
	|\gamma_{\lambda,\epsilon}^{(n)}(x,y)|\le \|T_{\lambda,\epsilon} \|_{L^\ii_x L^2_y}\|R_{\lambda,\epsilon} \|_{L^\ii_y L^2_x}\|R_{\lambda,\epsilon}\|_{L^\ii_y L^1_x}^{n-1}.
\end{equation}
Inserting the bounds from \eqref{eq:tbounds} and \eqref{eq:rbounds} we find that for all $n\ge1$ and all $(x,y)\in M\times M$,
$$|\gamma_{\lambda,\epsilon}^{(n)}(x,y)|\le \frac{C\mathcal M_\epsilon \mathcal M_\epsilon'}{\epsilon^2\lambda^{1/2}}\left(\frac{A \mathcal M_\epsilon}{\epsilon^2\lambda}\right)^{n-1}.$$
This implies that the Neumann series defining $\gamma_{\lambda,\epsilon}$ converges for $\lambda \epsilon^2 >A \mathcal M_\epsilon$ and for, say, $\lambda \epsilon^2 \ge 2 A \mathcal M_\epsilon$ we have
$$
|\gamma_{\lambda,\epsilon}(x,y)| \le \frac{C\mathcal M_\epsilon \mathcal M_\epsilon'}{\epsilon^2\lambda^{1/2}}.
$$
This implies that $\gamma_{\lambda,\epsilon}$, as an operator, maps $L^1$ to $L^\ii$, and hence also $L^2$ to $L^2$. To show that we indeed have $G_\lambda = T_{\lambda,\epsilon} - \gamma_{\lambda,\epsilon}$ (where now $\gamma_{\lambda,\epsilon}$ is defined as the sum of the above convergent series), we notice that by construction we have $(-\Delta_g+\lambda)_x (T_{\lambda,\epsilon} - \gamma_{\lambda,\epsilon})(x,y)=\delta_y$, or, as operators, $(-\Delta_g+\lambda)(T_{\lambda,\epsilon} - \gamma_{\lambda,\epsilon}) f = f$ for all $f\in C^\ii_0(M)$. Since there is a unique bounded operator on $L^2$ that satisfies this relation, we indeed have the identity $G_\lambda = T_{\lambda,\epsilon} - \gamma_{\lambda,\epsilon}$.

In order to prove Proposition \ref{prop:avak} we need to take the structure of $\gamma_{\lambda,\epsilon}^{(n)}$ more precisely into account and we split
$$\gamma_{\lambda,\epsilon}^{(n)}(x,y)=a^{(n)}_{\lambda,\epsilon}(x,y)+b^{(n)}_{\lambda,\epsilon}(x,y)$$
with
\begin{multline*}
	a^{(n)}_{\lambda,\epsilon}(x,y):=\int_{d_g(z_n,y)<\epsilon/2}\,dv_g(z_n)\int_{d_g(z_{n-1},z_n)<\epsilon/2}\,dv_g(z_{n-1})\cdots\int_{d_g(z_1,z_2)<\epsilon/2}\,dv_g(z_1)\times\\
	\times T_{\lambda,\epsilon}(x,z_1)R_{\lambda,\epsilon}(z_1,z_2)\cdots R_{\lambda,\epsilon}(z_n,y).
\end{multline*}

\begin{lemma}\label{lem:main}
	There are $A,C>0$ and for every $n\ge 1$ and every $\epsilon\in(0,\rho)$ there is a function $r^{(n)}_\epsilon:[0,\infty)\times M\to\R$ such that for all $x\in M$ and all $\lambda>0$
	$$\partial_\lambda a^{(n)}_{\lambda,\epsilon}(x,x)= - \int_0^\ii\frac{r^{(n)}_\epsilon(t,x)}{(t+\lambda)^3}\,dt
	$$
	and for all $x\in M$, all $n\ge1$ and all $\epsilon\in(0,\rho)$, $r^{(n)}_\epsilon(\cdot,x)$ is $C^1$ on $[0,\infty)$ and for all $t>0$
	$$
	|r^{(n)}_\epsilon(t,x)|\le C \mathcal M_{\rho/2} \mathcal M_{\rho/2}' \,\frac{\epsilon t}{\rho^2} \left( \frac{A \mathcal M_{\rho/2}\, \epsilon^2}{\rho^2} \right)^{n-1}.
	$$
\end{lemma}

\begin{proof}
	Due to the localization of the integrals, we have
	\begin{multline*}
		a^{(n)}_{\lambda,\epsilon}(x,y):=\int_{d_g(z_n,y)<\epsilon/2}\,dv_g(z_n)\int_{d_g(z_{n-1},z_n)<\epsilon/2}\,dv_g(z_{n-1})\cdots\int_{d_g(z_1,z_2)<\epsilon/2}\,dv_g(z_1)\times\\
		\times \chi^{(n)}_{\lambda,\epsilon}(x,z_1,\ldots,z_n,y)
	\end{multline*}
	with
	\begin{multline*}
		\chi^{(n)}_{\lambda,\epsilon}(x,z_1,\ldots,z_n,y) = \frac{e^{-\sqrt{\lambda}(d_g(x,z_1)+\cdots+d_g(z_n,y))}}{(4\pi)^{n+1}}\times\\
		\times\frac{U_0(x,z_1)\chi(\epsilon^{-1} d_g(x,z_1))\Delta U_0(z_1,z_2)\cdots\Delta U_0(z_n,y)}{d_g(x,z_1)\cdots d_g(z_n,y)}.
	\end{multline*}
	We now use the representation formula \eqref{eq:repr} with $\kappa$ defined by \eqref{eq:kappa} to write
	\begin{multline*}
		\partial_\lambda \chi^{(n)}_{\lambda,\epsilon}(x,z_1,\ldots,z_n,x)  = -\frac{1}{2(4\pi)^{n+1}} \int_0^\infty dt\,\frac{t^{3/2}}{(t+\lambda)^3} 
		\, \kappa(\sqrt t\,( d_g(x,z_1)+\cdots+d_g(z_n,x)))\times \\
		 \times \left( d_g(x,z_1)+\cdots+d_g(z_n,x)\right)  \frac{U_0(x,z_1)\chi(\epsilon^{-1} d_g(x,z_1))\Delta U_0(z_1,z_2)\cdots\Delta U_0(z_n,x)}{d_g(x,z_1)\cdots d_g(z_n,x)}.
	\end{multline*}
	This yields the formula in the lemma with
	\begin{align*}
		r^{(n)}_\epsilon(t,x)& :=\frac{t^{3/2}}{2(4\pi)^{n+1}}\int_{d_g(z_n,x)<\epsilon/2}\,dv_g(z_n)\int_{d_g(z_{n-1},z_n)<\epsilon/2}\,dv_g(z_{n-1})\cdots\int_{d_g(z_1,z_2)<\epsilon/2}\,dv_g(z_1) \\
		& \qquad\times (d_g(x,z_1)+\cdots+d_g(z_n,x)) \, \kappa(\sqrt t\,( d_g(x,z_1)+\cdots+d_g(z_n,x))) \\
		& \qquad\times\frac{U_0(x,z_1)\chi(\epsilon^{-1} d_g(x,z_1))\Delta U_0(z_1,z_2)\cdots\Delta U_0(z_n,x)}{d_g(x,z_1)\cdots d_g(z_n,x)}\,.
	\end{align*}
	In order to prove the bound in the lemma, we use the bound
	\begin{eqnarray}
	\label{eq:kappabound1}
	|\kappa(r)|\le C r^{-1} \qquad\text{for all}\ r>0 \,, 
	\end{eqnarray}
	and obtain
	\begin{align*}
		|r^{(n)}_\epsilon(t,x)|& \le \frac{C t}{2(4\pi)^{n+1}} \int_{d_g(z_n,x)<\epsilon/2}\,dv_g(z_n)\int_{d_g(z_{n-1},z_n)<\epsilon/2}\,dv_g(z_{n-1})\cdots\int_{d_g(z_1,z_2)<\epsilon/2}\,dv_g(z_1) \\
		& \qquad\times\frac{|U_0(x,z_1)\chi(\epsilon^{-1} d_g(x,z_1))\Delta U_0(z_1,z_2)\cdots\Delta U_0(z_n,x)|}{d_g(x,z_1)\cdots d_g(z_n,x)} \\
		& \le \frac{C t}{2(4\pi)^{n+1}} \mathcal M_{\rho/2}^{n} (\rho/2)^{-2n} \int\,dv_g(z_n)\int\,dv_g(z_{n-1})\cdots\int\,dv_g(z_1) \\
		& \qquad\times \frac{\1(d_g(x,z_1)<\epsilon)U_0(x,z_1)}{d_g(x,z_1)} \frac{\1(d_g(z_1,z_2)<\epsilon/2)U_0(z_2,z_1)^2}{d_g(z_1,z_2)}\times\\
		& \qquad \times\cdots\times \frac{\1(d_g(z_n,x)<\epsilon/2)U_0(x,z_n)^2}{d_g(z_n,x)}.
	\end{align*}
	We bound the right side similarly as in \eqref{eq:bound-gamman} using
	$$\left\| \frac{\1(d_g(x,y)<\epsilon/2) U_0(y,x)^2}{d_g(x,y)} \right\|_{L^\ii_y L^1_x}\le C \epsilon^2,\quad \left\| \frac{\1(d_g(x,y)<\epsilon/2) U_0(y,x)}{d_g(x,y)} \right\|_{L^\ii_y L^2_x}\le C \sqrt{\epsilon}.$$
	We obtain
	$$
	|r^{(n)}_\epsilon(t,x)| \le C \mathcal M_{\rho/2} \mathcal M_{\rho/2}' \,\frac{\epsilon t}{\rho^2} \left( \frac{A \mathcal M_{\rho/2}\, \epsilon^2}{\rho^2} \right)^{n-1},
	$$
	as claimed. Using that $\kappa'$ is bounded, one can similarly show that $r^{(n)}_{\epsilon}(\cdot,x)$ is $C^1$ on $[0,\infty)$ with a bound $|\partial_tr^{(n)}_\epsilon(t,x)|\le C_{\rho,\epsilon}\sqrt{t}(1+\sqrt{t})n^2(A_\rho\epsilon^2)^{n-1}$.
\end{proof}

\begin{lemma}\label{lem:remainder}
	There exists $A>0$ and $C>0$ such that for all $x\in M$, all $\lambda>0$, all $n\ge1$, and all $\epsilon\in(0,\rho)$ we have
	$$|\partial_\lambda b^{(n)}_{\lambda,\epsilon}(x,x)|\le C \mathcal M_\epsilon' \,\mathcal M_\epsilon \, \frac{e^{-\epsilon\sqrt{\lambda}/4}}{\epsilon \lambda} \left(\frac{A\mathcal M_\epsilon}{\epsilon^2\lambda}\right)^{n-1}.$$
\end{lemma}

\begin{proof}
	We first bound the integrand of $b_{\lambda,\epsilon}^{(n)}$ by absolute values and then bound
	$$
	1- \1(d_g(z_n,x)<\epsilon/2 \,,\ d_g(z_{n-1},z_n)<\epsilon/2\,,\ldots, d_g(z_1,z_2)<\epsilon/2) \le \sum_{k=1}^{n} \1(d_g(z_k,z_{k+1})\ge\epsilon/2) \,,
	$$
	where we put $z_{n+1}=x$ on the right side. In this way we can bound $|\partial_\lambda b^{(n)}_{\lambda,\epsilon}(x,x)|$ by a sum of $n(n+1)$ integrals. The factor of $n$ comes from the constraint $\1(d_g(z_k,z_{k+1})\ge\epsilon/2)$, $k=1,\ldots,n$, and the factor $n+1$ comes from the product rule when differentiating with respect to $\lambda$ the product of $n+1$ functions.
	
	The factors without $\lambda$ derivatives can be bounded as before by
	\begin{align*}
		|T_{\lambda,\epsilon}(x,y)| & \le C \frac{e^{-\sqrt\lambda d_g(x,y)}}{d_g(x,y)} U_0(x,y) \1(d_g(x,y)<\epsilon) \,, \\
		|R_{\lambda,\epsilon}(x,y)| & \le C \mathcal M_\epsilon \frac{e^{-\frac12\sqrt\lambda d_g(x,y)}}{\epsilon^2 d_g(x,y)} \min\{U_0(x,y)^2,U_0(y,x)^2\} \1(d_g(x,y)<\epsilon) \,.
	\end{align*}
	Using the explicit expression for $T_{\lambda,\epsilon}$ and $R_{\lambda,\epsilon}$ one easily finds that
	\begin{align*}
		|\partial_\lambda T_{\lambda,\epsilon}(x,y)| & \le C \frac{e^{-\sqrt\lambda d_g(x,y)}}{\sqrt\lambda} U_0(x,y) \1(d_g(x,y)<\epsilon) \,, \\
		|\partial_\lambda R_{\lambda,\epsilon}(x,y)| & \le C \mathcal M_\epsilon \frac{e^{-\sqrt\lambda d_g(x,y)}}{\epsilon^2 \sqrt{\lambda}} \min\{U_0(x,y)^2,U_0(y,x)^2\} \1(d_g(x,y)<\epsilon) \,.
	\end{align*}
	We see that the effect of the derivative $\partial_\lambda$ is to multiple the bounds by a factor of $d_g/\sqrt\lambda$ and, due to the localization of the integrals, this can be bounded by $\epsilon/\sqrt\lambda$. Moreover, we see that the difference between the bounds involving $R_{\lambda,\epsilon}$ and $T_{\lambda,\epsilon}$ (and their derivatives) is a factor $\epsilon^{-2}$ in the former. This shows that, again with the convention $z_{n+1}=x$,
	\begin{align*}
		|\partial_\lambda b^{(n)}_{\lambda,\epsilon}(x,x)| & \le C (n+1)\mathcal M_\epsilon^n \frac{1}{\epsilon^{2n}} \frac{\epsilon}{\sqrt{\lambda}} \sum_{k=1}^n \int dv_g(z_n) \cdots \int dv_g(z_1) \1(d_g(z_k,z_{k+1})\ge\epsilon/2) \\
		& \qquad \times 
		\frac{e^{-\frac12\sqrt\lambda d_g(x,z_1)}}{d_g(x,z_1)} U_0(x,z_1) \1(d_g(x,z_1)<\epsilon) \\
		& \qquad \times
		\frac{e^{-\frac12\sqrt\lambda d_g(z_1,z_2)}}{d_g(z_1,z_2)} \min\{U_0(z_1,z_2)^2,U_0(z_2,z_1)^2\} \1(d_g(z_1,z_2)<\epsilon)
		\cdots \\
		& \qquad\times
		\frac{e^{-\frac12\sqrt\lambda d_g(z_{n},z_{n+1})}}{d_g(z_{n},z_{n+1})} \min\{U_0(z_{n},z_{n+1})^2,U_0(z_{n+1},z_{n})^2\} \1(d_g(z_{n},z_{n+1})<\epsilon).
	\end{align*}
	For all $1\le k\le n$ we bound
	\begin{align*}
		& \left| \int dv_g(z_1)\cdots \int dv_g(z_n)\, K_0(x,z_1) K_1(z_1,z_2)\cdots K_n(z_n,y) \right| \\
		& \le \| K_0 \|_{L^\infty_x L^2_y} \left( \prod_{j=1}^{k-1} \|K_j\|_{L^\infty_x L^1_y}^{1/2} \|K_j\|_{L^\infty_y L^1_x}^{1/2} \right)
		\| K_k\|_{L^\infty_y L^2_x} \left( \prod_{\ell=k+1}^n \| K_\ell \|_{L^\infty_y L^1_x} \right),
	\end{align*}
	which is proved by factorizing the kernel as
	$$
	\left( K_0 K_1^{1/2}\cdots K_{k-1}^{1/2} K_{k+1}^{1/2} \cdots K_n^{1/2} \right) \\
	\times \left( K_1^{1/2} \cdots K_{k-1}^{1/2} K_k K_{k+1}^{1/2} \cdots K_n^{1/2} \right)
	$$
	(suppressing the arguments for simplicity) and applying the Cauchy--Schwarz inequality. This bound allows us to always put the factor involving the pair $(x,z_1)$ (which only has a single $U_0$) and the factor involving $(z_k,z_{k+1})$ (which has the additional cut-off away from the diagonal) into $L^2$. Using \eqref{eq:tbounds}, \eqref{eq:rbounds} together with
	$$
	\left\|\frac{e^{-\frac12\sqrt\lambda d_g(x,y)}}{d_g(x,y)}\min\{U_0(x,y)^2,U_0(y,x)\}^2\1(\epsilon/2\le d_g(x,y)<\epsilon)\right\|_{L^\ii_y L^2_x} \le \frac{C \cM_\epsilon'}{\lambda^{1/4}} e^{-\frac14\epsilon\sqrt{\lambda}},
	$$
	yields the bound 
	$$|\partial_\lambda b^{(n)}_{\lambda,\epsilon}(x,x)|\le Cn(n+1) \mathcal M_\epsilon' \,\mathcal M_\epsilon \,\frac{e^{-\epsilon\sqrt{\lambda}/4}}{\epsilon \lambda} \left(\frac{A\mathcal M_\epsilon}{\epsilon^2\lambda}\right)^{n-1}.$$
	The term $n(n+1)$ here can be dropped by increasing the absolute constant $A$. This completes the proof of the lemma.
\end{proof}

\begin{proof}[Proof of Proposition \ref{prop:avak}]
	With $r^{(n)}_\epsilon$ from Lemma \ref{lem:main} let $R_\epsilon(t,x)=\sum_{n\ge1}(-1)^{n+1} r^{(n)}_\epsilon(t,x)$. According to that lemma, if we choose $\epsilon=\epsilon_0=\rho/ \sqrt{2A\mathcal M_{\rho/2}}$,  then the series defining $R_\epsilon(t,x)$ converges and for all $x\in M$, $R_\epsilon(\cdot,x)$ is $C^1$ on $[0,\infty)$ and for all $t>0$,
	$$
	|R_\epsilon(t,x)| \le C \,\mathcal M_{\rho/2}^{1/2} \, \mathcal M_{\rho/2}' \,\frac{t}{\rho}.
	$$
	We may assume that the absolute constant in the definition of $\epsilon$ satisfies $A\ge 2$. This implies that $\epsilon_0\le\rho/2$, since $\mathcal M_\epsilon\ge 1$, which in turn is a consequence of $U_0(x,x)=1$.
	
	This bound on $R_\epsilon$ and the bound from Lemma \ref{lem:remainder} imply that $\gamma_{\lambda,\epsilon}$ can be differentiated with respect to $\lambda$, that this derivative is given by differentiating its series expansion termwise and that for all $\lambda\ge 2 A'\cM_{\rho/2}/\epsilon^2$
	\begin{align*}
		\left| \partial_\lambda\gamma_{\lambda,\epsilon}(x,x)+\int_0^\ii\frac{R_\epsilon(t,x)}{(t+\lambda)^3}\,dt \right| & = \left| \sum_{n\ge 1} (-1)^{n+1} \left(\gamma_{\lambda,\epsilon}^{(n)}(x,x) - a_{\lambda,\epsilon}^{(n)}(x,x) \right) \right| \\
		& = \left| \sum_{n\ge 1} (-1)^{n+1} b_{\lambda,\epsilon}^{(n)}(x,x) \right| \\
		& \le C \mathcal M_\epsilon\, \mathcal M_\epsilon' \,\frac{e^{-\epsilon\sqrt\lambda/4}}{\epsilon\lambda}\\
		& \le C \mathcal M_{\rho/2}^{3/2}\, \mathcal M_{\rho/2}' \frac{e^{-\epsilon\sqrt\lambda/4}}{\rho\lambda},
	\end{align*}
    where in the last line we used the fact that $\mathcal{M}_\epsilon$ and $\mathcal M_\epsilon'$  are nondecreasing in $\epsilon$ and the choice of $\epsilon$. On the other hand, by definition of $\gamma_{\lambda,\epsilon}$ and a computation of $\partial_\lambda T_{\lambda,\epsilon}(x,x)$ we have
	$$
	\partial_\lambda G_\lambda(x,x) = \partial_\lambda T_{\lambda,\epsilon}(x,x) - \partial_\lambda\gamma_{\lambda,\epsilon}(x,x) = -\frac{1}{8\pi\sqrt{\lambda}} - \partial_\lambda\gamma_{\lambda,\epsilon}(x,x) \,.
	$$
	Combining this identity with the bound on $\partial_\lambda\gamma_{\lambda,\epsilon}$, we obtain the bound in the proposition.
\end{proof}

\begin{remark}\label{rk:3d}
 The above proof is specific to three space dimensions because of the particular form of the resolvent kernel of the Laplacian in three dimensions, 
 $$(-\Delta_{\R^3}+\lambda)^{-1}(x,y)=\frac{e^{-\sqrt{\lambda}|x-y|}}{4\pi|x-y|} \,.$$
 In particular, we use in a crucial way that $e^{\sqrt{\lambda}|x-y|}(-\Delta_{\R^3}+\lambda)^{-1}(x,y)$ is independent of $\lambda$, so that the terms $a^{(n)}$ defined above can be written as exact Stieltjes transforms. This fact is not true anymore in other dimensions: for instance, in dimension five one has
 $$(-\Delta_{\R^5}+\lambda)^{-1}(x,y)=\frac{e^{-\sqrt{\lambda}|x-y|}}{8\pi^2|x-y|^3}\left(1+\sqrt{\lambda}|x-y|\right)$$
 and it is not clear how to deal with the term $1+\sqrt{\lambda}|x-y|$ on the right. In the footnote on p.~328 of \cite{Avakumovic-56}, Avakumovi\'c mentions that in higher dimensions, one should use the heat kernel instead of Green's functions, but without any further details. We were not able to reconstruct his argument in higher dimensions. 
\end{remark}

\begin{remark}[Regularity needed on $g$]\label{rk:regularity-g}
The previous arguments shows that we need uniform bounds on $\Delta U_0$, so, in particular, the assumption $U_0\in C^2$ suffices. By the definition of $U_0$, this means that the exponential map is $C^3$, and since the exponential map may be seen as a time-one flow of an ODE whose vector field depends on the Christoffel symbols (that involve one derivative of $g$), we see that overall this argument requires $g\in C^4$. It is interesting to compare this regularity with the one found in \cite{Donnelly-01,Smith-06} for the (weaker) $L^\ii$ bounds for eigenfunctions or quasimodes, where they only need essentially $C^2$ regularity (which is the minimal regularity needed for geodesics to exist). The additional two derivatives needed here can be explained by the choice of the parametrix. For $L^\ii$ bounds of eigenfunctions, one may replace the function $U_0$ by $1$ in the parametrix $T_{\lambda,\epsilon}$ above, which gives bounds of the same order for $R_{\lambda,\epsilon}$ (and hence no information on $\Delta U_0$ is indeed). We, however, need more than mere bounds for $R_{\lambda,\epsilon}$, we need that $e^{\sqrt{\lambda}d_g(x,y)}R_{\lambda,\epsilon}(x,y)$ is independent of $\lambda$ for $d_g(x,y)<\epsilon/2$, which motivates the introduction of $U_0$ in the parametrix.
\end{remark}


\section{Adding singular potentials}\label{sec:avakv}

The previous construction is quite robust and allows the inclusion of singular potentials. Indeed, let $V:M\to\R$ belonging to the \emph{Kato class} defined as the set of all measurable $V$ such that 
$$
\lim_{r\to0_+} \|V\|_{\cK(r)}=0 \,,
$$
where
$$
\|V\|_{\cK(r)} = \sup_{x\in M}\int_{d_g(x,y)<r}\frac{|V(y)|}{d_g(x,y)}\,dv_g(y) \,.
$$
Then $-\Delta_g+V$ defines a bounded below quadratic form, as we recall below in Appendix \ref{app:op-kato}. Define for $\lambda>-\inf\spec(-\Delta_g+V)$ and $x,y\in M$,
$$G_\lambda^V(x,y):=(-\Delta_g+V+\lambda)^{-1}(x,y) \,.$$

\begin{proposition}\label{prop:greender}
	There are constants $C_1,\ldots,C_6\ge 1$ and for every Kato class $V: M\to\R$ there is a constant $\epsilon_0>0$ such that for all $\epsilon\in(0,\epsilon_0]$ there are functions $R^V_\epsilon,r_{0,\epsilon}^V:(0,\infty)\times M\to\R$ with the following properties. For all $x\in M$, all $\lambda\ge C_1/\epsilon^2$, and all $\epsilon\in(0,\epsilon_0]$
	$$
	\left| \partial_\lambda G_\lambda^V(x,x) + \frac{1}{8\pi\sqrt\lambda} - \int_0^\infty \frac{t^{3/2}\, r_{0,\epsilon}^V(t,x) + R^V_\epsilon(t,x)}{(t+\lambda)^3}\,dt \right|
	\le C_2\, \frac{e^{-\epsilon\sqrt{\lambda}/4}}{\epsilon\lambda},
	$$
	and for all $x\in M$, all $\epsilon\in(0,\epsilon_0]$, all $t>0$,
	$$
	\left| R^V_\epsilon(t,x) \right| \le C_3\,\rho^{-2} \epsilon t \,,
	\qquad
	\left| r_{0,\epsilon}^V(t,x)\right|\le C_4\|V\|_{\cK(\epsilon)}.$$
	Moreover, for all $x\in M$, all $0<t\le t'$, and all $\epsilon\in(0,\epsilon_0]$,
	\begin{equation}
	\label{eq:almostmono1}
	\left| r_{0,\epsilon}^V(t,x) - r_{0,\epsilon}^V(t',x)\right|\le C_5\|V\|_{\cK(\epsilon)}\frac{\sqrt{t'}-\sqrt t}{\sqrt t}.
	\end{equation}
	\begin{equation}
	\label{eq:almostmono2}
	\left|t^{-3/2} R_\epsilon^V(t,x)-(t')^{-3/2} R_\epsilon^V(t',x)\right| \le C_5\,\frac{\epsilon^2}{\rho^2}\,\frac{\sqrt{t'}-\sqrt t}{\sqrt t}.
	\end{equation}
	The function $r_{0,\epsilon}^V$ is explicitly given by \eqref{eq:defr0} and \eqref{eq:r0n}. The constant $\epsilon_0$ satisfies
	\begin{equation*}
	C_6 \, \|V\|_{\mathcal K(\epsilon_0)} \le 1 \,,
	\qquad
	C_6 \, \epsilon_0 \le \rho \,,
	\end{equation*}
	and the constants $C_1,\ldots,C_6$ depend only on $\mathcal M_{\rho/2}$ and $\mathcal M_{\rho/2}'$.
\end{proposition}

Before proving this proposition let us derive the following important consequence which is the technical main result of this paper.

\begin{corollary}\label{coro:weyl-V}
	Let $V:M\to\R$ be in the Kato class and let $\epsilon_0>0$ be as in Proposition \ref{prop:greender}. Then, for all $\epsilon\in(0,\epsilon_0]$ we have, as $\lambda\to\infty$, uniformly in $x\in M$,
	\begin{equation}\label{eq:weyl-corrected}
	\1(-\Delta_g+V\le t)(x,x)=\frac{t^{3/2}}{6\pi^2}- \frac12\, r_{0,\epsilon}^V(t,x)\, t^{3/2}+\cO^{(\epsilon)}(t),\quad t\to+\ii
	\end{equation}
	where the $\cO^{(\epsilon)}$ is uniform in $x\in M$ (but depends on $\epsilon$). More precisely, we have for all $x\in M$ and all $t\ge 0$,
	$$
	\left| \1(-\Delta_g+V\le t)(x,x) - \frac{t^{3/2}}{6\pi^2} + \frac12\, r_{0,\epsilon}^V(t,x)\, t^{3/2} \right| \le C_0\, \epsilon^{-1} \left( t + \epsilon^{-2} \right),
	$$
	where $C_0$ depends only on upper bounds on $\mathcal M_{\rho/2}$ and $\mathcal M_{\rho/2}'$.
\end{corollary}

\begin{proof}[Proof of Corollary \ref{coro:weyl-V} assuming Proposition \ref{prop:greender}]
	Define for all $t\ge0$
	$$A(t,x) :=\1(-\Delta_g+V\le t)(x,x)-t^{3/2}/(6\pi^2)+r_{0,\epsilon}^V(t,x)\,t^{3/2}/2 + R^V_\epsilon(t,x)/2 \,.$$
	As in the proof of Corollary \ref{coro:weyl-3D}, by the spectral theorem for the operator $-\Delta_g +V$ we find
	$$\int_0^\ii\frac{\1(-\Delta_g+V\le t)(x,x)}{(\lambda+t)^3}\,dt= -\frac12\,\partial_\lambda G_\lambda^V(x,x) \,.$$
	Continuing to argue as in that proof, we have by Proposition \ref{prop:greender},
	\begin{align*}
	\left|\int_0^\ii\frac{A(t,x)}{(t+\lambda)^3}\,dt\right| &= \left|-\frac12\partial_\lambda G_\lambda(x,x)-\frac{1}{16\pi\sqrt{\lambda}}+\frac12\int_0^\ii\frac{t^{3/2}\, r_{0,\epsilon}^V(t,x) + R^V_\epsilon(t,x)}{(t+\lambda)^3}\,dt\right| \\
	&\le C_2\frac{e^{-\epsilon\sqrt{\lambda}/4}}{\epsilon\lambda}	 
	\end{align*}
	for all $\lambda\ge C_1/\epsilon^2$. The result then follows from Theorem \ref{thm:stieljes-3D} in the simplified form \eqref{eq:stieltjessimple} with $B_0\equiv  -(6\pi)^{-1}+r_{0,\epsilon}^V(t,x)/2$, $B_1(t)=\1(-\Delta_g+V\le t)(x,x)$ and $B_2(t)= R^V_\epsilon(t,x)/2$. Note that the assumptions on $B_0$ and $B_2$ required for Theorem \ref{thm:stieljes-3D} are satisfied by the bounds in Proposition \ref{prop:greender}. For instance, the bounds \eqref{eq:almostmono1} and \eqref{eq:almostmono2} imply that $u\mapsto u^3 r_{0,\epsilon}^V(u^2,x)$ and $u\mapsto R_\epsilon^V(u^2,x)$ are locally Lipschitz on $[0,\infty)$, and hence continuous and locally of bounded variation on $[0,\infty)$. Moreover, inequality \eqref{eq:almostmono1} implies that the almost monotonicity assumption on $B_0$ in Theorem \ref{thm:stieljes-3D} is satisfied for any $\delta\ge 0$ with $C_{B_0} \le C_5 \|V\|_{\mathcal K(\epsilon)} \delta$. We apply Theorem \ref{thm:stieljes-3D} for any $\epsilon\in(0,\epsilon_0]$ with the choice $\delta = \epsilon^{-1}$. In view of the bounds in Proposition \ref{prop:greender} this gives
	$$
	\left| \1(-\Delta_g+V\le t)(x,x) - \frac{t^{3/2}}{6\pi^2} + \frac12\, r_{0,\epsilon}^V(t,x)\, t^{3/2} \right| \le C_0\, \left( 1 + \epsilon^3 B_1(0) \right) \epsilon^{-1} \left( t + \epsilon^{-2} \right),
	$$
	where $C_0$ depends only on upper bounds on $\mathcal M_{\rho/2}$ and $\mathcal M_{\rho/2}'$. To complete the proof of the corollary, it suffices to show that
	\begin{equation}
	\label{eq:boundnegefs}
	B_1(0) = \1(-\Delta + V \le 0)(x,x) \le C \epsilon_0^{-3} \,.
	\end{equation}
	Since each negative eigenvalue $E_n$ of $-\Delta_g +V$ satisfies $-\lambda_0<E_n\le0$ with $\lambda_0:= 2 C_1/\epsilon_0^2$ (indeed, since the resolvent was shown to exist for all $\lambda\ge C_1/\epsilon_0^2$, this means in particular that the spectrum of $-\Delta_g+V$ is above $-C_1/\epsilon_0^2$), we have
	$$
	\lambda_0^{-2}\, \1(-\Delta + V \le 0)(x,x) \le (-\Delta_g +V+\lambda_0)^{-2}(x,x) = - \partial_\lambda G_{\lambda_0}^V(x,x) \,.
	$$
	On the other hand, Proposition \ref{prop:greender} implies that
	$$
	- \partial_\lambda G_{\lambda_0}^V(x,x)  \le \frac{1}{8\pi\sqrt{\lambda_0}} + \int_0^\infty \frac{t^{3/2} C_4 \|V\|_{\mathcal K(\epsilon)} + C_3 \rho^{-2} \epsilon_0 t}{(t+\lambda_0)^3}\,dt + C_2 \frac{e^{-\epsilon_0\sqrt{\lambda_0}/4}}{\epsilon_0\lambda_0} \le C \lambda_0^{-1/2}
	$$
	with a constant $C$ depending only on upper bounds on $\mathcal M_{\rho/2}$ and $\mathcal M_{\rho/2}'$. Multiplying by $\lambda_0^2$ and recalling the definition of $\lambda_0$, we obtain \eqref{eq:boundnegefs}. This completes the proof of the corollary.
\end{proof}

\subsection{Proof of Proposition \ref{prop:greender}}

We will see that the free parametrix $T_{\lambda,\epsilon}$ is still a good parametrix for $G_\lambda^V$. With $T_{\lambda,\epsilon}$ and $R_{\lambda,\epsilon}$ defined in the previous section, depending on a parameter $\epsilon\in(0,\rho]$ to be determined, we set
$$
\gamma_{\lambda,\epsilon}^V:=T_{\lambda,\epsilon}-G_\lambda^V
$$
and
$$
R_{\lambda,\epsilon}^V:= (-\Delta_g+V+\lambda)\gamma_{\lambda,\epsilon}^V=R_{\lambda,\epsilon}+VT_{\lambda,\epsilon} \,.
$$
We thus have the integral representation
$$
\gamma_{\lambda,\epsilon}^V(x,y)=\int_M T_{\lambda,\epsilon}(x,z) R_{\lambda,\epsilon}^V(z,y)\,dv_g(z)-\int_M\gamma_{\lambda,\epsilon}^V(x,z) R_{\lambda,\epsilon}^V(z,y)\,dv_g(z)
$$
and, at least formally, the series representation
$$
\gamma_{\lambda,\epsilon}^V(x,y)=\sum_{n\ge1}(-1)^{n+1}\gamma_{\lambda,\epsilon}^{(n,V)}(x,y)
$$
with
$$
\gamma_{\lambda,\epsilon}^{(n,V)}(x,y):=\int_M\,dv_g(z_1)\cdots\int_M\,dv_g(z_n)\, T_{\lambda,\epsilon}(x,z_1)R_{\lambda,\epsilon}^V(z_1,z_2)\cdots R_{\lambda,\epsilon}^V(z_n,y) \,.
$$

When trying to prove pointwise bounds on $\gamma_{\lambda,\epsilon}^{(n,V)}$ by the same method as in the previous section, one runs into the problem that, while there are sufficiently good bounds on $\|VT_{\lambda,\epsilon}\|_{L^\ii_y L^1_x}$, see \eqref{eq:tvbounds} below, the norm $\|VT_{\lambda,\epsilon}\|_{L^\ii_y L^2_x}$ might be infinite for Kato class potentials. One can, however, obtain integrated bounds since $T_{\lambda,\epsilon}\in L^\ii_y L^2_x$, showing that the series defining $\gamma_{\lambda,\epsilon}^V$ converges in $L^\ii_y L^2_x$ meaning that $\gamma_{\lambda,\epsilon}^V$ maps $L^2$ to $L^\infty$ and hence $L^2$ to $L^2$. Again by uniqueness of the resolvent, this shows the identity $G_\lambda^V=T_{\lambda,\epsilon}-\gamma_{\lambda,\epsilon}^V$, where $\gamma_{\lambda,\epsilon}^V$ is defined as the sum of the above convergent series.

The way around the impasse of obtaining a pointwise convergent series is to extract the most singular term from $\gamma_{\lambda,\epsilon}^{(n,V)}$. Note that, since $R_{\lambda,\epsilon}^V$ is a sum of two terms, namely, $R_{\lambda,\epsilon}$ and $VT_{\lambda,\epsilon}$, the quantity $\gamma_{\lambda,\epsilon}^{(n,V)}$ can be written as a sum of $2^n$ terms. The most singular term is the one where all factors are $VT_{\lambda,\epsilon}$, that is,
$$
\sigma_{\lambda,\epsilon}^{(n,V)}(x,y): =\int_M\,dv_g(z_1)\cdots\int_M\,dv_g(z_n)\, T_{\lambda,\epsilon}(x,z_1)V(z_1)T_{\lambda,\epsilon}(z_1,z_2)\cdots V(z_n)T_{\lambda,\epsilon}(z_n,y) \,.
$$

Before discussing this singular term, let us derive bounds on the difference $\gamma_{\lambda,\epsilon}^{(n,V)}-\sigma_{\lambda,\epsilon}^{(n,V)}$ that show, in particular, that
the series $\sum_{n\ge 1} (-1)^{n+1} \left( \gamma_{\lambda,\epsilon}^{(n,V)} - \sigma_{\lambda,\epsilon}^{(n,V)} \right) $ converges. We write symbolically (identifying the kernels with operators)
$$
\gamma_{\lambda,\epsilon}^{(n,V)} -\sigma_{\lambda,\epsilon}^{(n,V)} = \sum_{k=1}^n T_{\lambda,\epsilon} (VT_{\lambda,\epsilon})^{k-1} R_{\lambda,\epsilon} (R_{\lambda,\epsilon}^V)^{n-k} = \sum_{k=1}^n (T_{\lambda,\epsilon} V)^{k-1} T_{\lambda,\epsilon} R_{\lambda,\epsilon} (R_{\lambda,\epsilon}^V)^{n-k} \,.
$$
Moreover, by using the Cauchy--Schwarz inequality in the integral connecting $T_{\lambda,\epsilon}$ and $R_{\lambda,\epsilon}$ we find that
\begin{eqnarray}
\label{eq:boundmiddle}
\left\| (T_{\lambda,\epsilon} V)^{k-1} T_{\lambda,\epsilon} R_{\lambda,\epsilon} (R_{\lambda,\epsilon}^V)^{n-k} \right\|_{L^\infty_x L^\infty_y} \le \|T_{\lambda,\epsilon} V\|_{L^\infty_x L^1_y}^{k-1} \|T_{\lambda,\epsilon}\|_{L^\infty_x L^2_y} \|R_{\lambda,\epsilon}\|_{L^\infty_y L^2_x} \|R_{\lambda,\epsilon}^V\|_{L^\infty_y L^1_x}^{n-k} \,.
\end{eqnarray}
Combining this inequality with the bounds \eqref{eq:rbounds}, \eqref{eq:tbounds} and the obvious bounds
\begin{eqnarray}
\label{eq:tvbounds}
\|VT_{\lambda,\epsilon}\|_{L^\ii_y L^1_x}\le C \mathcal M_\epsilon' \, \|V\|_{\cK(\epsilon)} 
\qquad\text{and}\qquad
\|T_{\lambda,\epsilon} V\|_{L^\ii_x L^1_y}\le C \mathcal M_\epsilon' \, \|V\|_{\cK(\epsilon)},
\end{eqnarray}
we obtain
$$
\left| \gamma_{\lambda,\epsilon}^{(n,V)}(x,y) - \sigma_{\lambda,\epsilon}^{(n,V)}(x,y) \right| \le C\, \frac{\mathcal M_\epsilon \mathcal M_\epsilon'}{\epsilon^2\lambda^{1/2}} \left( A \left( \mathcal M_\epsilon' \| V \|_{\cK(\epsilon)} + \frac{\mathcal M_\epsilon}{\epsilon^2\lambda}\right) \right)^{n-1}.
$$
(Similarly as in the previous section, a factor of $n$ can be dropped by enlarging the absolute constant $A$.) This is the desired bound on $\gamma_{\lambda,\epsilon}^{(n,V)}(x,y) - \sigma_{\lambda,\epsilon}^{(n,V)}(x,y)$.

To proceed, we decompose
\begin{equation}
\label{eq:decompv}
\gamma_{\lambda,\epsilon}^{(n,V)}(x,y)-\sigma_{\lambda,\epsilon}^{(n,V)}(x,y)=a^{(n,V)}_{\lambda,\epsilon}(x,y)+b^{(n,V)}_{\lambda,\epsilon}(x,y) \,,
\end{equation}
where $a_{\lambda,\epsilon}^{(n,V)}$ and $b_{\lambda,\epsilon}^{(n,V)}$ are defined in the same way as in the previous sections, namely, in $a_{\lambda,\epsilon}^{(n,V)}$ all $n$ integrations with respect to the variables $z_k$, $k=1,\ldots,n$, are restricted to $d_g(z_k,z_{k+1})<\epsilon/2$, with $z_{n+1}=y$. The following two lemmas are the analogues of Lemmas \ref{lem:main} and \ref{lem:remainder}.

\begin{lemma}\label{lem:mainv}
	There are $A,C>0$ and for every $n\ge 1$ and $\epsilon\in(0,\rho]$ there is a function $r^{(n,V)}_\epsilon:(0,\infty)\times M\to\R$ such that for all $x\in M$
	$$\partial_\lambda a^{(n,V)}_{\lambda,\epsilon}(x,x)= - \int_0^\ii\frac{r^{(n,V)}_\epsilon(t,x)}{(t+\lambda)^3}\,dt
	$$
	and for all $x\in M$, all $t>0$, all $n\ge1$ and all $\epsilon\in(0,\rho]$
	$$
	|r^{(n,V)}_\epsilon(t,x)|\le C \mathcal M_{\rho/2} \mathcal M_{\epsilon}' \,\frac{\epsilon t}{\rho^2} \left( A \left( \mathcal M_\epsilon' \|V\|_{\cK(\epsilon)} + \frac{\mathcal M_{\rho/2}\, \epsilon^2}{\rho^2} \right) \right)^{n-1}.
	$$
\end{lemma}

\begin{proof}
	The proof is similar to that of Lemma \ref{lem:main} and we only sketch the main steps. Decomposing the integrand similarly as in the proof of the bound on $\gamma_{\lambda,\epsilon}^{(n,V)} -\sigma_{\lambda,\epsilon}^{(n,V)}$ and using the same representation formula as in the proof of Lemma \ref{lem:main} we obtain the claimed formula with
	\begin{align*}
	r^{(n,V)}_\epsilon(t,x)& :=\frac{t^{3/2}}{2(4\pi)^{n+1}} \sum_{k=1}^n \int_{d_g(z_n,x)<\epsilon/2}\,dv_g(z_n)\int_{d_g(z_{n-1},z_n)<\epsilon/2}\,dv_g(z_{n-1})\cdots\int_{d_g(z_1,z_2)<\epsilon/2}\,dv_g(z_1) \\
	& \qquad\times (d_g(x,z_1)+\cdots+d_g(z_n,x)) \, \kappa(\sqrt t\,( d_g(x,z_1)+\cdots+d_g(z_n,x))) \\
	& \qquad \times \frac{U_0(x,z_1)\chi(\epsilon^{-1} d_g(x,z_1))}{d_g(x,z_1)}
	\left( \prod_{\ell=1}^{k-1} V(z_\ell) \frac{U_0(z_\ell,z_{\ell+1})}{d_g(z_\ell,z_{\ell+1})} \right) \frac{\Delta U_0(z_k,z_{k+1})}{d_g(z_k,z_{k+1})} \\
	& \qquad\times \left( \prod_{m=k+1}^{n} \frac{\Delta U_0(z_{m},z_{m+1}) + V(z_{m}) U_0(z_{m},z_{m+1})}{d_g(z_m,z_{m+1})}  \right).
	\end{align*}
	Here we use the convention $z_{n+1}=x$. Using \eqref{eq:kappabound1} we obtain
	\begin{align*}
	\left| r^{(n,V)}_\epsilon(t,x)\right| & \le \frac{C t}{2(4\pi)^{n+1}} \sum_{k=1}^n \int_{d_g(z_n,x)<\epsilon/2}\,dv_g(z_n)\int_{d_g(z_{n-1},z_n)<\epsilon/2}\,dv_g(z_{n-1})\cdots\int_{d_g(z_1,z_2)<\epsilon/2}\,dv_g(z_1) \\
	& \qquad \times \frac{U_0(x,z_1)\1(d_g(x,z_1)<\epsilon)}{d_g(x,z_1)}
	\left( \prod_{\ell=1}^{k-1} |V(z_\ell)| \frac{U_0(z_\ell,z_{\ell+1})}{d_g(z_\ell,z_{\ell+1})} \right) \frac{|\Delta U_0(z_k,z_{k+1})|}{d_g(z_k,z_{k+1})} \\
	& \qquad\times \left( \prod_{m=k+1}^{n} \frac{|\Delta U_0(z_{m},z_{m+1})| + |V(z_{m})| U_0(z_{m},z_{m+1})}{d_g(z_m,z_{m+1})}  \right).
	\end{align*}
	Using the bound \eqref{eq:boundmiddle} we can estimate the $k$-th summand on the right side by
	\begin{align*}
	& \left\| \frac{U_0(x,y)\1(d_g(x,y)<\epsilon)}{d_g(x,y)} V(y) \right\|_{L^\infty_x L^1_y}^{k-1} \\
	& \times \left\| \frac{U_0(x,y)\1(d_g(x,y)<\epsilon/2)}{d_g(x,y)} \right\|_{L^\infty_x L^2_y} \frac{\mathcal M_{\rho/2}}{(\rho/2)^2} \left\| \frac{U_0(y,x)^2 \1(d_g(x,y)<\epsilon/2)}{d_g(x,y)} \right\|_{L^\infty_y L^2_x} \\
	& \times 	
	\left( \frac{\mathcal M_{\rho/2}}{(\rho/2)^2} \left\| \frac{U_0(y,x)^2 \1(d_g(x,y)<\epsilon/2)}{d_g(x,y)} \right\|_{L^\infty_y L^1_x} + \left\| V(x) \frac{U_0(x,y)\1(d_g(x,y)<\epsilon/2)}{d_g(x,y)} \right\|_{L^\infty_y L^1_x} \right)^{n-k}.
	\end{align*}
	Using similar bounds as before one sees that this is bounded by a constant times
	\begin{align*}
	\left( A \mathcal M_\epsilon' \|V\|_{\cK(\epsilon)} \right)^{k-1} \sqrt\epsilon \ \frac{\mathcal M_{\rho/2}\cM_\epsilon'\,\sqrt{\epsilon}}{\rho^2} \left( A \left( \frac{\mathcal M_{\rho/2}\,\epsilon^2}{\rho^2} + \mathcal M_\epsilon' \|V\|_{\cK(\epsilon)} \right) \right)^{n-k}.
	\end{align*}
	This leads to the claimed bound.	
\end{proof}

\begin{lemma}\label{lem:bv}
	There are $A,C>0$ such that for all $x\in M$, all $\lambda>0$, all $n\ge 1$, and all $\epsilon\in(0,\rho]$ we have
	$$
	\left| \partial_\lambda b_{\lambda,\epsilon}^{(n,V)}(x,x)\right| \le C \mathcal M_\epsilon' \mathcal M_\epsilon \frac{ e^{-\epsilon \sqrt\lambda/4}}{\epsilon\lambda} \left( A \left( \mathcal M_\epsilon' \|V\|_{\cK(\epsilon)} + \frac{\mathcal M_\epsilon}{\epsilon^2\lambda} \right) \right)^{n-1}.
	$$
\end{lemma}

\begin{proof}
	We recall that the integrand of $\gamma_{\lambda,\epsilon}^{(n,V)}-\sigma_{\lambda,\epsilon}^{(n,V)}$ can be written symbolically as a sum of terms of the form $(T_{\lambda,\epsilon} V)^{k-1} T_{\lambda,\epsilon} R_{\lambda,\epsilon} (R_{\lambda,\epsilon}^V)^{n-k}$ for some $1\le k\le n$, which will be fixed from now on. When taking the $\lambda$ derivative, by the product rule it can fall on any one of the $n+1$ terms. Next, as in the proof of Lemma \ref{lem:remainder}, we first take absolute values of the integral kernels and then multiply them by $\1(d_g(z_\ell,z_{\ell+1})\ge\epsilon/2)$ for some $1\le \ell\le n$, which will also be fixed from now on. Thus, $\partial_\lambda b_{\lambda,\epsilon}^{(n,V)}(x,x)$ can be bounded by a sum of $n^2(n+1)$ terms.
	
	Each of the terms $T_{\lambda,\epsilon}$, $R_{\lambda,\epsilon}$, $\partial_\lambda R_{\lambda,\epsilon}$ and $\partial_\lambda T_{\lambda,\epsilon}$ can be bounded as in the proof of Lemma \ref{lem:remainder} and we see again that the effect of the $\lambda$ derivative can be bounded by multiplying the bound without this derivative by a factor of $\epsilon/\sqrt\lambda$. Our task is thus to bound
	\begin{align*}
	& \frac{\epsilon}{\sqrt{\lambda}} \int_{d_g(z_n,x)<\epsilon}\,dv_g(z_n)\int_{d_g(z_{n-1},z_n)<\epsilon}\,dv_g(z_{n-1})\cdots\int_{d_g(z_1,z_2)<\epsilon}\,dv_g(z_1) \ \1(d_g(z_\ell,z_{\ell+1})\ge\epsilon/2) \\
	& \qquad \times 
	\1(d_g(x,z_1)<\epsilon) \left( \prod_{m=0}^{k-2} \frac{e^{-\sqrt\lambda d_g(z_m,z_{m+1})}}{d_g(z_m,z_{m+1})} U_0(z_m,z_{m+1})|V(z_{m+1})| \right) \\
	& \qquad \times
	\frac{e^{-\sqrt{\lambda}d_g(z_{k-1},z_k)}}{d_g(z_{k-1},z_k)}U_0(z_{k-1},z_k)\frac{\mathcal M_\epsilon}{\epsilon^2}
	\frac{e^{-\frac12\sqrt\lambda d_g(z_k,z_{k+1})}}{d_g(z_k,z_{k+1})} U_0(z_{k+1},z_k)^2 \\
	& \qquad\times
	\prod_{m'=k+1}^n \left( \frac{\mathcal M_\epsilon}{\epsilon^2}
	\frac{e^{-\frac12\sqrt\lambda d_g(z_{m'},z_{m'+1})}}{d_g(z_{m'},z_{m'+1})} U_0(z_{m'+1},z_{m'})^2 + |V(z_{m'})| \frac{e^{-\sqrt\lambda d_g(z_{m'},z_{m'+1})}}{d_g(z_{m'},z_{m'+1})} U_0(z_{m'},z_{m'+1}) \right),
	\end{align*}
	with the convention $z_0=x=z_{n+1}$. As in \eqref{eq:boundmiddle} we apply the Cauchy--Schwarz inequality in the $z_k$ variable and obtain
	\begin{align*}
	& \frac{\epsilon}{\sqrt{\lambda}}\  \mathcal F \ \left\| \frac{e^{-\sqrt\lambda d_g(x,y)}}{d_g(x,y)}
	U_0(x,y)\1(d_g(x,y)<\epsilon) V(y) \right\|_{L^\infty_x L^1_y}^{k-1} \\
	& \times \left\| \frac{e^{-\sqrt\lambda d_g(x,y)}}{d_g(x,y)}
	U_0(x,y)\1(d_g(x,y)<\epsilon) \right\|_{L^\infty_x L^2_y}
	\ \frac{\mathcal M_{\epsilon}}{\epsilon^2} \left\| \frac{e^{-\frac12 \sqrt\lambda d_g(x,y)}}{d_g(x,y)} U_0(y,x)^2 \1(d_g(x,y)<\epsilon) \right\|_{L^\infty_y L^2_x} \\
	& \times 	
	\left( \left\| \frac{\1(d_g(x,y)<\epsilon)}{d_g(x,y)} \left(
	\frac{\mathcal M_{\epsilon}}{\epsilon^2}  e^{-\frac12 \sqrt\lambda d_g(x,y)} U_0(y,x)^2 + |V(x)| U_0(x,y) e^{-\sqrt\lambda d_g(x,y)} \right) \right\|_{L^\infty_y L^1_x} \right)^{n-k}
	\end{align*}
	with
	\begin{align*}
	\mathcal F =
	\begin{cases}
	\left\| \frac{e^{-\sqrt\lambda d_g(x,y)}}{d_g(x,y)}
	U_0(x,y)\1(\epsilon/2\le d_g(x,y)<\epsilon) V(y) \right\|_{L^\infty_x L^1_y} / \left\| \frac{e^{-\sqrt\lambda d_g(x,y)}}{d_g(x,y)}
	U_0(x,y)\1(d_g(x,y)<\epsilon) V(y) \right\|_{L^\infty_x L^1_y} \\
	\qquad\qquad \text{if}\ \ell\le k-2 \,,\\
	\left\| \frac{e^{-\sqrt\lambda d_g(x,y)}}{d_g(x,y)}
	U_0(x,y)\1(\epsilon/2\le d_g(x,y)<\epsilon) \right\|_{L^\infty_x L^2_y} / \left\| \frac{e^{-\sqrt\lambda d_g(x,y)}}{d_g(x,y)}
	U_0(x,y)\1(d_g(x,y)<\epsilon) \right\|_{L^\infty_x L^2_y} \\
	\qquad\qquad \text{if}\ \ell=k-1 \,,\\
	\left\| \frac{e^{-\frac12 \sqrt\lambda d_g(x,y)}}{d_g(x,y)} U_0(y,x)^2 \1(\epsilon/2\le d_g(x,y)<\epsilon) \right\|_{L^\infty_y L^2_x} / \left\| \frac{e^{-\frac12 \sqrt\lambda d_g(x,y)}}{d_g(x,y)} U_0(y,x)^2 \1(d_g(x,y)<\epsilon) \right\|_{L^\infty_y L^2_x} \\
	\qquad\qquad \text{if}\ \ell=k \,,\\
	\left\| \frac{\1(\epsilon/2\le d_g(x,y)<\epsilon)}{d_g(x,y)} \left(
	\frac{\mathcal M_{\epsilon}}{\epsilon^2}  e^{-\frac12 \sqrt\lambda d_g(x,y)} U_0(y,x)^2 + |V(x)| U_0(x,y) e^{-\sqrt\lambda d_g(x,y)} \right) \right\|_{L^\infty_y L^1_x} \\
	\quad / \left\| \frac{\1(d_g(x,y)<\epsilon)}{d_g(x,y)} \left(
	\frac{\mathcal M_{\epsilon}}{\epsilon^2}  e^{-\frac12 \sqrt\lambda d_g(x,y)} U_0(y,x)^2 + |V(x)| U_0(x,y) e^{-\sqrt\lambda d_g(x,y)} \right) \right\|_{L^\infty_y L^1_x} \\
	\qquad\qquad \text{if}\ \ell\ge k+1 \,.
	\end{cases}
	\end{align*}
	Using similar bounds as before we obtain the upper bound
	\begin{align*}
	& C \frac{\epsilon}{\sqrt{\lambda}}\  \mathcal F' \ \left( A \mathcal M_\epsilon' \|V\|_{\cK(\epsilon)} \right)^{k-1} \frac{1}{\lambda^{1/4}} \ \frac{\mathcal M_\epsilon\,\mathcal M_\epsilon'}{\epsilon^2 \lambda^{1/4}}
	\left( A \left( \mathcal M_\epsilon' \|V\|_{\cK(\epsilon)} + \frac{\mathcal M_\epsilon}{\epsilon^2 \lambda} \right) \right)^{n-k}
	\end{align*}
	with
	\begin{align*}
	\mathcal F' =
	\begin{cases}
	\left\| \frac{e^{-\sqrt\lambda d_g(x,y)}}{d_g(x,y)}
	U_0(x,y)\1(\epsilon/2\le d_g(x,y)<\epsilon) V(y) \right\|_{L^\infty_x L^1_y} / ( \mathcal M_\epsilon' \|V\|_{\cK(\epsilon_\lambda)}) \\
	\qquad\qquad \text{if}\ \ell\le k-2 \,,\\
	\left\| \frac{e^{-\sqrt\lambda d_g(x,y)}}{d_g(x,y)}
	U_0(x,y)\1(\epsilon/2\le d_g(x,y)<\epsilon) \right\|_{L^\infty_x L^2_y} / \lambda^{-1/4} \\
	\qquad\qquad \text{if}\ \ell=k-1 \,,\\
	\left\| \frac{e^{-\frac12 \sqrt\lambda d_g(x,y)}}{d_g(x,y)} U_0(y,x)^2 \1(\epsilon/2\le d_g(x,y)<\epsilon) \right\|_{L^\infty_y L^2_x} / (\mathcal \, \mathcal M_\epsilon'\, \lambda^{-1/4}) \\
	\qquad\qquad \text{if}\ \ell=k \,,\\
	\left\| \frac{\1(\epsilon/2\le  d_g(x,y)<\epsilon)}{d_g(x,y)} \left(
	\frac{\mathcal M_{\epsilon}}{\epsilon^2}  e^{-\frac12 \sqrt\lambda d_g(x,y)} U_0(y,x)^2 + |V(x)| U_0(x,y) e^{-\sqrt\lambda d_g(x,y)} \right) \right\|_{L^\infty_y L^1_x} \\
	\quad / \left( \mathcal M_\epsilon' \|V\|_{\cK(\epsilon_\lambda)} + \frac{\mathcal M_\epsilon}{\epsilon^2 \lambda} \right) \\
	\qquad\qquad \text{if}\ \ell\ge k+1 \,.
	\end{cases}
	\end{align*}
	Since $\mathcal F'\le C e^{-\epsilon\sqrt\lambda/4}$, we obtain the claimed bound.
\end{proof}

This concludes our discussion of the terms $a_{\lambda,\epsilon}^{(n,V)}(x,x)$ and $b_{\lambda,\epsilon}^{(n,V)}(x,x)$ in the decomposition \eqref{eq:decompv}. We now focus on the term $\sigma^{(n,V)}_{\lambda,\epsilon}(x,x)$. 
Using the representation formula \eqref{eq:repr} with the function $\kappa$ from \eqref{eq:kappa} we can write 
$$\partial_\lambda\sigma^{(n,V)}_{\lambda,\epsilon}(x,x)= - \int_0^\ii\frac{t^{3/2}\, r_{0,\epsilon}^{(n,V)}(t,x)}{(t+\lambda)^3}\,dt$$
with
\begin{multline}
\label{eq:r0n}
r_{0,\epsilon}^{(n,V)}(t,x):= \frac{1}{2(4\pi)^{n+1}}\int\,dv_g(z_n)\int\,dv_g(z_{n-1})\cdots\int\,dv_g(z_1)\\
\times \kappa(\sqrt t( d_g(x,z_1)+\cdots+d_g(z_n,x)))\, \left( d_g(x,z_1)+\cdots+d_g(z_n,x) \right) \\
\times \frac{U_0(x,z_1)\chi(\epsilon^{-1} d_g(x,z_1))}{d_g(x,z_1)} V(z_1) \frac{U_0(z_1,z_2)\chi(\epsilon^{-1} d_g(z_1,z_2))}{d_g(z_1,z_2)} \cdots V(z_n)
\frac{U_0(z_n,x)\chi(\epsilon^{-1} d_g(z_n,x))}{d_g(z_n,x)} \,.
\end{multline}

To bound $r_{0,\epsilon}^{(n,V)}$, we use the Rodnianski--Schlag bound \eqref{eq:rodschl}. A small variation of this argument, combined with the bound
\begin{eqnarray}
\label{eq:kappabound2}
|\kappa(r)|\le C
\qquad\text{for all}\ r>0 \,,
\end{eqnarray}
gives
\begin{equation}
\label{eq:r0sup}
| r_{0,\epsilon}^{(n,V)}(t,x)|\le C\mathcal M_\epsilon'  \left( A \mathcal M_\epsilon' \|V\|_{\cK(\epsilon)} \right)^n.
\end{equation}
Note that this shows that $t^{3/2} r_{0,\epsilon}^{(n,V)}$ is  $\mathcal O(t^{3/2})$, while the other term $r_\epsilon^{(n,V)}$ was shown before to be $\mathcal O(t)$. What makes the term $r_{0,\epsilon}^{(n,V)}$ well-behaved is the following bound on its variation.

\begin{lemma}\label{lem:r0variation}
	There are constant $C$ and $A$ such that for all $0<t\le t'$, all $x\in M$, all $n\ge1$, and all $\epsilon\in(0,\rho]$
	$$
	\left| r_{0,\epsilon}^{(n,V)}(t,x) - r_{0,\epsilon}^{(n,V)}(t',x) \right| \le C\, \frac{\sqrt{t'} - \sqrt t}{\sqrt t}
	\, \mathcal M_\epsilon' \left( A \mathcal M_\epsilon' \|V\|_{\cK(\epsilon)} \right)^n
	$$
	and
	\begin{multline*}
	\left| t^{-3/2} r_{\epsilon}^{(n,V)}(t,x) - (t')^{-3/2} r_{\epsilon}^{(n,V)}(t',x) \right| \\
	\le C\, \frac{\sqrt{t'} - \sqrt t}{\sqrt t}
	\, \mathcal M_{\rho/2}\, \mathcal M_\epsilon' \,\frac{\epsilon^2}{\rho^2} \left( A \left( \mathcal M_\epsilon' \|V\|_{\cK(\epsilon)} + \frac{\mathcal M_{\rho/2}\,\epsilon^2}{\rho^2}\right) \right)^{n-1}.
	\end{multline*}
\end{lemma}

The proof of this lemma relies on the following property of the function $\kappa$ from \eqref{eq:kappa}.

\begin{lemma}\label{lem:g}
	There is a $C>0$ such that for all $v\ge u> 0$ and all $r\ge0$ we have 
	$$|\kappa(rv)-\kappa(ru)|\le\frac{C\,(v-u)}{u}.$$
\end{lemma}

\begin{proof}[Proof of Lemma \ref{lem:g}]
	By the mean value theorem, we have
	$$\kappa(rv)-\kappa(ru)=\kappa'(\xi) r(v-u) $$
	for some $ru\le \xi\le rv$. Using $|\kappa'(t)|\le C \min\{1,t^{-2}\}$ for all $t\ge 0$, we obtain
	$$
	|\kappa'(\xi)| \le C \min\{1,(ru)^{-2}\} \,.
	$$
	Since $\min\{1,(ru)^{-2}\}\le (ru)^{-1}$, we obtain the inequality in the lemma.
\end{proof}

\begin{proof}[Proof of Lemma \ref{lem:r0variation}]
	We bound
	\begin{multline*}
	 \left| r_{0,\epsilon}^{(n,V)}(t,x) - r_{0,\epsilon}^{(n,V)}(t',x) \right| \le \frac{1}{2(4\pi)^{n+1}} \int\,dv_g(z_n)\int\,dv_g(z_{n-1})\cdots\int\,dv_g(z_1)\times \\
	  \times \left| \kappa(\sqrt t( d_g(x,z_1)+\cdots+d_g(z_n,x)))
	- \kappa(\sqrt{t'}( d_g(x,z_1)+\cdots+d_g(z_n,x))) \right|\times\\
	 \times \, \left( d_g(x,z_1)+\cdots+d_g(z_n,x) \right)\times\\
	 \times \frac{U_0(x,z_1)\chi(\epsilon^{-1}d_g(x,z_1))}{d_g(x,z_1)} V(z_1) \frac{U_0(z_1,z_2)\chi(\epsilon^{-1}d_g(z_1,z_2))}{d_g(z_1,z_2)} \cdots V(z_n)
	\frac{U_0(z_n,x)\chi(\epsilon^{-1}d_g(z_n,x))}{d_g(z_n,x)} \,.
	\end{multline*}
	By Lemma \ref{lem:g} we have
	$$
	\left| \kappa(\sqrt t( d_g(x,z_1)+\cdots+d_g(z_n,x)))
	- 	\kappa(\sqrt{t'}( d_g(x,z_1)+\cdots+d_g(z_n,x))) \right|
	\le C \frac{\sqrt{t'} - \sqrt t}{\sqrt t} \,.
	$$
	The first bound in the lemma now follows from the Rodnianski--Schlag bound \eqref{eq:rodschl}. The second bound in the lemma is simpler to prove. We begin similarly as before, using Lemma \ref{lem:g}. This leads to a similar expression as the one we estimated in the proof of Lemma \ref{lem:mainv}, but with an additional factor $t^{-1}\frac{\sqrt{t'} - \sqrt t}{\sqrt t}\left( d_g(x,z_1)+ \cdots d_g(z_n,x)\right)$ under the integral. Because of the restriction in the domain of integration, this factor is bounded by $t^{-1}\frac{\sqrt{t'} - \sqrt t}{\sqrt t}( 1+n/2)\epsilon$, we arrive at the second bound in the lemma.
\end{proof}

Finally, we are ready to give the

\begin{proof}[Proof of Proposition \ref{prop:greender}]
	We denote by $A$ the maximum of the constants with this name appearing in Lemma \ref{lem:mainv}, Lemma \ref{lem:bv} and Lemma \ref{lem:r0variation}. By increasing $A$ even further, we may assume that $A\ge 1$. We set $\epsilon_1 := \rho/\sqrt{4A\mathcal M_{\rho/2}}$ and let $\epsilon_2>0$ be maximal with the property that $\|V\|_{\cK(\epsilon_2)}\le 1/(4A\mathcal M_{\rho/2}')$. In the following we assume that $0<\epsilon\le\epsilon_0:=\min\{\epsilon_1,\epsilon_2\}$. We note that $\epsilon\le\epsilon_1\le\rho/2$ and therefore $\mathcal M_\epsilon'\le \mathcal M_{\rho/2}'$ and $\mathcal M_\epsilon\le \mathcal M_{\rho/2}$. Then, by Lemma \ref{lem:mainv}, the series
	$$
	R^V_\epsilon(t,x) := \sum_{n\ge 1} (-1)^{n+1} r^{(n,V)}_\epsilon(t,x)
	$$
	converges and satisfies
	$$
	\left| R^V_\epsilon(t,x)\right| \le C \mathcal M_{\rho/2} \mathcal M_{\rho/2}' \frac{\epsilon\, t}{\rho^2} \,.
	$$
	Moreover, also the series
	\begin{equation}\label{eq:defr0}
	r_{0,\epsilon}^V(t,x) := \sum_{n\ge 1} (-1)^{n+1} r_{0,\epsilon}^{(n,V)}(t,x)
	\end{equation}
	converges and satisfies, by \eqref{eq:r0sup},
	$$
	\left|r_{0,\epsilon}^V(t,x)\right| \le  C\,  ( \mathcal M_{\rho/2}' )^2\, \|V\|_{\cK(\epsilon)}
	$$
	as well as, by Lemma \ref{lem:r0variation}, if $0<t\le t'$,
	$$
	\left| r_{0,\epsilon}^V(t,x) - r_{0,\epsilon}^V(t',x) \right| \le C\, \frac{\sqrt{t'}-\sqrt t}{\sqrt t}\, ( \mathcal M_{\rho/2}' )^2\, \|V\|_{\cK(\epsilon)}
	$$
	and
	$$
	\left|t^{-3/2} R_\epsilon^V(t,x)-(t')^{-3/2} R_\epsilon^V(t',x)\right| \le C\, \frac{\sqrt{t'}-\sqrt t}{\sqrt t}\, \mathcal M_{\rho/2} \mathcal M_\epsilon' \,\frac{\epsilon^2}{\rho^2} \,.
	$$
	Provided that $\lambda \ge 4A\mathcal M_{\rho/2}/\epsilon^2$, we obtain from Lemma \ref{lem:bv}
	\begin{align*}
	& \left| \partial_\lambda \gamma_{\lambda,\epsilon}^{V}(t,x) + \int_0^\infty \frac{R^V_\epsilon(t,x) + t^{3/2} r_{0,\epsilon}^V(t,x)}{(t+\lambda)^3}\,dt \right| \\
	& \qquad = \left| \sum_{n\ge 1} (-1)^{n+1} \left( \partial_\lambda \gamma_{\lambda,\epsilon}^{(n,V)}(x,x) - \partial_\lambda \sigma_{\lambda,\epsilon}^{(n,V)}(x,x)-
	\partial_\lambda a_{\lambda,\epsilon}^{(n,V)}(x,x) \right) \right| \\
	& \qquad = \left| \sum_{n\ge 1} (-1)^{n+1} \partial_\lambda b_{\lambda,\epsilon}^{(n,V)}(x,x) \right| \\
	& \qquad \le C \mathcal M_{\epsilon}' \mathcal M_{\epsilon} \frac{e^{-\epsilon\sqrt\lambda/4}}{\epsilon\lambda} \,.
	\end{align*}
	Combining this bound with the identity
	$$
	\partial_\lambda G_\lambda^V(x,x) = \partial_\lambda T_{\lambda,\epsilon}^V(x,x) - \partial_\lambda\gamma_{\lambda,\epsilon}^V(x,x) = - \frac{1}{8\pi\sqrt\lambda} - \partial_\lambda\gamma_{\lambda,\epsilon}^V(x,x) \,,
	$$
	we obtain the bound in the proposition.	
\end{proof}


\section{Consequences and discussion of Corollary \ref{coro:weyl-V}}\label{sec:weylproof}

In this section we show that the main results stated in the introduction are consequences of Corollary \ref{coro:weyl-V}.

\subsection{Pointwise Weyl laws}

We begin with the proof of the pointwise Weyl law for Kato class potentials.

\begin{proof}[Proof of Theorem \ref{thm:weylptw}]
	By Corollary \ref{coro:weyl-V} and Proposition \ref{prop:greender} we have for all $\epsilon\in(0,\epsilon_0)$
	$$
	\limsup_{t\to\infty} t^{-3/2} \sup_{x\in M} \left| \1(-\Delta_g +V\le t)(x,x) - \frac{t^{3/2}}{6\pi^2} \right| \le \frac12 \limsup_{t\to\infty} \sup_{x\in M} \left| r_{0,\epsilon}^V(t,x) \right|\le C\|V\|_{\cK(\epsilon)},
	$$
	with $C$ independent of $\epsilon$. By definition of the Kato class, this implies the theorem by letting $\epsilon\to0$.
\end{proof}

We next state and prove a pointwise spectral cluster bound, valid for Kato class potentials. This bound appears in \cite{BlaSirSog-19} under the additional assumption that $V\in L^{3/2}(M)$.

\begin{theorem}[Spectral cluster bounds]\label{thm:cluster}
	Let $V:M\to\R$ be in the Kato class. Then, uniformly in $x\in M$ and $t,T>0$,
	$$
	\left| \1\left(t<-\Delta_g +V\le (\sqrt t +\sqrt T)^2\right)(x,x) \right| \le C \left( t T^{1/2} + t +T^{3/2} + 1 \right).
	$$
\end{theorem}

\begin{proof}
	According to Corollary \ref{coro:weyl-V}, we have, uniformly in $x\in M$ and for a fixed $\epsilon\in(0,\epsilon_0)$
	\begin{align*}
	& \1\left(t<-\Delta_g +V\le (\sqrt t +\sqrt T)^2\right)(x,x) \\
	& = \1\left(-\Delta_g +V\le (\sqrt t +\sqrt T)^2\right)(x,x) - \1\left(-\Delta_g +V\le t\right)(x,x) \\
	& = \frac{(\sqrt t + \sqrt T)^3 - t^{3/2}}{6\pi^2} + \frac12 \left( r_{0,\epsilon}^V((\sqrt t+\sqrt T)^2,x) (\sqrt t+\sqrt T)^3 - r_{0,\epsilon}^V(t,x) t^{3/2} \right) + \mathcal O(t +T) \\
	& = \frac12 r_{0,\epsilon}^V((\sqrt t+\sqrt T)^2,x) \left( (\sqrt t+\sqrt T)^3 - t^{3/2} \right) + \frac12 \left( r_{0,\epsilon}^V((\sqrt t+\sqrt T)^2,x) - r_{0,\epsilon}^V(t,x) \right) t^{3/2} \\
	& \qquad + \mathcal O(t \sqrt T + T^{3/2} + t + T)	\,.
	\end{align*}
	According to Proposition \ref{prop:greender}, we have, uniformly in $x\in M$,
	$$
	r_{0,\epsilon}^V((\sqrt t+\sqrt T)^2,x) \left( (\sqrt t+\sqrt T)^3 - t^{3/2} \right) = \mathcal O(t\sqrt T + T^{3/2})
	$$
	and
	$$
	r_{0,\epsilon}^V((\sqrt t+\sqrt T)^2,x) - r_{0,\epsilon}^V(t,x) = \mathcal O( t^{-1/2} T^{1/2}) 
	$$
	Inserting this into the above bound, we obtain the assertion of the theorem.
\end{proof}

Next, we prove the sharp pointwise version of Weyl's law under an additional regularity assumption.

\begin{proof}[Proof of Theorem \ref{thm:weylptwsharp}]
	According to Corollary \ref{coro:weyl-V} it suffices to prove that
	$$
	\sup_{x\in M} \left| r_{0,\epsilon}^V(t,x) \right| \le C t^{-1/2} \,.
	$$
	We shall show that this holds for a certain choice of the parameter $\epsilon$ in the definition of $r_{0,\epsilon}^V$.
	
	Applying the bound \eqref{eq:kappabound1} in the definition \eqref{eq:r0n} we obtain
	\begin{align*}
	& \left| r_{0,\epsilon}^{(n,V)}(t,x) \right|\le  \frac{C\, t^{-1/2}}{2(4\pi)^{n+1}}\int dv_g(z_n) \cdots\int dv_g(z_1)\\
	& \times \frac{U_0(x,z_1)|\chi(\epsilon^{-1} d_g(x,z_1))|}{d_g(x,z_1)} |V(z_1)| \frac{U_0(z_1,z_2)|\chi(\epsilon^{-1} d_g(z_1,z_2))|}{d_g(z_1,z_2)} \cdots |V(z_n)|
	\frac{U_0(z_n,x)|\chi(\epsilon^{-1} d_g(z_n,x))|}{d_g(z_n,x)} \\
	& \le C t^{-1/2} A^n (\mathcal M_\epsilon')^{n+1} \int\,dv_g(z_n)\int\,dv_g(z_{n-1})\cdots\int\,dv_g(z_1)\\
	& \times \frac{\1( d_g(x,z_1)<\epsilon)}{d_g(x,z_1)} |V(z_1)| \frac{\1( d_g(z_1,z_2)<\epsilon)}{d_g(z_1,z_2)} \cdots |V(z_n)|
	\frac{\1( d_g(z_n,x)<\epsilon)}{d_g(z_n,x)} \,.
	\end{align*}
	We write $z_0=z_{n+1}=x$ and
	$$
	\prod_{k=0}^n \frac{1}{d_g(z_k,z_{k+1})} = \prod_{\ell=0}^n \frac{d_g(z_\ell,z_{\ell+1})^{1/n}}{\prod_{k=0}^n d_g(z_k,z_{k+1})^{1/n}} \,,
	$$
	so by H\"older's inequality, the above $n$-fold integral is bounded by
	\begin{align*}
	& \prod_{\ell=0}^n \left( \int dv_g(z_n)\cdots\int dv_g(z_1) 
	\frac{d_g(z_\ell,z_{\ell+1})^{(n+1)/n}}{\prod_{k=0}^n d_g(z_k,z_{k+1})^{(n+1)/n}}\, 
	\prod_{m=0}^n \1( d_g(z_m,z_{m+1})<\epsilon) |V(z_1)|\cdots |V(z_n)| \right)^{1/(n+1)} \\
	& \le \left( \sup_{x\in M} \int dv_g(y) \frac{\1(d_g(x,y)<\epsilon)}{d_g(x,y)^{(n+1)/n}} |V(y)| \right)^n \\
	& \le \left( \sup_{x\in M} \int dv_g(y) \frac{\1(d_g(x,y)<\epsilon)}{d_g(x,y)^2} |V(y)|
	\right) \left( \sup_{x\in M} \int dv_g(y) \frac{\1(d_g(x,y)<\epsilon)}{d_g(x,y)} |V(y)| \right)^{n-1}.
	\end{align*}
	Thus, we have shown that
	$$
	\left| r_{0,\epsilon}^{(n,V)}(t,x) \right| \le C t^{-1/2} (\mathcal M_\epsilon')^2 \left( \sup_{x\in M} \int dv_g(y) \frac{\1(d_g(x,y)<\epsilon)}{d_g(x,y)^2} |V(y)| \right)
	\left( A \mathcal M_\epsilon' \|V\|_{\cK(\epsilon)} \right)^{n-1}.
	$$
	We now choose $\epsilon\le\epsilon'$ to have the supremum on the right finite and $\epsilon>0$ so small that $\|V\|_{\cK(\epsilon)}\le 1/(2A \mathcal M_\epsilon')$ to have exponential decay in $n$. By \eqref{eq:defr0} we obtain, after summation over $n$,
	$$
	\left| r_{0,\epsilon}^{V}(t,x)\right| \le C t^{-1/2} (\mathcal M_\epsilon')^2 \left( \sup_{x\in M} \int dv_g(y) \frac{\1(d_g(x,y)<\epsilon)}{d_g(x,y)^2} |V(y)| \right).
	$$
	This proves the theorem.	
\end{proof}

By H\"older's inequality the assumption in Theorem \ref{thm:weylptwsharp} is satisfied if $V\in L^q$ for some $q>3$. Thus, we obtain

\begin{corollary}\label{cor:weylptwsharp}
	Assume that $V\in L^q(M)$ for some $q>3$. Then, uniformly in $x\in M$,
	$$
	\1(-\Delta_g +V\le t)(x,x) = \frac{t^{3/2}}{6\pi^2} + \mathcal O(t) \,.
	$$
\end{corollary}


\subsection{Violation of the sharp pointwise Weyl law}

While Corollary \ref{cor:weylptwsharp} shows that under the assumption $V\in L^q$ for some $q>3$ the pointwise Weyl law holds with the same remainder bound $\mathcal O(t)$ as in the case without $V$, we now show that there are potentials $V\in L^q(M)$ for some $3/2<q\le 3$ for which this remainder estimate is violated. These examples also show that for Kato class potentials the remainder $o(t^{3/2})$ in Theorem \ref{thm:weylptw} cannot be improved to $\mathcal O(t^{3/2-\epsilon})$ for some $\epsilon>0$ valid for all Kato class $V$.

Note that if $V\in L^q$ for some $q>3/2$, then $V$ is in the Kato class, so Corollary \ref{coro:weyl-V} is applicable.  Therefore our strategy to proving a violation of the $\mathcal O(t)$ bound will be to show that the bound $r_{0,\epsilon}^V(t,x)=\mathcal O(t^{-1/2})$ can be violated. Recall that $r_{0,\epsilon}^V$ is defined in \eqref{eq:defr0} as an infinite sum of terms $r_{0,\epsilon}^{(n,V)}$. As a first step towards showing the violation of an $\mathcal O(t^{-1/2})$ bound on $r_{0,\epsilon}^V$, we show that such a violation can come from at most a finite number of terms $r_{0,\epsilon}^{(n,V)}$.

\begin{lemma}\label{lem:Lq-potential}
 Let $V\in L^q(M)$ for some $q>3/2$ and set
 $$
 N_q : = 1+ \left\lfloor \frac{q}{2q-3}\right \rfloor \,.
 $$
 Then, uniformly in $x\in M$, 
 $$\sum_{n\ge N_q} |r^{(n,V)}_0(t,x)|=\cO(t^{-1/2})\,,$$
 provided $\epsilon$ in the definition of $r_{0,\epsilon}^{(n,V)}$ is chosen small enough.
\end{lemma}

As a consequence, there is only a finite number of additional terms in the corrected Weyl law \eqref{eq:weyl-corrected}. This number of terms grows as the singularity of the potential gets worse. Note that if $q>3$, then $N_q=1$ and the statement of the lemma is a consequence of the proof of Theorem \ref{thm:weylptwsharp}.

\begin{proof}
During the proof of Theorem \ref{thm:weylptwsharp} we have shown the bound
$$
\left| r_{0,\epsilon}^{(n,V)}(t,x) \right| \le C t^{-1/2} \mathcal M_\epsilon' \left( A \mathcal M_\epsilon' \sup_{x\in M} \int dv_g(y) \frac{\1(d_g(x,y)<\epsilon)}{d_g(x,y)^{(n+1)/n}} |V(y)| \right)^n.
$$
We bound
$$
\int dv_g(y) \frac{\1(d_g(x,y)<\epsilon)}{d_g(x,y)^{(n+1)/n}} |V(y)|
\le \|V\|_{L^q} \left( \int dv_g(y)\, \frac{\1(d_g(x,y)<\epsilon)}{d_g(x,y)^{(n+1)q'/n}} \right)^{1/q'}.
$$
Under the assumption $(n+1)q'/n<3$, which is the same as $n\ge N_q$, we have
\begin{align*}
\int dv_g(y)\, \frac{\1(d_g(x,y)<\epsilon)}{d_g(x,y)^{(n+1)q'/n}}
&\le \sup_{d_g(y,x)<\epsilon} U_0(x,y)^{-2} \int_{|z|<\epsilon} \frac{dz}{|z|^{(n+1)q'/n}} \\
&\le C' \epsilon^{3-(n+1)q'/n} \sup_{d_g(y,x)<\epsilon} U_0(x,y)^{-2}.
\end{align*}
Note that the constant $C'$ here can be chosen independently of $n$. Thus,
$$
\left| r_{0,\epsilon}^{(n,V)}(t,x) \right| \le C t^{-1/2} \epsilon^{-1} \mathcal M_\epsilon' \left( A \mathcal M_\epsilon' \|V\|_{L^q} (C')^{1/q'} \epsilon^{3/q'-1} \sup_{d_g(y,x)<\epsilon} U_0(x,y)^{-2/q'} \right)^n.
$$
Note that $3/q'-1 = 2-3/q>0$. Therefore, by choosing $\epsilon$ small enough, the bound is exponentially decaying in $n$ and we obtain the assertion in the lemma.
 \end{proof}

 The following example shows that the additional terms $r_{0,\epsilon}^{(n,V)}$ can cause a violation of the sharp pointwise Weyl law.

\begin{proposition}\label{lem:example-potential}
 Let $\eta\in(0,1)$, $x_0\in M$, $\gamma\in\R$ and let $V:M\to\R$ defined by 
 $$V(x)=\gamma\, \frac{\chi(d_g(x,x_0))}{d_g(x,x_0)^{2-\eta}} \,,\qquad x\in M.$$
 Then
 $$
 \1(-\Delta_g + V\le t)(x,x) = \frac{t^{3/2}}{6\pi^2} - \frac12 \sum_{1\le n<1/\eta} (-1)^{n+1} r_{0,\epsilon}^{(n,V)}(t,x)\,t^{3/2} + \mathcal O(t) \,.
 $$
 Moreover, for all $n<2/\eta$ we have
 \begin{equation}\label{eq:example-potential-bound}
 \sup_{x\in M} |r^{(n,V)}_{0,\epsilon}(t,x)| \le C t^{-n\eta/2}.
 \end{equation}
 and, in geodesic normal coordinates around $x_0$, for all $y\in T_{x_0}M\sim\R^3$,
 \begin{multline}\label{eq:example-potential-lim}
    \lim_{t\to+\ii} t^{n\eta/2} r^{(n,V)}_{0,\epsilon}(t,y/\sqrt{t}) = \frac{\gamma^n}{2(4\pi)^{n+1}}\int_{\R^3}dz_1\cdots\int_{\R^3}dz_n \\
    \times\frac{\kappa(|z_1|+|z_1-z_2|+\cdots+|z_{n-1}-z_n|+|z_n|)\, (|z_1|+|z_1-z_2|+\cdots+|z_{n-1}-z_n|+|z_n|)}{|z_1||z_1-z_2|\cdots|z_{n-1}-z_n||z_n||z_1+y|^{2-\eta}\cdots|z_n+y|^{2-\eta}} \,,
 \end{multline}
 where the right side is finite.
\end{proposition}

After the proof of this proposition we will discuss the singular term for $n=1$ in more detail and show, in particular, that this term is not identically zero for $\gamma\neq 0$. This shows that the sharp pointwise Weyl law can be violated for a Kato class potential. We also note that as the potential becomes more singular (that is, as $\eta$ decreases to zero), the number of `additional' terms becomes arbitrarily large.

\begin{proof}
	Since the potential $V$ belongs to $L^q$ for all $q<3/(2-\eta)$, we may apply Lemma \ref{lem:Lq-potential} and obtain 
	$$
	\sup_{x\in M}\sum_{n\ge 1 + \left\lfloor 1/\eta \right\rfloor	} |r^{(n,V)}_0(t,x)|=\cO(t^{-1/2}) \,.
	$$
	Here we used the fact that $\frac{q}{2q-3}\searrow \frac 1\eta$ as $q\nearrow 3/(2-\eta)$, and therefore $N_q=1+\left\lfloor\frac{q}{2q-3}\right\rfloor = 1+ \left\lfloor 1/\eta \right\rfloor$ for all $q$ below but sufficiently close to $3/(2-\eta)$.
	
	This, together with Corollary \ref{coro:weyl-V} implies that
	$$
	\1(-\Delta_g + V\le t)(x,x) = \frac{t^{3/2}}{6\pi^2} - \frac12 \sum_{1\le n\le \lfloor 1/\eta\rfloor} (-1)^{n+1} r_{0,\epsilon}^{(n,V)}(t,x)\, t^{3/2} + \mathcal O(t) \,.
	$$
	This proves the first claim in the proposition, except in the case where $1/\eta$ is an integer and $n=1/\eta$. The fact that this term can also be included in the remainder term follows from \eqref{eq:example-potential-bound} (noting that $1/\eta <2/\eta$), which will be proved below.
	
	From now on, we concentrate on proving \eqref{eq:example-potential-bound} and \eqref{eq:example-potential-lim}. It is easy to see that the limit is independent of $\epsilon$. Therefore we may assume in the following that $\epsilon$ is so small (depending on $n$) and $t$ is so large (depending on $y$) that all points in the support of the integrand of the $n$-fold integral defining $r_{0,\epsilon}^{(n,V)}$ are contained in $B_{\rho/2}(x_0)$. Consequently we can compute the integral defining $r^{(n,V)}_0(t,y)$  in geodesic normal coordinates around $x_0$. We identify $x_0$ with the origin, $T_{x_0}M$ with $\R^3$ and $B_{\rho/2}(x_0)$ with $\cU$. We recall that in these coordinates, $dv_g(z)=\theta(z,0)\,dz$. We will also identify $\exp_{x_0}(y)$ with $y$ and $g_{x_0}$ with the Euclidean metric to lighten the notation. In particular, we have $d_g(x_0,y)=|y|$ for all $y$. We also use the following representation of the Riemannian distance:
 $$d_g(z,z')=\sqrt{\langle z-z',G(z,z')(z-z')\rangle}=:|z-z'|_{z,z'},\quad z,z'\in \cU,$$
 for some smooth family $(G(z,z'))$ of positive definite symmetric matrices with 
 $$G(z,z)=(D_z\exp_{x_0})^T g(\exp_{x_0}(z)) D_z\exp_{x_0},$$
 see for instance \cite[Eq. (4.3)]{DosKenSal-14}. In particular, the map $u\mapsto |u|_{z,z'}$ is homogeneous of degree $1$. Changing variables $z_j= (y+\zeta_j)/\sqrt{t}$ in the integral defining $r^{(n,V)}(t,y)$, we find (with the convention $\zeta_0=0=\zeta_{n+1}$)
 \begin{align*}
 	& r_{0,\epsilon}^{(n,V)}(t,y) =\frac{t^{-n\eta/2}\,\gamma^n}{2(4\pi)^{n+1}}
 \int_{\R^3} d\zeta_n \int_{\R^3} d\zeta_{n-1} \cdots \int_{\R^3} d\zeta_1 \\
 	& \quad \left( \prod_{m=1}^n \frac{\theta(0,(y+\zeta_m)/\sqrt t)\, \chi( |\zeta_{m-1}-\zeta_m|/(\epsilon\sqrt t))}{\theta((y+\zeta_{m-1})/\sqrt t,(y+\zeta_m)/\sqrt t)^{1/2}} \right) \frac{\chi(|\zeta_{n}|_{(y+\zeta_n)/\sqrt t,y/\sqrt t}/(\epsilon\sqrt t))}{\theta((y+\zeta_n)/\sqrt t,y/\sqrt t)^{1/2}} \\
	& \quad \times \kappa\left( \sum_{m=0}^n |\zeta_{m}-\zeta_{m+1}|_{(y+z_{m})/\sqrt t,(y+\zeta_{m+1})/\sqrt{t}}\right) 
	\frac{\sum_{m=0}^n |\zeta_{m}-\zeta_{m+1}|_{(y+z_{m})/\sqrt t,(y+\zeta_{m+1})/\sqrt{t}}}{\prod_{m=0}^n |\zeta_{m}-\zeta_{m+1}|_{(y+z_{m})/\sqrt t,(y+\zeta_{m+1})/\sqrt{t}}} \\ 
	& \quad \times \prod_{m=1}^n \frac{\chi(|y+\zeta_m|/\sqrt t)}{|y+\zeta_m|^{2-\eta}} \,.
 \end{align*}
 We now let $t\to\infty$ and use dominated convergence. Note that the integrand in the above formula tends pointwise to the integrand in \eqref{eq:example-potential-lim} since for all $y,\zeta,\zeta',\zeta''\in\R^3$, 
 $$
 |\zeta|_{(y+\zeta')/\sqrt{t},(y+\zeta'')/\sqrt{t}}\to |\zeta|_{0,0}=|\zeta|
 \qquad\text{as}\ t\to+\ii \,,
 $$
 as well as $\theta(0,0)=1=\chi(0)$.
 
 To get the integrable majorant, one can use 
 $$ (1/C)|z-z'|\le d_g(z,z') \le C|z-z'|
 \qquad\text{for all}\ z,z'\in \cU \,,
 $$
 as well as $1/C \le \theta\le C$ and $\chi\le C$. The required bound then follows immediately from the bound $|\kappa(r)|\le C (1+|r|)^{-2}$ and the fact that, if $n<2/\eta$,
 $$\int_{\R^3}\cdots\int_{\R^3}\frac{\sum_{m=0}^n |z_m - z_{m+1}|\, dz_1\cdots dz_n}{(1+ \sum_{m=0}^n |z_m - z_{m+1}|)^2\  (\prod_{m=0}^n |z_m - z_{m+1}| ) \ |z_1+y|^{2-\eta}\cdots|z_n+y|^{2-\eta}}<+\ii \,,
 $$
 where again $z_0=0=z_{n+1}$. This fact may be proved inductively, using the estimates
 \begin{align*}
 & \sup_{z''}\int_{\R^3}\frac{dz}{(1+a+|z|+|z-z'|)^2|z||z+z''|^{2-\eta}}\le\frac{C}{(1+a+|z'|)^{2-\eta}} \,,\\
 & \sup_{z''}\int_{\R^3}\frac{dz}{(1+a+|z|+|z-z'|)^\beta|z-z'||z+z''|^{2-\eta}}\le\frac{C}{(1+a+|z'|)^{\beta-\eta}}  \,,
 \end{align*}
 which hold for all $z'\in\R^3$, all $a>0$, and all $\beta>\eta$ for some $C>0$ independent of $z',a$. The claimed dependence of the right sides on $a$ follows simply by scaling. 
 
 This completes the proof of \eqref{eq:example-potential-lim}. The same argument yields the uniform bound \eqref{eq:example-potential-bound}.
\end{proof}

In the remainder of this subsection, we discuss the limiting profile that arises in the case $n=1$ in \eqref{eq:example-potential-lim}. We set
$$
\Xi_\eta(y) := \frac{1}{16 \,\pi^2}\int_{\R^3}dz \, \frac{\kappa(2|z|)}{|z| \, |z+y|^{2-\eta}} \,,
$$
so that \eqref{eq:example-potential-lim} reads $\lim_{t\to+\infty} t^{\eta/2} r_{0,\epsilon}^{(1,V)}(t,y/\sqrt t) = \gamma\, \Xi_\eta(y)$ and, by Corollary \ref{coro:weyl-V},
$$
\1(-\Delta_g+V\le t)(y/\sqrt t,y/\sqrt t) = \frac{t^{3/2}}{6\pi^2} - \frac{\gamma\, t^{(3-\eta)/2}}2\, \Xi_\eta(y) + o(t^{(3-\eta)/2}) \,. 
$$
Note that due to the oscillation of the function $\kappa$ it is not a priori clear whether $\Xi_\eta$ is not identically zero. However, physical intuition suggests that places where the potential $V$ assumes large positive values repel `particles' and conversely, places where $V$ is large negative attract `particles'. Thus, we expect the particle density at the singularity, $\1(-\Delta_g+V\le t)(x_0,x_0)$, to be decreasing in $\gamma$ and therefore $\Xi_\eta(0)>0$. The following lemma confirms this and also provides the asymptotic behavior of $\Xi_\eta$ at infinity.

\begin{lemma}\label{xifunction}
	Let $\eta\in(0,1)$. Then
	$$
	\Xi_\eta(0) = \frac{1}{\pi^2}\, \frac{2^{1-\eta}}{3-\eta} \, \frac{\cos(\pi\eta/2)}{1-\eta}\Gamma(\eta) >0
	$$
	and
	$$
	\lim_{|y|\to\infty} |y|^{2-\eta}\, \Xi_\eta(y) = \frac{1}{2\pi^2} \,.	
	$$
\end{lemma}

Numerical integration indicates that $\Xi_\eta$ is positive and decreasing, but not convex.

\begin{proof}
	We have
	\begin{align*}
	16\,\pi^2 \,\Xi_\eta(0) = 2^{-\eta} 4\pi  \int_0^\infty dr\, \frac{\kappa(r)}{r^{1-\eta}} = 2^{5-\eta} \int_0^\infty dr\, \frac{\sin r - r\cos r}{r^{4-\eta}} = \frac{2^{5-\eta}}{3-\eta} \int_0^\infty dr\,\frac{\sin r}{r^{2-\eta}} \,.
	\end{align*}
	Here, when integrating by parts we used the fact that $\lim_{r\to 0} r^{-3+\eta}(\sin r - r\cos r)= 0$. We cannot integrate by parts again if we want to have the integral absolutely convergent at infinity. We write
	$$
	\int_0^\infty dr\,\frac{\sin r}{r^{2-\eta}} = \im \int_0^\infty dr\, \frac{e^{ir}-1}{r^{2-\eta}} \,.
	$$
	The term $1$ was subtracted to make the integral converge at the origin. Thus, we are integrating the analytic function $(e^{iz}-1)/e^{(2-\eta)\log z}$ along the positive real axis. Since the integrand is $O(|z|^{-2+\eta})$ as $|z|\to\ii$ on the upper half plane and $O(|z|^{-1+\eta})$ as $|z|\to0$, we can move the contour to the positive imaginary axis and obtain
	$$
	\int_0^\infty dr\, \frac{e^{ir}-1}{r^{2-\eta}} = i \int_0^\infty dt\, \frac{e^{-t}-1}{t^{2-\eta} e^{(1-\eta/2)i\pi}} = i e^{-(1-\eta/2)i\pi} \Gamma(-1+\eta) \,.
	$$
	In the last formula we have used the standard method of analytic continuation of the gamma function. Using
	$$
	i e^{-(1-\eta/2)i\pi} \Gamma(-1+\eta) = \frac{i e^{i\pi\eta/2}}{1-\eta}\Gamma(\eta) \,,
	$$
	we obtain the claimed formula for $\Xi_\eta(0)$.
	
	In order to compute the large $y$ asymptotics we set $s:=|y|$ and compute
	\begin{align*}
	16\,\pi^2 \,\Xi_\eta(y) & = 2^{-\eta}\, 2\pi \int_0^\infty dr\, r\,\kappa(r) \int_{-1}^1 \frac{dt}{(r^2 - 4 r s t+ 4s^2)^{1-\eta/2}} \\
	& = \frac{2^{-\eta} \pi}{\eta} \int_0^\infty dr\, \frac{\kappa(r)}{s} \left( (r+2s)^{\eta} - |r-2s|^{\eta} \right) \\
	& = \frac{2^{3-\eta}}{\eta} \, s^{-1+\eta} \int_0^\infty d\rho\, \sin(s\rho)\, F(\rho) \,,
	\end{align*}
	where
	$$
	F(\rho) := \rho \int_\rho^\infty \frac{d\rho'}{(\rho')^3} \left( (\rho'+2)^\eta - |\rho'-2|^\eta \right). 
	$$
	Here we integrated by parts once. We shall show that
	$$
	\int_0^\infty d\rho\,\sin(s\rho)\, F(\rho) = \eta\,2^\eta\, s^{-1} + o(s^{-1}) \,,
	$$
	which implies the asymptotics in the lemma. Integrating by parts once and using the integrability of $F'$ on $[1,+\ii)$ and the Riemann--Lebesgue lemma, we obtain
	$$
	\int_1^\infty d\rho\, \sin(s\rho)\, F(\rho) = \frac{\cos s}{s}\, F(1) + o(s^{-1}) \,.
	$$
	This term is cancelled to leading order by
	$$
	\int_0^1 d\rho\, \sin(s\rho)\, \rho F(1) = - \frac{\cos s}{s}\, F(1) + \frac{\sin s}{s^2}\, F(1) \,.
	$$
	This leaves us with
	\begin{align*}
	\int_0^1 d\rho \sin(s\rho) \left( F(\rho) - \rho F(1)\right) & = \int_0^1 d\rho \sin(s\rho)\, \rho \int_\rho^1  \frac{d\rho'}{(\rho')^3} \left( (\rho'+2)^\eta - |\rho'-2|^\eta \right) \\
	& = \eta\, 2^\eta \int_0^1 d\rho \sin(s\rho)\, \rho \int_\rho^1  \frac{d\rho'}{(\rho')^2} \left( 1 + (\rho')^2 g(\rho') \right),
	\end{align*}
	where
	$$
	g(\rho') := \eta^{-1}\, 2^{-\eta} (\rho')^{-3} \left( (\rho'+2)^\eta - |\rho'-2|^\eta - \eta\,2^\eta\,\rho' \right).
	$$
	Using the fact that $g\in C^1[0,1]$ one finds after two integrations by part that
	$$
	\int_0^1 d\rho \sin(s\rho)\, \rho \int_\rho^1  d\rho'\, g(\rho') = \mathcal O(s^{-2}) \,.
	$$
	Finally, by an explicit computation,
	$$
	\int_0^1 d\rho \sin(s\rho)\, \rho \int_\rho^1  \frac{d\rho'}{(\rho')^2} = \frac1s + \mathcal O(s^{-2}) \,.
	$$
	Collecting all these bounds, we obtain the claimed asymptotics of $\int d\rho\,\sin(s\rho)\,F(\rho)$.	
\end{proof}



\subsection{Integrated Weyl laws}

Next, we show that, while the term $r_{0,\epsilon}^V$ can lead to a violation of the pointwise sharp Weyl law, the integrated version of the sharp Weyl law remains valid for Kato class potentials. In fact, it even holds for potentials which are sums of Kato class and $L^{3/2}(M)$ functions.

\begin{proof}[Proof of Theorem \ref{thm:weylintsharp}]
	We begin with the case where $V$ belongs to the Kato class (which was treated by a different method in \cite{HuaSog-20}). Clearly, we have
	$$
	N(t,-\Delta_g+V) = \int_M \1(-\Delta_g +V\le t)(x,x)\,dx \,.
	$$
	In view of the uniform pointwise asymptotics in Corollary \ref{coro:weyl-V}, it suffices to show that
	$$
	\int_M r_{0,\epsilon}^V(t,x)\,dx=\cO(t^{-1/2})
	\qquad\text{as}\ t\to+\ii \,.
	$$ 
	The parameter $\epsilon$ which enters the definition of $r_{0,\epsilon}^V$ is fixed throughout this proof.
	
 From \eqref{eq:kappabound1} we deduce that for all $a,b,x\in M$ and all $\alpha\ge0$,
 \begin{align*}
 	& \int_{\substack{d_g(x,a)<\epsilon\\ d_g(x,b)<\epsilon}}\kappa(\sqrt{t}(d_g(x,a)+\alpha+d_g(b,x))\frac{d_g(x,a)+\alpha+d_g(b,x)}{d_g(a,x)d_g(x,b)}\,dv_g(x) \\
	& \qquad\le \frac{C}{\sqrt{t}}\int_{\substack{d_g(x,a)<\epsilon\\ d_g(x,b)<\epsilon}}\frac{dx}{d_g(a,x)d_g(x,b)}\le \frac{C' \epsilon}{\sqrt{t}}.
 \end{align*}
 We do not track the dependence of the constants in this proof on the geometry of $M$. Inserting this bound into the definition of \eqref{eq:defr0} we easily obtain that for all $n\ge1$ 
 $$
 \int_M \left| r^{(n,V)}_0(t,x) \right| dx \le C''\epsilon (A\|V\|_{\cK(\epsilon)})^{n-1} \|V\|_{L^1}  t^{-1/2} \,. 
 $$
 Decreasing $\epsilon$ if necessary we can assume that $A\|V\|_{\cK(\epsilon)}\le 1/2$ and then we obtain the claimed bound by summing over $n$.
 
 Let us now consider the general case where $V=V_1+V_2$ with $V_1$ in the Kato class and $V_2\in L^{3/2}(M)$. We make use of the fact that for any self-adjoint lower $A$ and $B$ one has $N(0,A+B)\le N(0,A)+N(0,B)$. This implies that for any $\delta\in (0,1]$,
 \begin{align*}
 & N(0,-(1+\delta)\Delta_g+V_1-t) - N(0,-\delta\Delta_g-V_2) \le N(0,-\Delta_g+V-t) \\
 & \qquad \le N(0,-(1-\delta)\Delta_g+V_1-t) + N(0,-\delta\Delta_g+V_2) \,.
 \end{align*}
By the first part of the proof, we have for any $\delta\in(0,1/2]$,
$$
N(0,-(1\pm\delta)\Delta_g+V_1-t) = N((1\pm\delta)^{-1} t,-\Delta_g + (1\pm\delta)^{-1} V_1) = \frac{t^{3/2}}{6\pi^2}+\cO(\delta t^{3/2}+t) \,.
$$ 
(Here we use the fact that the bounds in Corollary \ref{coro:weyl-V} and in the first part of the proof hold for $(1\pm\delta)^{-1}V$ instead of $V$ with constants independent of $\delta\in(0,1/2]$.)

In order to deal with the contribution from $V_2$, we shall use the CLR inequality
 \begin{equation}\label{eq:CLR}
    N(0,-\Delta_g+W)\le C\int_M (W(x)-E_*)^{3/2}_-\,dx
 \end{equation}
 with constants $C$ and $E_*$ depending only on $M$. This bound is essentially known and we give some references after the proof. In our case it implies that
 $$
 N(0,-\delta\Delta_g\pm V_2) = N(0,-\Delta_g\pm \delta^{-1}V_2) \le C \int_M (\delta^{-1}V_2(x) -E_*)^{3/2}_-\,dx = \mathcal O(\delta^{-3/2}) \,.
 $$
 Choosing $\delta=t^{-3/5}$ we get $\delta t^{3/2}=\delta^{-3/2}=t^{9/10} = \mathcal O(t)$. Combining the previous bounds we obtain the assertion of the theorem.
\end{proof}

Since we have not been able to find the CLR inequality \eqref{eq:CLR} stated explicitly in the literature, let us briefly comment on its proof. Inequalities of this type go back to Cwikel \cite{Cwikel-77}, Lieb \cite{Lieb-76a,Lieb-79} and Rozenblum \cite{Rozenblum-72,Rozenblum-76} in the Euclidean case. Further proofs in that case were given in \cite{Fefferman-83,LiYau-83,Conlon-85,Frank-14}. The proofs in \cite{Rozenblum-72,LiYau-83} extend to the case of a manifold. For instance, \cite{FraLieSei-10}, which is based on \cite{LiYau-83} (see also \cite{LevSol-97,BlaStuRez-87}), implies that the Sobolev inequality
$$
\int_M \left( |\nabla_g u|^2 + E_* |u|^2\right)dv_g(x) \ge S \left( \int_M |u|^6\,dv_g(x) \right)^{1/3},
$$
which is valid with some positive constants $S$ and $E_*$ depending on $M$, implies  \eqref{eq:CLR} with $C\le \sqrt e S^{-3/2}$. (Conversely, it is elementary to see that \eqref{eq:CLR} implies the Sobolev inequality with constant $S$ satisfying $C\ge S^{-3/2}$.) Alternatively, one can deduce \eqref{eq:CLR} from \cite[Thm.~3.2 and Lem.~3.4]{Frank-14}. This needs as an ingredient the bound $\1(-\Delta_g\le E)(x,x)\le C'(E+E_*)^{3/2}$ for all $E\ge 0$ and $x\in M$, which itself follows from a rough version of the Weyl law for $-\Delta_g$.

\begin{remark}\label{rk:improved-remainder}
	The argument in the previous proof based on \eqref{eq:CLR} shows that if $V\in L^{3/2}(M)$ and $N(t,-\Delta_g) = t^{3/2}/(6\pi^2) + o(t)$, then also $N(t,-\Delta_g +V) = t^{3/2}/(6\pi^2) + o(t)$. More precisely, if $\left| N(t,-\Delta_g) - t^{3/2}/(6\pi^2)\right| \le r(t)$ for a nondecreasing function $r$, then the previous proof shows that, if $V\in L^{3/2}(M)$, then
	$$
	\left|N(t,-\Delta_g +V) - t^{3/2}/(6\pi^2)\right|\le r(2t) + C t^{9/10} \,.
	$$
	Since the CLR inequality is also valid in any dimension $N\ge3$, the previous argument extends to this case with the remainder $\cO(t^{9/10})$ replaced by $\cO(t^{N^2/(2(N+2))})$. For instance, when $M$ has nonpositive sectional curvatures, then this bound and B\'erard's theorem yield
	$$
	N(t,-\Delta_g +V) = \frac{|\Sph^{N-1}|}{N\,(2\pi)^N}\ t^{N/2} + \mathcal O (t^{(N-1)/2}/\ln t)
	$$
	for $V\in L^{N/2}(M)$. The same result, but for Kato class $V$, was shown in \cite[Thm.~1.3]{HuaSog-20}. On the other hand, our argument  is not strong enough to reproduce the $L^{3/2}$ analogue of the $\mathcal O(t^{3/4})$ bound for $M=(\R/\Z)^3$ and Kato class $V$, which is shown in \cite[Thm.~1.4]{HuaSog-20}.
\end{remark}

We expect that $N(t,-\Delta_g) = t^{3/2}/(6\pi^2) + o(t)$ implies $N(t,-\Delta_g+V) = t^{3/2}/(6\pi^2) + o(t)$ for all Kato class $V$, but we do not know how to prove it.


\section{Euclidean domains}\label{sec:eucl}

In this section we will show that the previous construction extends easily to domains in $\R^3$. Let $\Omega\subsetneq\R^3$ be an open set. We emphasize that no regularity assumption on $\partial\Omega$ and no finite measure assumption are imposed, except in Corollary \ref{coro:weyl-V-domain}. Let us denote by $-\Delta_\Omega^D$ the Dirichlet Laplacian on $\Omega$. Let $V$ a Kato class potential on $\Omega$. The starting point is a parametrix for the Green's function 
$$G_\lambda^V(x,y):=(-\Delta_\Omega^D+V+\lambda)^{-1}(x,y),\ x\neq y\in\Omega.$$
which will be the Green's function for same operator on $\R^N$, thus getting rid of the boundary:
$$T_{\lambda}^V(x,y):=(-\Delta_{\R^3}+V\1_\Omega+\lambda)^{-1}(x,y),\ x\neq y\in\R^N.$$
We begin with an estimate on this Green's function. 

\begin{lemma}\label{lem:domain-decay-green-V}
 There exist $C>0$ and $\Lambda>0$ such that for all $x\neq y\in\R^N$ and for all $\lambda\ge\Lambda$, we have
 $$|T_{\lambda}^V(x,y)|\le C\frac{e^{-\frac{\sqrt{\lambda}}{2}|x-y|}}{|x-y|}.$$
\end{lemma}

\begin{proof}
 Using the resolvent formula, we expand $T_{\lambda}^V=T_{\lambda}^0+\sum_{n\ge1} T_{\lambda}^{(n,V)}$, where 
 $$T_{\lambda}^{(n,V)}(x,y)=\frac{1}{(4\pi)^{n+1}}\int_\Omega\,dz_1\cdots\int_\Omega\,dz_n \frac{e^{-\sqrt{\lambda}(|x-z_1|+|z_1-z_2|+\cdots+|z_n-y|)}}{|x-z_1||z_1-z_2|\cdots|z_n-y|}V(z_1)\cdots V(z_n).$$
 By the triangle inequality, we have for all $z_1,\ldots,z_n$,
 $$e^{-\frac{\sqrt{\lambda}}{2}(|x-z_1|+|z_1-z_2|+\cdots+|z_n-y|)}\le e^{-\frac{\sqrt{\lambda}}{2}|x-y|}$$
 and
 \begin{align*}
    \frac{|x-y|}{|x-z_1||z_1-z_2|\cdots|z_n-y|} &\le \frac{|x-z_1|+|z_1-z_2|+\cdots+|z_n-y|}{|x-z_1||z_1-z_2|\cdots|z_n-y|}\\
    &= \frac{1}{|z_1-z_2|\cdots|z_n-y|}+\frac{1}{|x-z_1||z_2-z_3|+\cdots+|z_n-y|}\\
    &\qquad+\cdots+\frac{1}{|x-z_1||z_1-z_2|\cdots|z_{n-1}-z_n|}.
 \end{align*}
 As a consequence, we deduce that 
 \begin{align*}
    \sup_{x\neq y}|x-y|e^{\frac{\sqrt{\lambda}}{2}|x-y|}|T_{\lambda}^{(n,V)}(x,y)|&\le\frac{n+1}{(4\pi)^{n+1}}\left(\sup_{x\in\R^N}\int_{\Omega}\frac{e^{-\frac{\sqrt{\lambda}}{2}|x-z|}}{|x-z|}|V(z)|\,dz\right)^n\\
    &\le C\,(n+1) (A\|V\|_{\cK(8/\sqrt{\lambda})})^n,
 \end{align*}
 where in the last line we used Lemma \ref{lem:est-kato}.
\end{proof}

\begin{remark}
 By the maximum principle, we know that $0\le G_\lambda^0\le T_\lambda^0$ for all $\lambda\ge0$ (pointwise). Hence, the same proof using the resolvent identity for $G_\lambda^V$ shows that the same bound
 \begin{equation}\label{eq:est-GlambdaV}
    |G_\lambda^V(x,y)|\le C\frac{e^{-\frac{\sqrt{\lambda}}{2}|x-y|}}{|x-y|}
 \end{equation}
 holds for $\lambda$ large enough.
\end{remark}

We can now estimate the error between $G_\lambda^V$ and $T_{\lambda}^V$.

\begin{lemma}\label{lem:domain-error-green}
 There are $C>0$ and $\Lambda>0$ such that for all $x\neq y\in\Omega$ and $\lambda\ge\Lambda$ we have 
 $$\left| G_\lambda^V(x,y) - T_{\lambda}^V(x,y) \right| \le C \frac{e^{-\frac{\sqrt{\lambda}}{2}d(x,\partial\Omega)}}{d(x,\partial\Omega)}.$$
\end{lemma}

\begin{proof}
 For all $x\in\Omega$, the function $\gamma_{\lambda}^V:=G_\lambda^V-T_{\lambda}^V$   satisfies the equation
 $$\begin{cases}
    (-\Delta+V+\lambda)_y\gamma_{\lambda}^V(x,\cdot)=0 & \text{in}\ \Omega,\\
    \gamma_{\lambda}^V(x,\cdot) = -T_{\lambda}^V(x,\cdot) &\text{on}\ \partial\Omega.
   \end{cases}
 $$
 Let us show that there are $C>0$ and $\Lambda>0$ such that for all $x\in\Omega$ and $\lambda\ge\Lambda$ we have 
 $$
 \|\gamma_{\lambda}^V(x,\cdot)\|_{L^\ii(\Omega)}\le C\|\gamma_{\lambda}^V(x,\cdot)\|_{L^\ii(\partial\Omega)} \,.
 $$
 Since $\|\gamma_{\lambda}^V(x,\cdot)\|_{L^\ii(\partial\Omega)} =\|T_{\lambda}^V(x,\cdot)\|_{L^\ii(\partial\Omega)}$, this together with Lemma \ref{lem:domain-decay-green-V} gives the result. 
 
 To prove the desired bound, let us introduce the function $\tilde{\gamma}:=-G_\lambda^0 (V\gamma_{\lambda}^V(x,\cdot))$, so that
  $$\begin{cases}
        (-\Delta+\lambda)\tilde{\gamma}=-V\gamma_{\lambda}^V(x,\cdot) &\text{in}\ \Omega, \\
        \tilde{\gamma}=0 &\text{on}\ \partial\Omega \,.
   \end{cases}
 $$
Then, the function $\gamma_{\lambda}^V(x,\cdot)-\tilde{\gamma}$ satisfies the equation
 $$\begin{cases}
    (-\Delta+\lambda)(\gamma_{\lambda}^V(x,\cdot)-\tilde{\gamma}) = 0 &\text{in}\ \Omega,\\
    \gamma_{\lambda}^V(x,\cdot)-\tilde{\gamma} = \gamma_{\lambda}^V(x,\cdot) &\text{on}\ \partial\Omega.
   \end{cases}
 $$
 By the maximum principle, this implies that 
 $$\|\gamma_{\lambda}^V(x,\cdot)-\tilde{\gamma}\|_{L^\ii(\Omega)}\le \|\gamma_{\lambda}^V(x,\cdot)\|_{L^\ii(\partial\Omega)}$$
 and hence 
 $$\|\gamma_{\lambda}^V(x,\cdot)\|_{L^\ii(\Omega)}\le \|\gamma_{\lambda}^V(x,\cdot)\|_{L^\ii(\partial\Omega)} + \|\tilde{\gamma}\|_{L^\ii(\Omega)}.$$
 On the other hand, we have 
 $$\|\tilde{\gamma}\|_{L^\ii(\Omega)}=\|G_\lambda^0 V \gamma_{\lambda}^V (x,\cdot)\|_{L^\ii(\Omega)}\le \|G_\lambda^0 V\|_{L^\ii\to L^\ii}\|\gamma_{\lambda}^V(x,\cdot)\|_{L^\ii(\Omega)}.$$
 Since, again by the maximum principle, $0\le G_\lambda^0\le T_{\lambda}^0$ pointwise, we deduce using Lemma \ref{lem:est-kato} that
 $$\|G_\lambda^0 V\|_{L^\ii\to L^\ii}\le\|T_{\lambda}^0 V\|_{L^\ii\to L^\ii}\le C\|V\|_{\cK(2/\sqrt{\lambda})}.$$
 For $\lambda>0$ large enough so that $C\|V\|_{\cK(2/\sqrt{\lambda})}<1$, we deduce that
 $$\|\gamma_{\lambda}^V(x,\cdot)\|_{L^\ii(\Omega)}\le (1-C\|V\|_{\cK(2/\sqrt{\lambda})})^{-1}\|\gamma_{\lambda}^V(x,\cdot)\|_{L^\ii(\partial\Omega)},$$
 which is the desired result.
\end{proof}

\begin{corollary}\label{coro:error-domain}
 There are $C>0$ and $\Lambda>0$ such that for all $x\in\Omega$ and all $\lambda\ge\Lambda$ we have 
 $$\left|\partial_\lambda G_\lambda^V(x,x) - \partial_\lambda T_{\lambda}^V(x,x)\right|\le C\frac{e^{-\frac{\sqrt{\lambda}}{2}d(x,\partial\Omega)}}{\lambda d(x,\partial\Omega)}.$$
\end{corollary}

\begin{proof}
 We first remark that $\partial_\lambda G_\lambda^V = -(G_\lambda^V)^2$, where the product on the right side is in the sense of operators. Thus, 
 $$\partial_\lambda G_\lambda (x,x) = -\int_\Omega |G_\lambda^V(x,y)|^2\,dy.$$
 This and the corresponding formula for $T_{\lambda}^V$ gives
 $$\left|\partial_\lambda G_\lambda^V(x,x) - \partial_\lambda T_{\lambda}^V(x,x)\right|\le \int_\Omega|G_\lambda^V(x,y)-T_{\lambda}^V(x,y)|(|G_\lambda^V(x,y)|+|T_{\lambda}^V(x,y)|)\,dy.$$
 Using Lemma \ref{lem:domain-decay-green-V}, \eqref{eq:est-GlambdaV} and Lemma \ref{lem:domain-error-green}, we deduce that
 \begin{align*}
    \left|\partial_\lambda G_\lambda^V(x,x) - \partial_\lambda T_{\lambda}^V(x,x)\right| &\le C\frac{e^{-\frac{\sqrt{\lambda}}{2}d(x,\partial\Omega)}}{d(x,\partial\Omega)}\int_\Omega\frac{e^{-\frac{\sqrt{\lambda}}{2}|x-y|}}{|x-y|}\,dy\\
    &\le C\frac{e^{-\frac{\sqrt{\lambda}}{2}d(x,\partial\Omega)}}{\lambda d(x,\partial\Omega)} \,,
 \end{align*}
 which is the desired result.
\end{proof}

It will also be useful to have the following uniform bound.

\begin{lemma}\label{lem:domain-uniform}
 There are $C>0$ and $\Lambda>0$ such that for all $x\in\Omega$ and all $\lambda\ge\Lambda$ we have 
 $$|\partial_\lambda G_\lambda^V(x,x)|\le \frac{C}{\sqrt{\lambda}}.$$
\end{lemma}

\begin{proof}
 We have seen that 
 $$\sup_{x\in\Omega}|\partial_\lambda G_\lambda^V(x,x)|=\sup_{x\in\Omega}\int_\Omega|G_\lambda^V(x,y)|^2\,dy=\|G_\lambda^V\|_{L^2\to L^\ii}^2.$$
 Now using the resolvent identity as operators
 $$ G_\lambda^V = G_\lambda^0 -G_\lambda^0 V G_\lambda^V,$$
 we deduce that 
 $$\|G_\lambda^V\|_{L^2\to L^\ii}\le \|G_\lambda^0\|_{L^2\to L^\ii} +\|G_\lambda^0 V\|_{L^\ii\to L^\ii}\|G_\lambda^V\|_{L^2\to L^\ii}.$$
 As in the proof of Lemma \ref{lem:domain-error-green} we have $\|G_\lambda^0 V\|_{L^\ii\to L^\ii}\le C\|V\|_{\cK(2/\sqrt{\lambda})}<1$ for all large $\lambda$, so
 $$\|G_\lambda^V\|_{L^2\to L^\ii}\le(1-C\|V\|_{\cK(2/\sqrt{\lambda})})^{-1}\|G_\lambda^0\|_{L^2\to L^\ii}.$$
 Using again the pointwise inequality $0\le G_\lambda^0\le T_{\lambda}^0$, we deduce the desired bound.
\end{proof}

One can then use the same strategy as in the previous section to decompose $\partial_\lambda T_{\lambda}^V(x,x)$, and we obtain

\begin{proposition}\label{prop:weyl-V-domain}
 Let $V:\Omega\to\R$ in the Kato class. Then there is an $\epsilon_0>0$ such that for all $\epsilon\in(0,\epsilon_0]$ one has
 $$
 \1(-\Delta_D^\Omega+V\le t)(x,x) = \frac{t^{3/2}}{6\pi^2} - \frac12 r_{0,\epsilon}^V(t,x)\,t^{3/2}+\cO^{(\epsilon)}\left(\frac{t}{d(x,\partial\Omega)}\right),\quad t\to+\ii$$
 where the $\cO^{(\epsilon)}$ is uniform in $x\in \Omega$ (but depends on $\epsilon$) and where
 $$
 r_{0,\epsilon}^V(t,x)= \sum_{n\ge1}(-1)^n r_{0,\epsilon}^{(n,V)}(t,x)
 $$
 with $r_{0,\epsilon}^{(n,V)}(t,x)$ defined as in the previous section. More precisely, one has for all $x\in\Omega$, $t\ge 0$ and $\epsilon\in (0,\epsilon_0]$,
 $$
 \left|  \1(-\Delta_D^\Omega+V\le t)(x,x) - \frac{t^{3/2}}{6\pi^2} + \frac12 r_{0,\epsilon}^V(t,x)\,t^{3/2} \right| \le C_0\, \alpha^{-1}\, (t + \alpha^{-2}) \,,
 $$
where $\alpha = \min\{d(x,\partial\Omega),\epsilon\}$ and where $C_0$ is an absolute constant.
\end{proposition}

\begin{proof}
	We only sketch the main steps in the proof. We apply again Theorem \ref{thm:stieljes-3D} in the simplified from \eqref{eq:stieltjessimple} with $B_0(t) = -(6\pi^2)^{-1} + \frac12 r_{0,\epsilon}^V(t,x)$ and $B_1(t) = \1(-\Delta_D^\Omega+V\le t)(x,x)$. Note that no term $C_{B_2}$ arises in this situation because the remainder $R_{\lambda,\epsilon}$ in the parametrix for $V=0$ vanishes in a ball of radius $\epsilon/2$. Thus, we have $C_{B_2}=0$ and $t_0=0$. The number $\epsilon_0$ in the theorem is chosen so that $\|V\|_{\mathcal K(\epsilon_0)}$ is smaller than an absolute constant, so that the series defining $r_{0,\epsilon}^V$ converges and so that Corollary \ref{coro:error-domain} is applicable for all $\lambda\ge C/\epsilon_0^{-2}$. This, together with the Rodnianski--Schlag bound \eqref{eq:rodschl}, shows that $\|B_0\|_{L^\infty}$ is bounded by an absolute constant and that $B_0$ satisfies the almost monotonicity condition for any $\delta\ge 0$ with constant $C_{B_0}$ bounded by an absolute constant times $\delta$. Next, we have $0\le B_1(0)\le C \epsilon_0^{-3}$ with an absolute constant $C$. This follows as in the proof of Corollary \ref{coro:weyl-V} from Lemma \ref{lem:domain-uniform}.
	Finally, the decomposition of $\partial_\lambda T_\lambda^V(x,x)$ leads to an error $(\epsilon\lambda)^{-1} e^{-\epsilon\sqrt\lambda/4}$ for $\lambda\ge C \epsilon^{-2}$ for any $\epsilon\in(0,\epsilon_0]$. On the other hand, Corollary \ref{coro:error-domain} leads to an error  $(d(x,\partial\Omega)\lambda)^{-1} e^{-d(x,\partial\Omega)\sqrt\lambda/2}$ for $\lambda \ge C \epsilon_0^{-2}$. Together, this gives an error $(\alpha\lambda)^{-1} e^{-\alpha\sqrt\lambda/4}$ for $\lambda\ge C \epsilon^{-2}$ for any $\epsilon\in(0,\epsilon_0]$. Thus, we apply Theorem \ref{thm:stieljes-3D} with $\epsilon_0=\alpha/2$, $\Lambda = C \alpha^{-2}$ and $C_0 = C \alpha^{-1}$. With these choices, the corollary is a consequence of \eqref{eq:stieltjessimple}.
\end{proof}

\begin{corollary}\label{coro:weyl-V-domain}
	Assume that $\Omega$ has finite measure and satisfies $|\{ x\in\Omega :\ d(x,\partial\Omega)\le\delta\}|\le C\delta$ for all sufficiently small $\delta>0$. Let $V:\Omega\to\R$ be in the Kato class. Then, one has
 $$N(t,-\Delta_\Omega^D+V)=\frac{t^{3/2}}{6\pi^2}\,|\Omega| + \cO(t\log t),\quad t\to+\ii.$$
\end{corollary}

For $V=0$ this corollary is due to Courant \cite{Courant-20}. Still for $V=0$, the error estimate $\cO(t\log t)$ was improved by Seeley \cite{Seeley-78} to the sharp order $\cO(t)$. We have not tried to adapt Seeley's proof to the case $V\neq0$.

\begin{proof}
 We write $N(t,-\Delta_\Omega^D+V)= \int_\Omega \1(-\Delta_\Omega^D+V\le t)(x,x)\,dx$ and split the integral into the region $\{ x\in\Omega:\ d(x,\partial\Omega)>1/\sqrt{t}\}$ and its complement.
 
 In the set $\{ x\in\Omega:\ d(x,\partial\Omega)>1/\sqrt{t}\}$ we use the bound of Proposition \ref{prop:weyl-V-domain} together with
 \begin{equation}
 \label{eq:bdrynbh}
\int_\Omega\frac{\1(d(x,\partial\Omega)>1/\sqrt{t})}{d(x,\partial\Omega)}\,dx\le C\log t,
 \end{equation} 
 and obtain 
 \begin{align*}
    \int_\Omega\1(d(x,\partial\Omega)>1/\sqrt{t}) & \1(-\Delta_\Omega^D+V\le t)(x,x)\,dx \\
    &= \frac{t^{3/2}}{6\pi^2}|\{x\in\Omega,\ d(x,\partial\Omega)>1/\sqrt{t})\}|+\cO(t\log t) \\
    &= \frac{t^{3/2}}{6\pi^2}(|\Omega|-\cO(t^{-1/2}))+\cO(t\log t)\\
    &= \frac{t^{3/2}}{6\pi^2}|\Omega|+\cO(t\log t).
 \end{align*}
 To prove \eqref{eq:bdrynbh} under the assumptions of the corollary, we write, using the coarea formula and integration by parts,
 \begin{align*}
 \int_\Omega & \frac{\1(d(x,\partial\Omega)>1/\sqrt{t})}{d(x,\partial\Omega)} \,dx \\
 & = \int_{1/\sqrt t}^\infty \mathcal H^2(\{ x\in\Omega:\ d(x,\partial\Omega)=s\}) \,\frac{ds}{s} \\
 & = \int_{1/\sqrt t}^\infty |\{ x\in\Omega:\ d(x,\partial\Omega)\le s\}| \,\frac{ds}{s^2}
 - |\{ x\in\Omega:\ d(x,\partial\Omega)\le 1/\sqrt t \}|\sqrt t 
 \end{align*}
 Since $|\{ x\in\Omega:\ d(x,\partial\Omega)\le s\}|$ is $\le C s$ for all sufficiently small $s$ and $\le |\Omega|$ for all $s$, we obtain \eqref{eq:bdrynbh}.
  
 In the set $\{ x\in\Omega:\ d(x,\partial\Omega)\le 1/\sqrt{t}\}$ we use the uniform bound of Lemma \ref{lem:domain-uniform}, which implies that for all $\lambda$ large enough one has
 $$
 \frac{1}{(t+\lambda)^2}\1(-\Delta_\Omega^D+V\le t)(x,x)\le -\partial_\lambda G_\lambda^V(x,x)\le \frac{C}{\sqrt{\lambda}} \,.
 $$
 Taking $\lambda=t$ gives for all $t$ large enough,
 $$\1(-\Delta_\Omega^D+V\le t)(x,x)\le Ct^{3/2},$$
 and hence 
 $$\int_\Omega \1(d(x,\partial\Omega)\le1/\sqrt{t})\1(-\Delta_\Omega^D+V\le t)(x,x)\,dx\le C\frac{t^{3/2}}{\sqrt{t}}=Ct \,.$$
 
 Combining the bounds in the two different sets we obtain the claimed asymptotics.
\end{proof}


\appendix

\section{Schr\"odinger operators with Kato-class potentials}\label{app:op-kato}

In this appendix we work in general dimensions $N\ge 1$, since this presents no extra effort and since the results might be useful elsewhere. The Kato class was explicitly introduced by Kato \cite{Kato-72} in the context of selfadjointness of Schr\"odinger operators (see also \cite{Schechter-71}), and developed by Aizenman and Simon \cite{AizSim-82} as a natural class of potentials to study Schr\"odinger semigroups. For a comprehensive review, we refer to \cite{Simon-82}.

\subsection{Kato class potentials}

Let
$$w_N(r)=\begin{cases}
r^{-(N-2)} &\text{if}\ N\ge3,\\
\log(1/r) &\text{if}\ N=2,\\
1 &\text{if}\ N=1,
\end{cases}
$$
and set, for any measurable function $V:\R^N\to\R$,
$$\|V\|_{\cK(r)}:=\sup_{x\in\R^N}\int_{|y-x|<r}|V(y)|w_N(|x-y|)\,dy.$$
We say that $V:\R^N\to\R$ belongs to the Kato class if it is measurable and
$$\lim_{r\to0}\|V\|_{\cK(r)} = 0 \,.$$

\begin{remark}\label{katovs32}
When $N\ge2$, any potential $V\in L^p(\R^N)+L^\ii(\R^N)$ for $p>N/2$ belongs to the Kato class by the Hölder inequality (one can even replace $L^p(\R^N)$ by $L^p_{{\rm unif}}(\R^N)$). However, the Kato class and $L^{N/2}$ are distinct, and neither one is included in the other. For instance, $x\mapsto\1(|x|\le1)/(|x|^2\log|x|)\in L^{N/2}$ but does not belong to the Kato class; and $x\mapsto\1(|x|\le2)(|x|-1)^{-\alpha}$ for $2/N<\alpha<1$ belongs to the Kato class but not to $L^{N/2}$. 
\end{remark}

For $\lambda>0$ let us introduce the Green's function of the Laplacian,
$$
G_\lambda(x-y):=(-\Delta+\lambda)^{-1}(x,y)=\frac{1}{(2\pi)^N}\int_{\R^N}\frac{e^{i\xi\cdot (x-y)}}{|\xi|^2+\lambda}\,d\xi \,,\qquad x,y\in\R^N \,.
$$

\begin{lemma}\label{lem:est-kato}
 Let $V$ a Kato class potential. Then, for all $\lambda\ge1$ we have 
 $$\|V*G_\lambda\|_{L^\ii}\le C\|V\|_{\cK(2/\sqrt{\lambda})}.$$
\end{lemma}

The proof is based on the following bounds on the Green's function $G_\lambda$.

\begin{lemma}\label{lem:bound-green}
	There is a $C>0$ such that for any $\lambda>0$ and $x\neq y\in\R^N$,  the function $G_\lambda$ satisfies:	
	 if $|x-y|\le1/\sqrt{\lambda}$,
	$$|G_\lambda(x-y)|\le\begin{cases}
	\frac{C}{|x-y|^{N-2}} & \text{if}\ N\ge3,\\
	C\log\frac{1}{\sqrt{\lambda}|x-y|}& \text{if}\ N=2,\\
	\frac{C}{\sqrt{\lambda}}&\text{if}\ N=1,
	\end{cases}
	$$
	and if $|x-y|\ge1/\sqrt{\lambda}$,
	$$|G_\lambda(x-y)|\le C\lambda^{\frac{N-2}{2}}\frac{e^{-\sqrt{\lambda}|x-y|}}{(\sqrt{\lambda}|x-y|)^{\frac{N-1}{2}}}=C\frac{e^{-\sqrt{\lambda}|x-y|}}{|x-y|^{N-2}}(\sqrt{\lambda}|x-y|)^{\frac{N-3}{2}}.$$
\end{lemma}

This lemma is well-known and there are several ways to prove it. One is by writing $(|\xi|^2+\lambda)^{-1}$ as a superposition of Gaussians $e^{-t|\xi|^2}$ and using the Fourier transform of the Gaussian together with some straightforward estimates. Another way is to recognize $G_\lambda(x-y)$ as an inverse power of $|x-y|$ times a Bessel function of the third kind and using standard estimates for these.

\begin{proof}[Proof of Lemma \ref{lem:est-kato}]
 Let $x\in\R^N$. Let us split
 \begin{multline*}
    \int_{\R^N}|V(y)| G_\lambda(x-y) \,dy = \int_{|x-y|\le1/\sqrt{\lambda}}|V(y)| G_\lambda(x-y)\,dy \\
    + \sum_{k\ge1} \int_{k/\sqrt{\lambda}<|x-y|\le(k+1)/\sqrt{\lambda}} |V(y)| G_\lambda(x-y)\,dy.
 \end{multline*}
 Using the first bounds of Lemma \ref{lem:bound-green}, we deduce that 
 $$\int_{|x-y|\le1/\sqrt{\lambda}}|V(y)| G_\lambda(x-y)\,dy \le \|V\|_{\cK(1/\sqrt{\lambda})}\sup_{|z|\le1/\sqrt{\lambda}} \frac{G_\lambda(z)}{w_N(|z|)},$$
 and since for all $|z|\le1/\sqrt{\lambda}$ we have
 $$\frac{G_\lambda(z)}{w_N(|z|)}\le
 \begin{cases}
  C &\text{if}\ N\ge3,\\
  C\frac{\log \frac{1}{\sqrt{\lambda}|z|}}{\log\frac{1}{|z|}}=C(1+\frac{\log \frac{1}{\sqrt{\lambda}}}{\log\frac{1}{|z|}}) \le 2C &\text{if}\ N=2,\\
  \frac{C}{\sqrt{\lambda}} \le C &\text{if}\ N=1,
 \end{cases}
 $$
 we deduce that 
 $$\int_{|x-y|\le1/\sqrt{\lambda}}|V(y)| G_\lambda(x-y)\,dy \le C\|V\|_{\cK(1/\sqrt{\lambda})}.$$
 For any $k\ge1$, split the annulus 
 $$A=\{ z\in\R^N\ :\ k/\sqrt{\lambda}< |z|\le (k+1)/ \sqrt{\lambda}\}= A\cap \bigcup_{j=1}^{M(k,\lambda)} B(z_j,2/\sqrt{\lambda})$$
 in $M(k,\lambda)\le Ck^{N-1}$ balls of radius $2/\sqrt{\lambda}$. Thus, we have
 \begin{align*}
    \int_{k/\sqrt{\lambda}<|x-y|\le(k+1)/\sqrt{\lambda}} |V(y)| G_\lambda(x-y)\,dy &\le \sum_{j=1}^{M(k,\lambda)} \int_{\substack{|x-y-z_j|\le 2/\sqrt{\lambda} \\ x-y\in A}} |V(y)| G_\lambda(x-y)\,dy\\
    &\le \|V\|_{\cK(2/\sqrt{\lambda})} \sum_{j=1}^{M(k,\lambda)}\sup_{\substack{|z-z_j|\le2/\sqrt{\lambda} \\ z\in A}}\frac{G_\lambda(z)}{w_N(|z-z_j|)} \,.
 \end{align*}
 Since $r\mapsto 1/w_N(r)$ is an increasing function, for all $|z-z_j|\le 2/\sqrt{\lambda}$ we have
 $$\frac{1}{w_N(z-z_j)}\le \frac{1}{w_N(2/\sqrt{\lambda})} \,,$$
 while, since $G_\lambda$ is radially decreasing,
 \begin{align*}
    \sum_{j=1}^{M(k,\lambda)}\sup_{\substack{|z-z_j|\le2/\sqrt{\lambda} \\ z\in A}}\frac{G_\lambda(z)}{w_N(|z-z_j|)} &\le CM(k,\lambda)\frac{G_\lambda(k/\sqrt{\lambda})}{w_N(2/\sqrt{\lambda})}\\
    &\le Ck^{N-1}\frac{\lambda^{(N-2)/2}}{w_N(2/\sqrt{\lambda})}\frac{e^{-k}}{k^{N-2}} k^{(N-3)/2} \\
    &\le Ck^{(N-1)/2} e^{-k} \,.
 \end{align*}
 Summing over $k$, this leads to the result.
\end{proof}

As a consequence of Lemma \ref{lem:bound-green} we now prove the following quadratic form estimate, following \cite[Rk.~(2), p.~459]{Simon-82}.

\begin{lemma}
 Let $V$ in the Kato class. Then, for all $u\in C^\ii_0(\R^N)$ and all $\lambda\ge1$ we have 
 $$\int_{\R^N} V|u|^2\,dx \le C\|V\|_{\cK(2/\sqrt{\lambda})}\left(\int_{\R^N}|\nabla u|^2\,dx+\lambda\int_{\R^N}|u|^2\,dx\right).$$
\end{lemma}

Note that, since the Kato class is included in $L^1_{{\rm loc}}(\R^N)$, the left-side of the previous inequality is well-defined. Moreover, the inequality shows that the quadratic form $u\mapsto \langle u,Vu\rangle$ extends continuously from $C^\ii_0(\R^N)$ to $H^1(\R^N)$, and that it is relatively bounded with respect to the quadratic form of $-\Delta$ with relative bound zero (by taking $\lambda$ large enough). By the KLMN theorem, the quadratic form $u\mapsto\int\left(|\nabla u|^2+V|u|^2\right)dx$ defined on $H^1(\R^N)$ is thus associated to a unique bounded below self-adjoint operator with domain included in $H^1(\R^N)$. This is the definition of $-\Delta+V$ for $V$ in the Kato class.

\begin{proof}
 We prove the estimate for bounded, compactly supported $V$. Then, replacing $V$ by $V_n=\1(|V|\le n)\1(|x|\le n)V$ and taking the limit $n\to\ii$ by monotone convergence on the left side and using $\|V_n\|_{\cK(r)}\le\|V\|_{\cK(r)}$ on the right side leads to the result. Now let $v$ a Schwartz function. Define for all $z\in\C$,
 $$v_z=|v|^z \frac{v}{|v|}$$
 and consider the function
 $$\phi:s\in S\mapsto \left\langle |V|^s v_{2s},(-\Delta+\lambda)^{-1}|V|^{1-s} v_{2(1-s)}\right\rangle$$
 defined on the strip $S=\{z\in\C:\ 0\le \re z\le 1\}$. Due to the bound for all $s\in S$ and for all $x\in\R^N$,
 $$||V(x)|^s v_{2s}(x)|\le \1_{\supp V} (1+\|V\|_{L^\ii})(1+\|v\|_{L^\ii}^2),$$
 we deduce that $\| |V|^s v_{2s}\|_{L^2}\le C$, and since $\|(-\Delta+\lambda)^{-1}\|_{L^2\to L^2}\le1$, we deduce that $\phi$ is well-defined, analytic on the interior of $S$, continuous and bounded on $S$. Now for $s=it$ we have, by Lemma \ref{lem:est-kato}
 \begin{align*}
    |\phi(it)| &\le \||V|^{it} v_{2it}\|_{L^\ii}\|(-\Delta+\lambda)^{-1}|V|^{1-it}\|_{L^\ii\to L^\ii}\| v_{2(1-it)}\|_{L^1}\\
    &\le \|G_\lambda * |V|\|_{L^\ii}\|v\|_{L^2}^2\\
    &\le C\|V\|_{\cK(2/\sqrt{\lambda})}\|v\|_{L^2}^2.
 \end{align*}
 For $s=1+it$ we have, again by Lemma \ref{lem:est-kato}
 \begin{align*}
  |\phi(1+it)| &\le \|v_{2(1+it)}\|_{L^1}\|(-\Delta+\lambda)^{-1}|V|^{1+it}\|_{L^\ii\to L^\ii}\||V|^{-it}v_{-2it}\|_{L^\ii}\\
  &\le \|G_\lambda * |V| \|_{L^\ii} \|v\|_{L^2}^2\\
  &\le C\|V\|_{\cK(2/\sqrt{\lambda})}\|v\|_{L^2}^2. 
 \end{align*}
 Using Hadamard's three line lemma, we deduce that 
 $$ \| (-\Delta+\lambda)^{-1/2}|V|^{1/2} v\|_{L^2}^2 =|\phi(1/2)| \le C\|V\|_{\cK(2/\sqrt{\lambda})}\|v\|_{L^2}^2.$$
 As a consequence, for any Schwartz function $w$, 
 $$\langle |V|^{1/2}(-\Delta+\lambda)^{-1/2}w,v\rangle \le C\|V\|_{\cK(2/\sqrt{\lambda})}^{1/2}\|v\|_{L^2}\|w\|_{L^2}.$$
 Taking the supremum over all $v$ shows that 
 $$\| |V|^{1/2}(-\Delta+\lambda)^{-1/2} w\|_{L^2}  \le C\|V\|_{\cK(2/\sqrt{\lambda})}^{1/2}\|w\|_{L^2}.$$
 Choosing $w=(-\Delta+\lambda)^{1/2}u$, we obtain the result. 
\end{proof}

Finally, we recall the useful bound of Rodnianski and Schlag \cite[Lemma 2.5]{RodSch-04} which is proved for $M=\R^3$ and $\epsilon=+\ii$ but carries easily to the more general following situation:

\begin{lemma}
 Let $M$ a three-dimensional Riemannian manifold and $V$ a Kato class potential on $M$. Then, for all $n\ge1$ and all $\epsilon>0$ we have 
\begin{multline}
\label{eq:rodschl}
\sup_{x_0,x_{n+1}\in M\times M}\int_{M_\epsilon(x_0,x_{n+1})}\frac{\prod_{j=1}^n |V(x_j)|}{\prod_{j=0}^n d_g(x_j,x_{j+1})}\sum_{\ell=0}^n d_g(x_\ell,x_{\ell+1}) \,dv_g(x_1)\cdots\,dv_g(x_n) \\
\le(n+1)\|V\|_{\cK(\epsilon)}^n,
\end{multline}
where $M_\epsilon(x_0,x_{n+1}):=\{(x_1,\ldots,x_n)\in M^n,\ \forall j=0,\ldots,n,\ d_g(x_j,x_{j+1})<\epsilon\}$ and 
$$\|V\|_{\cK(\epsilon)}:=\sup_{x\in M}\int_{d_g(x,y)<\epsilon}\frac{|V(y)|}{d_g(x,y)}\,dv_g(y).$$
\end{lemma}

\subsection{The case of bounded domains or compact manifolds}

If we replace $\R^N$ by an open set $\Omega\subset\R^N$ or a compact Riemannian manifold $M$, one can still define the Kato class by replacing $\sup_{x\in\R^N}$ by $\sup_{x\in\Omega}$ or $\sup_{x\in M}$, and in the case of a manifold by additionally replacing each $|x-y|$ by the Riemannian distance $d_g(x,y)$. In both cases, the same strategy as above shows that $V$ is infinitesimally form bounded with respect to $-\Delta_\Omega^D$ or $-\Delta_g$, using in both cases the fact that the Green's functions $(-\Delta_\Omega^D+\lambda)^{-1}(x,y)$ or $(-\Delta_g+\lambda)^{-1}(x,y)$ can be controlled by the one on $\R^N$. This last point can be proved by the maximum principle in the case of a Euclidean domain, or using the estimates of Section \ref{sec:parametrix} on a Riemannian manifold. Indeed, it can be shown that the function $\gamma_{\lambda,\epsilon}$ satisfies 
$$|\gamma_{\lambda,\epsilon}(x,y)|\le C e^{-\sqrt{\lambda}d_g(x,y)/4}$$
for $\epsilon^2\lambda\ge A'$, defining 
$$B=\sup_{x,y\in M}e^{\sqrt{\lambda}d_g(x,y)/4}|\gamma_{\lambda,\epsilon}(x,y)|$$
and estimating $B$ using the integral equation satisfied by $\gamma_{\lambda,\epsilon}$. The smallness comes from the estimate
$$\left|\int_M \gamma_{\lambda,\epsilon}(x,z)R_{\lambda,\epsilon}(z,y)\,dv_g(z)\right|\le C B e^{-\sqrt{\lambda}d_g(x,y)/4}\int_M \frac{e^{-\sqrt{\lambda}d_g(z,y)/4}}{\epsilon^2d_g(z,y)}\,dv_g(z)\le Ce^{-\sqrt{\lambda}d_g(x,y)/4}\frac{B}{\epsilon^2\lambda}.$$

The infinitesimal form boundedness also proves that $-\Delta+V$ either on an open set of finite measure in $\R^N$ or on a compact manifold has discrete spectrum, since for $\lambda>0$ large enough we have 
$$-\Delta+V+\lambda+1=\epsilon(-\Delta)+V+\lambda+(1-\epsilon)(-\Delta)+1\ge (1-\epsilon)(-\Delta+1)$$
in the sense of quadratic forms, hence the operator 
$$(-\Delta+1)^{1/2}(-\Delta+V+\lambda+1)^{-1/2}$$
is bounded. Since the operator $(-\Delta+1)^{-1/2}$ is compact, this shows that $(-\Delta+V+\lambda+1)^{-1/2}$ is compact and hence $-\Delta+V$ has discrete spectrum (and is bounded below).


\section{Details in the proof of Theorem \ref{thm:stieljes-3D}}\label{sec:tauberdetails}

In this appendix we keep track of the various constants in the proof of Theorem \ref{thm:stieljes-3D}. It will be convenient to abbreviate
$$
A_0  :=\|B_0\|_{L^\ii}\Lambda^{3/2} + |B_1(0)| + C_{B_2}(t_0+\Lambda) + C_0\Lambda.
$$
\emph{Step 1.} We show that $A$ satisfies the following pointwise bound for all $t\ge 0$,
\begin{equation}
\label{eq:quantabound}
|A(t)|\le C\left(\|B_0\|_{L^\ii} t^{3/2} + (C_{B_2}+C_0)t+ A_0 \right).
\end{equation}
The starting point of the proof is the bound
\begin{equation}
\label{eq:quantproof0}
|A(t)|\le \|B_0\|_{L^\ii}t^{3/2}+\max\{|B_1(0)|,|B_1(t)|\}+C_{B_2}(t_0+t)
\end{equation}
for all $t\ge 0$, which follows from the fact that $B_1$ is nondecreasing. The claimed inequality \eqref{eq:quantabound} therefore follows if we can show that for all $t\ge 0$,
\begin{equation}\label{eq:quantb1t}
|B_1(t)|\le C\left(\|B_0\|_{L^\ii} t^{3/2} + (C_{B_2}+C_0) t+ A_0 \right).
\end{equation}
Since $B_1$ is nondecreasing, we have for all $\lambda>0$
\begin{equation}
\label{eq:quantproof1}
\int_0^\ii \frac{B_1(t)-B_1(0)}{(t+\lambda)^3}\,dt\ge\int_\lambda^\ii \frac{B_1(t)-B_1(0)}{(t+\lambda)^3}\,dt \ge \frac{B_1(\lambda)-B_1(0)}{8\lambda^2}.
\end{equation}
On the other hand, using
$$\int_0^\ii \frac{|B_0(t)|t^{3/2}}{(t+\lambda)^3}\,dt\le C\|B_0\|_{L^\ii}\lambda^{-1/2}
\qquad\text{and}\qquad
\int_0^\ii \frac{|B_2(t)|}{(t+\lambda)^3}\,dt \le C C_{B_2}(t_0\lambda^{-2}+\lambda^{-1}),$$
we obtain because of the assumption on $A$ that for all $\lambda\ge\Lambda$,
\begin{align*}
\int_0^\ii \frac{B_1(t)-B_1(0)}{(t+\lambda)^3}\,dt & = \int_0^\ii \frac{A(t)-B_0(t)t^{3/2}-B_2(t)-B_1(0)}{(t+\lambda)^3}\,dt \\
& \le C\left(|B_1(0)|\lambda^{-2}+\|B_0\|_{L^\ii}\lambda^{-1/2}+C_{B_2}(t_0\lambda^{-2}+\lambda^{-1})+C_0\frac{e^{-\epsilon_0\sqrt{\lambda}}}{\lambda}\right).
\end{align*}
Combining this with \eqref{eq:quantproof1} we obtain for all $\lambda\ge\Lambda$,
\begin{equation*}
B_1(\lambda) - B_1(0) \le  C\left(|B_1(0)|+\|B_0\|_{L^\ii}\lambda^{3/2}+C_{B_2}(t_0+\lambda)+C_0\lambda e^{-\epsilon_0\sqrt{\lambda}} \right).
\end{equation*}
This, together with the bound 
$$
|B_1(\lambda)|\le \left(B_1(\lambda)-B_1(0)\right) + |B_1(0)| \le \left(B_1(\max\{\lambda,\Lambda\})-B_1(0)\right) + |B_1(0)| \,,
$$
yields \eqref{eq:quantb1t}. This completes the proof of \eqref{eq:quantabound}.

\emph{Step 2.} Next, we derive bounds on $\int_0^\infty e^{-su} A(u^2)\,du$ both for $s=\sigma>0$ and for $|s|\le\epsilon_0/2$.
First, using \eqref{eq:quantabound} we obtain for all $\sigma>0$
\begin{align}
\label{eq:quantalapl}
\left| \int_0^\ii e^{-\sigma u}A(u^2)\,du \right| & \le  C\left(\|B_0\|_{L^\ii}\sigma^{-4}+(C_{B_2}+C_0)\sigma^{-3} + A_0 \sigma^{-1} \right).
\end{align}
In order to derive a bound for $|s|\le\epsilon_0/2$, we first note that because of \eqref{eq:quantabound} we have, for all $\lambda> 0$,
\begin{align*}
\int_0^\ii\frac{|A(t)|}{(t+\lambda)^3}\,dt & \le C\left(\|B_0\|_{L^\ii}\lambda^{-1/2}
+(C_{B_2}+C_0)\lambda^{-1} + A_0 \lambda^{-2} \right).
\end{align*}
Combining this estimate with \eqref{eq:kappabound} we deduce that for all $|s|\le\epsilon_0/2$,
\begin{equation*}
\int_0^\Lambda \lambda^{3/2} |\kappa(s\sqrt\lambda)|\int_0^\ii \frac{|A(t)|}{(t+\lambda)^3}\,dt\,d\lambda 
\le C \sqrt\Lambda \, e^{\epsilon_0\sqrt{\Lambda}/2} A_0 \,.
\end{equation*}
On the other hand, using the assumption on $A$ and again \eqref{eq:kappabound} we deduce that for all $|s|\le\epsilon_0/2$,
$$\int_\Lambda^\ii \lambda^{3/2} |\kappa(s\sqrt\lambda)| \left|
\int_0^\ii \frac{A(t)}{(t+\lambda)^3}\,dt \right| d\lambda \le C C_0 \int_\Lambda^\ii \lambda^{1/2} e^{-\epsilon_0\sqrt{\lambda}}e^{|s|\sqrt{\lambda}}\,d\lambda\le C\, C_0\, \epsilon_0^{-3} .$$
To summarize, we have for all $|s|\le\epsilon_0/2$,
\begin{align}\label{eq:quantalapl2}
\left| \int_0^\ii e^{-su}A(u^2)\,du \right| & \le C \left(  \sqrt\Lambda\, e^{\epsilon_0\sqrt{\Lambda}/2} A_0 + C_0\,\epsilon_0^{-3}\right).
\end{align}

\emph{Step 3.} Next, we derive bounds on $\int_0^\infty e^{-su} \,dg_0(u)$ both for $s=\sigma>0$ and for $|s|\le\epsilon_0/2$. 
We first consider the expression with $g$ instead of $g_0$. Using $|A(0)|\le |B_1(0)|+C_{B_2}t_0$, we obtain immediately from \eqref{eq:quantalapl} and \eqref{eq:quantalapl2} that for all $\sigma>0$
\begin{equation*}
\left| \int_0^\ii e^{-\sigma u} dg(u) \right| \le C\Big(
\|B_0\|_{L^\ii}\sigma^{-3}+(C_{B_2}+C_0)\sigma^{-2} + A_0  \Big)
\end{equation*}
and for all $|s|\le\epsilon_0/2$
\begin{equation*}
\left| \int_0^\ii e^{-su} dg(u) \right| \le C \left(
\epsilon_0 \sqrt\Lambda\, e^{\epsilon_0\sqrt{\Lambda}/2} A_0  + C_0\,\epsilon_0^{-2} + |B_1(0)|+C_{B_2}t_0 \right).
\end{equation*}
This implies that for all $\sigma>0$
\begin{equation}\label{eq:quantg0lapl}
\left| \int_0^\ii e^{-\sigma u} dg_0(u) \right| \le Ce^{-\sigma u_0}\Big(
\|B_0\|_{L^\ii}\sigma^{-3}+(C_{B_2}+C_0)\sigma^{-2} + A_0  \Big)
\end{equation} 
and for all $|s|\le\epsilon_0/2$
\begin{equation}\label{eq:quantg0lapl2}
\left| \int_0^\ii e^{-su} dg_0(u) \right| \le C e^{\epsilon_0 u_0/2}\left(
\epsilon_0\sqrt\Lambda\, e^{\epsilon_0\sqrt{\Lambda}/2} A_0 + C_0\,\epsilon_0^{-2} + |B_1(0)|+C_{B_2}t_0 \right).
\end{equation}

\emph{Step 4.} Now we prove bounds on $f(0)$ and $f'(s)$ for $|s|\le\epsilon_0/2$. According to \eqref{eq:quantf}, $f$ is the second antiderivative of $s\mapsto \int_0^\ii e^{-su} dg_0(u)$, normalized to vanish at infinity. To estimate the antiderivative for $|s|\le\epsilon_0/2$, we choose a path of integration which goes from $s$ to $\epsilon_0/2$ and from $\epsilon_0/2$ to $+\infty$ in straight lines. To estimate it for $\sigma>0$, we simply integrate from $\sigma$ to $+\infty$ in a straight line. Using the bounds from the previous step we find easily that for all $\sigma>0$,
\begin{equation*}
|f'(\sigma)|\le C e^{-\sigma u_0} \Big(
\|B_0\|_{L^\ii}\sigma^{-2}+(C_{B_2}+C_0)\sigma^{-1}
+ u_0^{-1} A_0 \Big)
\end{equation*}
and for all $|s|\le\epsilon_0/2$,
\begin{align*}
|f'(s)| & \le C\epsilon_0 e^{\epsilon_0 u_0/2}\left(
\epsilon_0\sqrt\Lambda\, e^{\epsilon_0\sqrt{\Lambda}/2}A_0  +C_0\,\epsilon_0^{-2} + |B_1(0)|+C_{B_2}t_0 \right) \\
& \quad + C e^{-\epsilon_0u_0/2} \left(
\|B_0\|_{L^\ii} \epsilon_0^{-2}+(C_{B_2}+C_0) \epsilon_0^{-1} + u_0^{-1} A_0  \right).
\end{align*}
Integrating once more, we obtain from these two bounds that
\begin{align}\label{eq:quanthlapl}
|f(0)|& \le C\epsilon_0^2 e^{\epsilon_0 u_0/2}\left(
\epsilon_0 \sqrt\Lambda\, e^{\epsilon_0\sqrt{\Lambda}/2} A_0  + C_0\,\epsilon_0^{-2} + |B_1(0)|+C_{B_2}t_0 \right) \notag \\
& \quad + C e^{-\epsilon_0u_0/2} \left(
\|B_0\|_{L^\ii} \epsilon_0^{-1}+(C_{B_2}+C_0)\left( 1+ \ln_+\tfrac{1}{\epsilon_0 u_0}\right) +  (1+u_0\epsilon_0)u_0^{-2} A_0\right).
\end{align}
This is the desired bound on $a=f(0)$. Moreover, if we define $\phi(s)=(f(s)-a)/s$, then for all $|s|\le\epsilon_0/2$,
\begin{align}\label{eq:quanthlapl2}
|\phi(s)|\le\sup_{|s'|\le\epsilon_0/2}|f'(s')|
& \le C\epsilon_0 e^{\epsilon_0 u_0/2}\left(
\epsilon_0\sqrt\Lambda\, e^{\epsilon_0\sqrt{\Lambda}/2}A_0  +C_0\,\epsilon_0^{-2} + |B_1(0)|+C_{B_2}t_0 \right) \notag \\
& \quad + C e^{-\epsilon_0u_0/2} \left(
\|B_0\|_{L^\ii} \epsilon_0^{-2}+(C_{B_2}+C_0) \epsilon_0^{-1} + u_0^{-1} A_0  \right).
\end{align}
Thus, setting $T=\epsilon_0/2$, we obtain
\begin{multline}\label{eq:bound-a-phi}
|a| + \int_{-T}^T |\phi(it)|\,dt 
\le C\epsilon_0^2 e^{\epsilon_0 u_0/2}\left(
\epsilon_0\sqrt\Lambda\, e^{\epsilon_0\sqrt{\Lambda}/2}A_0 + C_0\,\epsilon_0^{-2} + |B_1(0)|+C_{B_2}t_0 \right) \\
\quad + C e^{-\epsilon_0u_0/2} \left(
\|B_0\|_{L^\ii} \epsilon_0^{-1}+(C_{B_2}+C_0)\left( 1+ \ln_+\tfrac{1}{\epsilon_0 u_0}\right) + \left( 1 + \epsilon_0 u_0 \right) u_0^{-2} A_0 \right).
\end{multline}

\emph{Step 5.} We now show that the bound $h(v)-h(u)\ge -c$ holds for all $0\le u\le v\le u+\delta$ with
\begin{equation}\label{eq:bound-c}
c \le C \left( C_{B_0} + \|B_0\|_{L^\infty} \left(\delta + u_0 \right) + C_{B_2} \left( 1 + \ln(1+\delta u_0^{-1}) + t_0 u_0^{-2} \right)\right).
\end{equation}
To prove this, we decompose
$$
g_0(u)=\tilde{B_0}(u)u^3+\tilde{B_1}(u)+\tilde{B_2}(u)
$$
where
$$\tilde{B_0}(u)=\begin{cases}
B_0((u-u_0)^2) & \text{if}\ u\ge u_0,\\
B_0(0) & \text{if}\ u<u_0,
\end{cases},\quad
\tilde{B_1}(u)=\begin{cases}
B_1((u-u_0)^2) - B_1(0) & \text{if}\ u\ge u_0,\\
0 & \text{if}\ u<u_0,
\end{cases},$$
$$\tilde{B_2}(u)=\begin{cases}
((u-u_0)^3-u^3)B_0((u-u_0)^2)+B_2((u-u_0)^2)-B_2(0) & \text{if}\ u\ge u_0,\\
- u^3 B_0(0) & \text{if}\ u<u_0,
\end{cases}.  
$$
We have, for all $0\le u\le v$,
\begin{align}\label{eq:almostmonoproof}
h(v)-h(u) &= \int_{(u,v]}\frac{dg_0(w)}{w^2} =\int_{(u,v]}\frac{d(\tilde{B_0}(w)w^3)}{w^2}+\int_{(u,v]}\frac{d\tilde{B_1}(w)}{w^2}+\int_{(u,v]}\frac{d\tilde{B_2}(w)}{w^2}
\end{align}
and we discuss separately the three  terms on the right side. The function $\tilde{B_1}$ is non-decreasing since $B_1$ is non-decreasing and therefore the second term is nonnegative.

For the first term in \eqref{eq:almostmonoproof} we have if $u\le v\le u+\delta$, using the right-continuity of $w\mapsto w^3\tilde{B_0}(w^2)$,
\begin{align*}
\int_{(u,v]}\frac{d(\tilde{B_0}(w)w^3)}{w^2} &= \tilde{B_0}(v)v-\tilde{B_0}(u)u+2\int_u^v\tilde{B_0}(w)\,dw \\
&= \tilde{B_0}(v)(v-u)+(\tilde{B_0}(v)-\tilde{B_0}(u))u+2\int_u^v\tilde{B_0}(w)\,dw \\
&\ge -3 \|B_0\|_{L^\ii}\delta +(\tilde{B_0}(v)-\tilde{B_0}(u))u.
\end{align*}
If $u\le u_0$, we simply bound
$$
(\tilde{B_0}(v)-\tilde{B_0}(u))u \ge - 2\|B_0\|_{L^\infty} u_0 \,.
$$
If $u> u_0$ we use the assumption on $B_0$ to deduce that 
$$\tilde{B_0}(v)-\tilde{B_0}(u)\ge -C_{B_0}/(u-u_0),$$
and hence
\begin{align*}
(\tilde{B_0}(v)-\tilde{B_0}(u))u &= (\tilde{B_0}(v)-\tilde{B_0}(u))(u-u_0)+(\tilde{B_0}(v)-\tilde{B_0}(u))u_0 \\
& \ge -C_{B_0} -2\|B_0\|_{L^\ii}u_0.
\end{align*}
To summarize, for all $0\le u \le v\le u+\delta$,
$$
\int_{(u,v]}\frac{d(\tilde{B_0}(w)w^3)}{w^2} \ge -\left( C_{B_0} + \|B_0\|_{L^\infty} (3\delta + 2u_0)\right).
$$
For the third term in \eqref{eq:almostmonoproof} we first note that if $u<u_0$, then
$$
\int_{(u,\min\{v,u_0\}]} \frac{d \tilde{B_2}(w)}{w^2} = - 3 B_0(0) \left( \min\{v,u_0\}-u \right).
$$
If $v\le u+\delta$, the right side is $\ge - 3 \|B_0\|_{L^\infty}\delta$. Thus, in the following we will bound the integral under the assumption that $u\ge u_0$. Since $\tilde{B_2}$ is right-continuous, we write
\begin{align*}
\int_{(u,v]}\frac{d \tilde{B_2}(w)}{w^2} &= \frac{\tilde{B_2}(v)}{v^2}-\frac{\tilde{B_2}(u)}{u^2}+2\int_u^v\frac{\tilde{B_2}(w)}{w^3}\,dw \,.
\end{align*}
We use $|u^3-(u-u_0)^{3}|\le Cu_0(u_0^2+u^2)\le Cu_0u^2$ to get
$$
|\tilde{B_2}(u)|\le C \left( u_0 u^2 \|B_0\|_{L^\infty} + C_{B_2}(t_0 + (u-u_0)^2) \right).
$$
Therefore, since $u_0\le u$,
\begin{align*}
\frac{|\tilde{B_2}(u)|}{u^2} & \le C \left( u_0 \|B_0\|_{L^\infty} + \frac{C_{B_2} t_0}{u^2} + \frac{C_{B_2}(u-u_0)^2}{u^2} \right) \le C \left( u_0 \|B_0\|_{L^\infty} + \frac{C_{B_2} t_0}{u_0^2} + C_{B_2}  \right)
\end{align*}
and, using
$$
\int_u^{u+\delta} \frac{(w-u_0)^2}{w^3}\,dw \le C \left( 1+ \ln\left(1+\frac{\delta}{u_0}\right)\right),
$$
also
\begin{align*}
\int_u^v\frac{|\tilde{B_2}(w)|}{w^3}\,dw & \le C \left( u_0 \|B_0\|_{L^\infty} \ln\frac{v}{u} + C_{B_2}t_0 \left( \frac{1}{u^2} - \frac{1}{v^2}\right) + C_{B_2} \left( 1+ \ln\left(1+\frac{\delta}{u_0}\right) \right) \right) \\
& \le C \left( u_0 \|B_0\|_{L^\infty} \ln\left( 1+ \delta u_0^{-1}\right) + C_{B_2} \left( 1+ \ln\left( 1+ \delta u_0^{-1}\right)+ t_0 u_0^{-2}\right) \right).
\end{align*}
To summarize, for all $0\le u\le v\le u+\delta$,
\begin{align*}
\int_{(u,v]}\frac{d \tilde{B_2}(w)}{w^2} &
\ge -C\left( \|B_0\|_{L^\infty} \left( u_0 + \delta \right)  +  C_{B_2} \left( 1+ \ln(1+\delta u_0^{-1}) + t_0 u_0^{-2} \right) \right).
\end{align*}
Combining our bounds on the three terms in \eqref{eq:almostmonoproof} gives \eqref{eq:bound-c}.

\emph{Step 6.} 
Applying Theorem \ref{coro:tauber-approx-increasing} to $h$ with $T=\epsilon_0/2$ and using the bounds \eqref{eq:bound-a-phi} and \eqref{eq:bound-c}, we find that for all $u\ge 0$ we have 
\begin{align*}
|h(u)|& \le C \epsilon_0^2 e^{\epsilon_0 u_0/2}\left(
\epsilon_0\sqrt\Lambda\, e^{\epsilon_0\sqrt{\Lambda}/2}A_0 + C_0\,\epsilon_0^{-2} + |B_1(0)|+C_{B_2}t_0 \right) \\
& \quad + C e^{-\epsilon_0u_0/2} \left(
\|B_0\|_{L^\ii} \epsilon_0^{-1}+(C_{B_2}+C_0)\left( 1+ \ln_+\tfrac{1}{\epsilon_0 u_0}\right) + \left( 1 + \epsilon_0 u_0 \right) u_0^{-2} A_0\right) \\
& \quad + C \left( 1 + \delta^{-1}\epsilon_0^{-1}\right) \left( C_{B_0} + \|B_0\|_{L^\infty} \left(\delta + u_0 \right) + C_{B_2} \left( 1 + \ln(1+\delta u_0^{-1}) + t_0 u_0^{-2} \right)\right) \\
& =: \Theta \,.
\end{align*}
Thus, if $x\ge u_0$,
$$
|A((x-u_0)^2) - A(0)| = |g_0(x)| = \left| x^2 h(x) -2 \int_0^x uh(u)\,du \right| \le 2 \Theta x^2
$$
and thus, if $u\ge 0$,
$$
|A(u)| \le |A(0)| +  2\Theta (u_0 + \sqrt t)^2 \le |B_1(0)| + C_{B_2} t_0 + 4 \Theta (u_0^2 + t) \,.
$$
Note that with the choice $u_0=\epsilon_0^{-1}$ we get
\begin{align*}
\Theta & \le C \epsilon_0^2 \left(
\left(1+\epsilon_0\sqrt\Lambda\, e^{\epsilon_0\sqrt{\Lambda}/2}\right)A_0 + C_0\,\epsilon_0^{-2} \right) \\
& \quad + C \left( 1 + \delta^{-1}\epsilon_0^{-1}\right) \left( C_{B_0} + \|B_0\|_{L^\infty} \left(\delta + \epsilon_0^{-1} \right) + C_{B_2} \left( 1 + \ln(1+\delta \epsilon_0) + t_0 \epsilon_0^2 \right)\right)
\end{align*}
and therefore, finally,
\begin{multline}\label{eq:quantaboundfinal}
|A(u)|  \le C \epsilon_0^2 \left(
\left(1+\epsilon_0\sqrt\Lambda\, e^{\epsilon_0\sqrt{\Lambda}/2}\right)A_0 + C_0\,\epsilon_0^{-2} \right) \left( t+ \epsilon_0^{-2} \right) \\
 \quad + C \left( 1 + \delta^{-1}\epsilon_0^{-1}\right) \left( C_{B_0} + \|B_0\|_{L^\infty} \left(\delta + \epsilon_0^{-1} \right) + C_{B_2} \left( 1 + \ln(1+\delta \epsilon_0) + t_0 \epsilon_0^2 \right)\right) \left( t+\epsilon_0^{-2}\right).
\end{multline}
This proves the bound claimed in Theorem \ref{thm:stieljes-3D}.


\end{document}